\documentclass[12pt,twoside]{article}
\usepackage{tikz}
\usepackage{times}
\usepackage{geometry}
\usepackage{amsthm}
\usepackage{amsmath}
\usepackage{amssymb}
\usepackage{esint}
\usepackage{microtype}
\usepackage[unicode=true,pdfusetitle,
 bookmarks=true,bookmarksnumbered=false,bookmarksopen=false,
 breaklinks=false,pdfborder={0 0 1},backref=false,colorlinks=false]
 {hyperref}
\usepackage{breakurl}
\usepackage{marginnote}

\makeatletter

\newcommand{\noun}[1]{\textsc{#1}}

\numberwithin{equation}{section}
\numberwithin{figure}{section}
\theoremstyle{plain}
\newtheorem{thm}{\protect\theoremname}
  \theoremstyle{plain}
  \newtheorem{prop}[thm]{\protect\propositionname}
  \theoremstyle{plain}
  \newtheorem{lem}[thm]{\protect\lemmaname}
  \theoremstyle{remark}
  \newtheorem{rem}[thm]{\protect\remarkname}
  \theoremstyle{plain}
  \newtheorem*{fact*}{\protect\factname}
  \theoremstyle{plain}
  \newtheorem{cor}[thm]{\protect\corollaryname}

\makeatother

  \providecommand{\corollaryname}{Corollary}
  \providecommand{\factname}{Fact}
  \providecommand{\lemmaname}{Lemma}
  \providecommand{\propositionname}{Proposition}
  \providecommand{\remarkname}{Remark}
\providecommand{\theoremname}{Theorem}

\DeclareMathSizes{10}{10}{8.5}{6.7}   
\DeclareMathSizes{11}{11}{9.2}{7.3}   
\DeclareMathSizes{12}{12}{10}{8}  

\begin{document}

\title{A renormalization group method by harmonic extensions and the classical
dipole gas}

\author{\Large{Hao Shen} \footnote{Email: pkushenhao@gmail.com}
\\
\normalsize{Princeton University} }

\maketitle
\begin{abstract}
In this paper we develop a new renormalization group method, which
is based on conditional expectations and harmonic extensions, to study
functional integrals of small perturbations of Gaussian
fields. In this new method one integrates Gaussian fields inside domains
at all scales conditioning on the fields outside these domains, and
by variation principle solves local elliptic problems. It does not
rely on an a priori decomposition of the Gaussian covariance. We apply
this method to the model of classical dipole gas on the lattice, and
show that the scaling limit of the generating function with smooth
test functions is the generating function of the renormalized Gaussian
free field.
\end{abstract}

\section{Introduction}

In this paper we develop a renormalization group (RG) method to estimate
functional integrals, based on ideas of conditional expectations and
harmonic extensions. We demonstrate this method with the model of
classical dipole gas, which has always been considered as a simple
model to start with for this type of problems. For the classical dipole
model, earlier important works are \cite{frohlich_correlation_1978,frohlich_statistical_1981}.
The renormalization group approach to this model originated from the
works by Gawedzki and Kupiainen \cite{gawedzki_rigorous_1980,gawedzki_block_1983},
based on Kadanoff spin blockings. 
A different method that uses the idea of decomposition of the covariance of
the Gaussian field was initiated from \cite{brydges_grad_1990},
and was simplified and pedagogically presented in the lecture notes
\cite{brydges_lectures_2007}, see also \cite{dimock_infinite_2009}.
The latter method has achieved several important applications in other
problems such as the two-dimensional Coulomb gas model \cite{Dimock_SineGordon,falco_kosterlitz_2012,falco_Critical_2013},
$\phi^{4}$-type field theories \cite{brydges_short_1995,brydges_nonGaussian_1998,brydges_critical_2003,Brydges_2014-8}
and self-avoiding walks \cite{Brydges_functional_2009,brydges_renormalization_2010,bauerschmidt_structural} (see also the recent works \cite{Brydges_2014-1,Brydges_2014-2,Brydges_2014-3,Brydges_2014-4,Brydges_2014-5,Brydges_2014-6}).
The $\phi^{4}$ field theory problems are also studied in the p-adics setting 
by \cite{abdesselam} which yields some strong consequences.

Our method is different from the above two methods, and may be as
well regarded as a variation of the method by Brydges et al. Their
decomposition of covariance scheme, which was also used by other people
such as \cite{gallavotti_1985}, could be implemented by Fourier analysis.
In \cite{brydges_finite_2004}, a decomposition of Gaussian covariance
with every piece of covariance having finite range was constructed
using elliptic partial differential equation techniques, which also
depends to some extent on Fourier analysis, and this decomposition
is the foundation of the simplified version of their RG method (see
also \cite{brydges_finite_2006,bauerschmidt_simple_2012,adams_2013}
for alternative constructions of such decompositions). We do not perform
such a decomposition of covariance. Instead we directly take harmonic
extensions as our basic scheme and use the Poisson kernel to smooth
the Gaussian field. We do not need Fourier analysis; instead, real
space decay rates of Poisson kernels and (derivatives of) Green's
functions are essential. Some complexities in \cite{brydges_finite_2004}
such as proof of elliptic regularity theorem on lattice are avoided.
Many elements of this method such as the polymer expansions and so
on are very close to the method by Brydges et al, especially to \cite{brydges_lectures_2007},
while we also have many new features, such as simpler norms and regulators.
We keep notations as close as possible to \cite{brydges_lectures_2007}
for convenience of the readers who are familiar with \cite{brydges_lectures_2007}. 

Very roughly speaking, our method is aimed to study functional integrals
of the form
\[
Z=\mathbb{E}\big[e^{V(\phi)}\big]
\]
where $\phi$ is a Gaussian field and $\mathbb{E}$ is an expectation
with respect to a Gaussian measure. Similarly with \cite{brydges_lectures_2007}
we will rewrite the integrand into a local expansion over subsets $X$
of an explicit part and an implicit remainder. 
For instance in the
model considered in this paper, the above quantity $Z$ will be rewritten 
into an expression of roughly the following form (more  precisely, see Proposition~\ref{prop:firstprop})
\[
\mathbb{E}\Big[\sum_{X}
	e^{\sigma\sum_{x\notin X}\left(\partial\phi(x)\right)^{2}}
	K(X,\phi)\Big]
\]
where $K(X,\phi)$ depends only on $\phi(x)$ with $x$ in (a neighborhood of) $X$. We will then take
a family of conditional expectations at a sequence of scales parametrized
by integer $j$ - so our approach is a multi-scale analysis.
To give a quick glance of the main idea, 
at a scale $j$ we will have expressions, which up to several subtleties look
 as follows:
\[
\mathbb{E}\Big[
	\sum_{Y}e^{\sigma_{j}\sum_{x\notin Y}
	\mathbb{E} \,\left[\,\partial\phi(x)\,|\,B_{x}^{c}\,\right]^{2}}
	\,\mathbb{E}
	\big[K_{j}^{\prime}(Y,\phi)|Y^{c}\big]
	\Big] \;.
\]
The actual expressions will be slightly different and more complicated and we refer to 
Section~\ref{subsec:outline} for the exact expressions,
but at this stage we point out that some conditional expectations have appeared 
inside the overall expectation.
Indeed, 
for any function of the
field $F(\phi)$, the notation
$\mathbb{E}\left[F(\phi)|X^{c}\right]$  means integrating \emph{all} the variables $\{\phi(x):x\in X\}$
with $\{\phi(x):x\in X^{c}\}$ \emph{fixed} ($X^c$ is the complement of $X$). 
Also, $B_{x}$ is a block
containing $x$, and
 $\sigma_{j}$ is the
most important dynamical parameter (which corresponds to renormalization
of the dielectric constant in the dipole model). 
This idea of conditional expectation is close to
Frohlich and Spencer's work on Kosterlitz-Thouless transition \cite{frohlich_kosterlitzPRL_1981,frohlich_kosterlitzCMP_1981}
where the authors take inside an expectation conditional integrations,
each over \emph{all} variables $\{\phi(x):x\in\Omega\}$ where $\Omega$
is a bounded region around a charge density $\rho$ with diameter
$\sim2^{j}$
at a scale $j$. 

Such conditional expectations can be carried out by minimizing the
quadratic form in the Gaussian measure with conditioning variables
fixed. Since the Gaussian is associated to a Laplacian these minimizers
are harmonic extensions of $\phi$ from $X^{c}$ into $X$. These
harmonic extensions result in smoother dependence of the integrand
of the expectation on the field. Some elliptic PDE methods along with
random walk estimates will be used. We remark that this variational
viewpoint also shows up in Balaban's RG method (see for instance \cite{balaban_1983}
or Section 2.2 - 2.3 of \cite{dimock_renormalization-1}). 

~\\

\noun{Acknowledgement: }I would like to thank my advisor Weinan E
who has been supporting my work on renormalization group methods over
years and giving me many good suggestions. I am very grateful for
the kind hospitality of David Brydges during my visits to University
of British Columbia, as well as a lot of encouragement and helpful
conversations by him from the stage of shaping the basic ideas of this paper to that of the final proofreading. I also appreciate mumerous discussions with
Stefan Adams, Arnulf Jentzen, and especially Roland Bauerschmidt.

\section{Outline of the paper} \label{sec:Outline}

\subsection{Settings, notations and conventions\label{sub:Conventions-about-notations}}

Let $\mathbb{Z}^{d}$ be the $d$ dimensional lattice with $d\geq2$.
Denote the sets of lattice directions as $\mathcal{E}_{+}=\{e_{1},...,e_{d}\}$
and $\mathcal{E}_{-}=\{-e_{1},...,-e_{d}\}$ where $e_{k}=(0,\cdots,1,\cdots,0)$
with only the k-th element being $1$. Let $\mathcal{E}=\mathcal{E}_{+}\cup\mathcal{E}_{-}$.
For $e\in\mathcal{E}$, $\partial_{e}f(x)=f(x+e)-f(x)$ is the lattice
derivative. For $x,y\in\mathbb{Z}^{d}$, we say that $(x,y)$ is a
nearest neighbor pair and write $x\sim y$ if there exists an $e\in\mathcal{E}$
such that $x=y+e$. Denote $E(\mathbb{Z}^{d})$ to be the set of all
nearest neighbor pairs of $\mathbb{Z}^{d}$. For $X\subset\mathbb{Z}^{d}$,
we define $E(X):=\{(x,y)\in E(\mathbb{Z}^{d}):x,y\in X\}$.

Let $L$ be a positive odd integer, and $N\in\mathbb{N}$. Let 
\[
\Lambda=[-L^{N}/2,L^{N}/2]^{d}\cap\mathbb{Z}^{d} \;,
\]
and we will consider functions on $\Lambda$ with periodic boundary
condition. In other words we view $\Lambda$ as a torus by identifying
the boundary points of $\Lambda$ in the usual way.

For $x,y\in\Lambda$, define $d(x,y)$ to be the length of a shortest
path of nearest neighbor sites in the torus $\Lambda$ connecting
$x$ and $y$. Also define $\partial X$ to be the ``outer boundary'':
$\partial X=\{x\in\Lambda:d(x,X)=1\}$. Write $X^{c}$ to be
the complement of $X$. 

For a function $\phi$ on $\mathbb{Z}^{d}$, when it doesn't cause
confusions, we write for short 
\[
\sum_{X}(\partial\phi)^{2}=\sum_{x\in X}(\partial\phi(x))^{2}:=\frac{1}{2}\sum_{x\in X}\sum_{e\in\mathcal{E}}(\partial_{e}\phi(x))^{2}
\]
and similarly for other such type of summations. If $\mathbb{E}$
is the expectation over $\phi$, we will use a short-hand notation
for conditional expectation
\[
\mathbb{E}\left[-\big|X\right]:=\mathbb{E}\left[-\big|\{\phi(x)\big|x\in X\}\right] \;,
\]
namely, the expectation with $\phi|_{X}$ fixed.

Unless we specify otherwise, Poisson kernels and Green's functions
will be associated with the operator $-\Delta+m^{2}$ where $m$ is a
small mass regularization. For any set $X$, $P_{X}$ or $P_{X}(x,y)$
($x\in X$, $y\in\partial X$) is the Poisson kernel for $X$. If
$x\notin X$ then $P_{X}f(x)=f(x)$ is always understood. In other
words, $P_{X}f$ is the harmonic extension of $f$ from $X^{c}$ into
$X$ with $f\big|_{X^{c}}$ unchanged.

\subsection{The dipole gas model and the scaling limit\label{sub:Definition-of-model}}

Let $\mu$ be the Gaussian measure on the space of functions $\{\phi(x):x\in\Lambda\}$
with mean zero and covariance $C_{m}=(-\Delta+m^{2})^{-1}$ where
$m>0$. In other words, $\phi$ is the Gaussian free field on the
$\Lambda$ with covariance $C_{m}$. Let $\mathbb{E}$ be the expectation
over $\phi$. Then the classical dipole gas model is defined by the
following measure:
\[
\nu(\phi)=e^{zW(\phi)}\mu(\phi)
\]
where 
\[
W(\phi):=\sum_{x\in\Lambda}\sum_{\substack{e\in\mathcal{E}}
}\cos\left(\sqrt{\beta}\partial_{e}\phi(x)\right) \;.
\]

Such a measure is obtained by a definition of the model via
the great canonical ensemble followed by a Sine-Gordon transformation,
for instance, see \cite{brydges_grad_1990}.

We would like to study the problem of scaling limit. More precisely,
let $\tilde{\Lambda}:=[-\frac{1}{2},\frac{1}{2}]^{d}\subset\mathbb{R}^{d}$.
Given a mean zero function $\tilde{f}\in C^{\infty}(\tilde{\Lambda})$,
$\int_{\tilde{\Lambda}} \tilde f=0$ with periodic boundary condition, we
study the (real) generating function 
\begin{equation} \label{eq:scalinglimit-2}
Z_{N}(f):=
\lim_{m\rightarrow0}
\frac{
	\mathbb{E} \big[e^{\sum_{x\in\Lambda}f(x)\phi(x)}e^{zW(\phi)}\big]
}{
	\mathbb{E}\big[e^{zW(\phi)}\big]}
\end{equation}
where 
\[
f(x)=f_{N}(x):=L^{-(d+2)N/2}\tilde{f}(L^{-N}x)  \;.
\]
The main question is the scaling limit of $Z_{N}(f)$ as $N\rightarrow\infty$.

\subsection{Some preparative steps before RG\label{sub:Tuning}}

As the start of our strategy to study this problem, we perform an
a priori tuning of the Gaussian measure. 
This tuning anticipates the fact that the best Gaussian approximation  to $\nu$
is not the Gaussian measure currently defined on $\phi$.
For any $X\subseteq \Lambda$ define
\begin{equation} \label{eq:def_V}
V(X,\phi):=\frac{1}{4}
	\sum_{\substack{x\in X,e\in\mathcal{E}}
}\left(\partial_{e}\phi(x)\right)^{2} \;.
\end{equation}
The tuning is to split part of the quadratic form of the Gaussian
measure into the integrand, so that the resulting Gaussian field has
covariance $[\epsilon(-\Delta+m^{2})]^{-1}$, with the associated
expectation called $\mathbb{E}^{\epsilon}$: 
\begin{equation} \label{eq:tuned}
Z_{N}(f)  =
\lim_{m\rightarrow0}
\frac{\mathbb{E}^{\epsilon}\big[e^{\sum_{x\in\Lambda}f(x)\phi(x)}e^{(\epsilon-1)V(\Lambda,\phi)+zW(\Lambda,\phi)}\big]}
{\mathbb{E}^{\epsilon}
	\big[e^{(\epsilon-1)V(\Lambda,\phi)+zW(\Lambda,\phi)}\big]} \;.
\end{equation}
Note that normalization factors caused by re-definition of Gaussian:
\[
\mathbb{E}^{\epsilon}\left[\exp\left((\epsilon-1)V(\Lambda,\phi)\right)\right]
\]
appear in both numerator and denominator and are thus cancelled. 

We would like to make the expectation (and thus the RG maps which
we will define later) independent of $\epsilon$. So we rescale $\phi\rightarrow\phi/\sqrt{\epsilon}$
and let $\sigma=\epsilon^{-1}-1$ and obtain 
\begin{equation} \label{eq:ZprimebyZ2primes}
Z_{N}(f)  =\lim_{m\rightarrow0}
\frac{
	\mathbb{E} \big[
		e^{\sum_{x\in\Lambda}f(x)\phi(x)/\sqrt{\epsilon}} 
		\cdot e^{-\sigma V(\phi)+zW(\sqrt{1+\sigma}\phi)}
	\big]
}{
	\mathbb{E}\big[e^{-\sigma V(\phi)+zW(\sqrt{1+\sigma}\phi)}\big]} \;.
\end{equation}

We also shift the Gaussian field to get rid of the linear term $\sum f\phi/\sqrt{\epsilon}$. Write
$-\Delta_{m}=-\Delta+m^{2}$ and make a translation 
$\phi\rightarrow\phi+\xi_m$
where $\xi_m=(-\sqrt{\epsilon}\Delta_{m})^{-1}f$ in the numerator in
(\ref{eq:ZprimebyZ2primes}). 
Since the function $\xi_m$ appears frequently below,
we will simply write $\xi=\xi_m$ without explicitly referring to its dependence on $m$.
Then, one has
\begin{equation} \label{eq:generating_good}
Z_{N}(f)=\lim_{m\rightarrow0}e^{\frac{1}{2}\sum_{x\in\Lambda}f(x)(-\epsilon\Delta_{m})^{-1}f(x)}Z_{N}^{\prime}(\xi)\big/Z_{N}^{\prime}(0)
\end{equation}
where 
\begin{equation} \label{eq:Zprime}
Z_{N}^{\prime}(\xi)=\mathbb{E}\left[e^{-\sigma V(\Lambda,\phi+\xi)+zW((\phi+\xi)/\sqrt{\epsilon})}\right] \;.
\end{equation}

Let $-\tilde{\Delta}_{m}=-\tilde{\Delta}+m^{2}$, where $\tilde{\Delta}$
is the Laplacian acting on the functions on $\tilde\Lambda$, 
and $\tilde{C}_{m}:=(-\tilde{\Delta}_{m})^{-1}$
and $\tilde{\xi}_m:=(-\sqrt{\epsilon}\tilde{\Delta}_{m})^{-1}\tilde{f}$.
We can verify that 
\[
L^{2N}\tilde C_{L^{N}m}(L^{-N}x)=C_{m}(x)
\qquad
\mbox{and} 
\qquad
L^{-\frac{d-2}{2}N}\tilde \xi_{L^N m}(L^{-N}x)=\xi_m (x) \;.
\] 
Let $q<\frac{d}{d-1}$ and define
\[
R:=\sup_{m>0}\max \Big(\Vert\tilde{C}_{m}\Vert_{L^{q}},\Vert\partial\tilde{C}_{m}\Vert_{L^{q}} \Big)
\]
Note that $R<\infty$ since the worst local singularity is $O(|x|^{1-d})$ 
which is $L^q$ integrable for any $q<\frac{d}{d-1}$.
We will assume that $\Vert\tilde{f}\Vert_{L^{p}}\leq h/R$ ($p>d$),
for a constant $h$ to be specified later, so that for $\alpha=0,1$
\begin{equation}
\left\Vert \partial^{\alpha}\xi\right\Vert _{L^{\infty}}\leq hL^{-(\frac{d-2}{2}+\alpha)N}\label{eq:smallness_xi}
\end{equation}
by Young's inequality. 

Before the RG steps, we write both $Z_{N}^{\prime}(\xi)$ and $Z_{N}^{\prime}(0)$
into a form of ``polymer expansion''. For any set $X\subseteq\Lambda$,
write
\[
W(X,\phi):=\sum_{x\in X}\sum_{\substack{e\in\mathcal{E}}
}\cos\left(\sqrt{\beta}\partial_{e}\phi(x)\right) \;.
\]

\begin{prop} \label{prop:firstprop}
With $W$ defined above and $Z_{N}^{\prime}(\xi)$ given by (\ref{eq:Zprime}), we have 
\begin{equation}
Z_{N}^{\prime}(\xi)=\mathbb{E}\bigg[\sum_{X\subseteq\Lambda}
	I_0(\Lambda\backslash X,\phi+\xi)K_0(X,\phi+\xi)\bigg]\label{eq:Mayer}
\end{equation}
where $I_0(X)=\prod_{x\in X}I_0(\{x\})$ and 
\[
I_0(\{x\},\phi+\xi)=e^{-\frac{1}{4}\sigma\sum_{e\in\mathcal{E}}(\partial_{e}\phi(x)+\partial_{e}\xi(x))^{2}} \;,
\]
\[ 
K_0(X,\phi+\xi)=\prod_{x\in X}e^{-\frac{1}{4}\sigma\sum_{e\in\mathcal{E}}(\partial_{e}\phi(x)+\partial_{e}\xi(x))^{2}}\left(e^{zW\left(\{x\},(\phi+\xi)/\sqrt{\epsilon}\right)}-1\right) \;.
\] 
\end{prop}

The subscript $0$ indicates that we are at the $0$-th RG step, and 
we will write $\sigma_0=\sigma$.
The quantity $Z_{N}^{\prime}(0)$ has the same expansion with
$\xi=0$. 
\begin{proof}
Consider equation (\ref{eq:Zprime}): following Mayer expansion,
\[
\begin{aligned}
Z_{N}^{\prime}(\xi)= & \mathbb{E}\big[
	e^{zW(\Lambda)-\sigma V(\Lambda)}\big]\\
= & \mathbb{E}\Big[
	\prod_{x\in\Lambda}\Big(
		e^{-\sigma V(\{x\})}+ \big(e^{zW(\{x\})}-1\big) e^{-\sigma V(\{x\})}\Big)
	\Big] \;.
\end{aligned}
\]
Expanding the product amounts to associating
a set $X\subseteq \Lambda$ to the second term and the complement $\Lambda\backslash X$ to the first term.
This proves the statement (\ref{eq:Mayer}).
\end{proof}

\subsection{Outline of main ideas} \label{subsec:outline}

Our renormalization group method is based on the idea of rewriting
the expectation into an expectation of an expression involving many
conditional expectations. We will carry out a multiscale analysis;
an RG map will be iterated from one scale to the next one, during
which we will re-arrange the conditional expectations. A basic algebraic
structure and analytical bound will be propagated to every scale.
In order to describe these structures and bounds, we first give some
definitions.

\subsubsection{Basics of polymers} \label{sec:Basics-Polymers}
\begin{enumerate}
\item We call blocks of size $L^{j}$ j-blocks. These are translations of
$\{x\in\mathbb{Z}^{d}:\left|x\right|<\frac{1}{2}(L^{j}-1)\}$ by vectors
in $\left(L^{j}\mathbb{Z}\right)^{d}$. In particular a 0-block is
a single site in $\mathbb{Z}^{d}$. A j-polymer $X$ is a union of
j-blocks. In particular the empty set is also a j-polymer. The number
of lattice sites in $X\subset\mathbb{Z}^{d}$ is denoted by $\left|X\right|$.
The number of j-blocks in a j-polymer $X$ is denoted by $\left|X\right|_{j}$. 
\item $X\subset\mathbb{Z}^{d}$ is said to be connected if for any two points
$x,y\in X$ there exists a path $(x_{i}:i=0,\dots,n)$ with $\left|x_{i+1}-x_{i}\right|_{\infty}=1$
connecting $x$ and $y$. Here, $|x|_{\infty}$ is the maximum of
all coordinates of $x$; note that for instance $\{(0,0),(1,1)\}$
is connected if $d=2$. Connected sets are not empty. Two sets $X,Y$
are said to be strictly disjoint if there is no path from $x$ to
$y$ when $x\in X$ and $y\in Y$; otherwise we say that they touch.
\item For any $X\subset\mathbb{Z}^{d}$ we let $\mathcal{C}(X)$ be the set of connected components of $X$. 
\item For a j-polymer $X$ we have the following notations. $\mathcal{B}_{j}(X)$
is the set of all j-blocks in $X$. $\mathcal{P}_{j}(X)$ is the set
of all j-polymers in $X$. $\mathcal{P}_{j,c}(X)$ is the set of all
connected j-polymers in $X$. We sometimes just write $\mathcal{B}_{j},\mathcal{P}_{j},\mathcal{P}_{j,c}$
and so on when $X=\Lambda$. 
\item Let $X\in\mathcal{P}_{j}$. Define for $j\geq1$
\[
\begin{aligned}
\hat{X} &:=\cup\{B\in\mathcal{B}_{j}:B\mbox{ touches }X\}  \;,\\
X^{+}&:=\cup\{x\in\Lambda:d(x,X)\leq\frac{1}{3}L^{j}\} \;, \\
\ddot{X}&:=\cup\{x\in\Lambda:d(x,X)\leq\frac{1}{6}L^{j}\} \;,\\
\dot{X}&:=\cup\{x\in\Lambda:d(x,X)\leq\frac{1}{12}L^{j}\} \;.
\end{aligned}
\]
 Note that we have $X\subset\dot{X}\subset\ddot{X}\subset X^{+}\subset\hat{X}$.
Only $X,\hat{X}$ belong to $\mathcal{P}_{j}$. 
\item When $j=0$ and $X\in\mathcal{P}_{0}$, we define $\dot{X}=\ddot{X}=X^{+}=\hat{X}=X$,
and the Poisson kernel at scale $0$ is understood as $P_{X^{+}}:=id$.
\end{enumerate}
We also have the following notations for functions of the fields.
\begin{enumerate}
\item Define $\mathcal{N}$ to be the set of functions of $\phi$ and $\xi$. 
Define
$\mathcal{N}(X) \subseteq\mathcal{N}$ to be the set of functions of
$\{\phi(x),\xi(x) \big|x\in X\}$. $\mathcal{N}^{\mathcal{P}_{j}}$ is the
set of maps $K:\mathcal{P}_{j}\rightarrow\mathcal{N}$ such that $K(X)\in\mathcal{N}(\hat{X})$.
We define $\mathcal{N}^{\mathcal{B}_{j}}$, $\mathcal{N}^{\mathcal{P}_{j,c}}$
similarly.
\item For $I\in\mathcal{N}^{\mathcal{B}_{j}}$ we write
\[
I(X)=I^{X}:=\prod_{B\in\mathcal{B}_{j}(X)}I(B)\qquad\mbox{for }X\in\mathcal{P}_{j} \;.
\]
For $K\in\mathcal{N}^{\mathcal{P}_{j}}$ we say that $K$ factorizes
over connected components if
\begin{equation}\label{eq:factorization}
K(X)=\prod_{Y\in\mathcal{C}(X)}K(Y) \;.
\end{equation}
In this case, $K$ is determined by its value 
on connected polymers, so
we can write $K\in\mathcal{N}^{\mathcal{P}_{j,c}}$.
\end{enumerate}

The \emph{basic structure} that we want to propagate to every scale
of the RG iterations is, for $j\geq0$
\begin{equation} \label{eq:basic_structure}
Z_{N}^{\prime}(\xi)=e^{\mathcal{E}_{j}}\,
\mathbb{E}\bigg[\sum_{X\in\mathcal{P}_{j}(\Lambda)}I_{j}(\Lambda\backslash\hat{X},\phi,\xi)K_{j}(X,\phi,\xi)\bigg] \;.
\end{equation}
Here, $e^{\mathcal{E}_{j}}$ is a $\phi,\xi$ independent constant
factor. This constant will be shown to be the same for $Z_{N}^{\prime}(\xi)$
and $Z_{N}^{\prime}(0)$ and thus cancels. $K_{j}(X,\phi,\xi)$ only
depends on the values of $\phi,\xi$ in a small neighborhood of $X$.
Note that there is a ``corridor'' between each $X$ and $\Lambda\backslash\hat{X}$
(namely, the union of $X$ and $\Lambda\backslash \hat X$ 
is not the entire $\Lambda$, and we call this ``missing part" $\hat X \backslash X$ heuristically as a ``corridor'').
These ``corridors'' will be important in our conditional expectation
method.

Furthermore, for $j<N$, the function $I_{j}$ will have a local form in the sense that it
factorizes over $j$-blocks 
$I_{j}(X,\phi,\xi)=\prod_{B\in\mathcal{B}_j(X)}I_{j}(B,\phi,\xi)$
and 
\begin{equation} \label{eq:defI}
I_{j}(B,\phi,\xi)=e^{-\frac{1}{4}\sigma_{j}\sum_{x\in B,e\in\mathcal{E}}(\partial_{e}P_{B^{+}}\phi(x)+\partial_{e}\xi(x))^{2}} \;.
\end{equation}
$I_{j}(B)$ is essentially determined by the dynamical parameter $\sigma_{j}$.
On the other hand, $K_{j}$ will only factorize over ``connected
components of polymer''.

The \emph{basic bounds} that hold on every scale about $K_{j}$, 
whose form will not be explicit, is as follows. For $X$ connected,
\begin{equation}
\sum_{n=0}^{4} \frac{1}{n!}
\left\Vert 
	K_{j}^{(n)}(X,\phi,\xi)
	\right\Vert 
\leq  \left\Vert K\right\Vert _{j}A^{-|X|_{j}}G(\ddot{X},X^{+}) \;.
\end{equation}
Here, $K_j^{(n)}$ is an $n$-th derivative of $K_j$;
the precise definition of it and the norm will be given later.
For any two sets $X\subset Y$, $G(X,Y)$ is a normalized conditional expectation
called the ``regulator''
\begin{equation}
G(X,Y)=\mathbb{E}\left[e^{\frac{\kappa}{2}\sum_{X}(\partial\phi)^{2}}\big|\phi_{Y^{c}}\right]\big/N(X,Y)
\end{equation}
and the normalization factor is 
\begin{equation}
N(X,Y)=\mathbb{E}\left[e^{\frac{\kappa}{2}\sum_{X}(\partial\phi)^{2}}\big|\phi_{Y^{c}}=0\right] \;.
\end{equation}
This form of regulator is different from the one defined in \cite{brydges_lectures_2007};
in particular it is itself a conditional expectation. It will be shown
to have some interesting properties.

Now we outline the steps to go from scale $j$ to scale $j+1$ while
the structure (\ref{eq:basic_structure}) is preserved.

\subsubsection*{1) Extraction and reblocking. }

Reblocking is a procedure which rewrites (\ref{eq:basic_structure})
into an expansion over ``$j+1$ scale polymers''; and we extract
the components that grow too fast under this reblocking. 
\begin{prop}
\label{prop:extra-reblo}
Suppose that $L$ is sufficiently large. If
at the scale $j$ one has 
\begin{equation}
Z_{N}^{\prime}(\xi)=e^{\mathcal{E}_{j}}\,
\mathbb{E}\bigg[\sum_{X\in\mathcal{P}_{j}}I_{j}^{\Lambda\backslash\hat{X}}(\phi,\xi)K_{j}(X,\phi,\xi)\bigg]\label{eq:form_j}
\end{equation}
with $I_{j}\in\mathcal{N}^{\mathcal{B}_j}$ given by (\ref{eq:defI}),
then there exist $\mathcal{E}_{j+1},I_{j+1}\in\mathcal{N}^{\mathcal{B}_{j+1}}$
and $K_{j}^{\sharp}\in\mathcal{N}^{\mathcal{P}_{j+1,c}}$
(namely $K_{j}^{\sharp}$ factorizes
over connected components in the sense of \eqref{eq:factorization}), 
so that the following expansion at the scale $j+1$ holds 
\begin{equation} \label{eq:next-scale1}
Z_{N}^{\prime}(\xi) =  
	e^{\mathcal{E}_{j+1}}\,
	\mathbb{E}
	\bigg[
		\sum_{U\in\mathcal{P}_{j+1}}I{}_{j+1}^{\Lambda\backslash\hat{U}}(\phi,\xi) \, K_{j}^{\sharp}(U,\phi,\xi)
	\bigg]
\end{equation}
where $\mathcal{E}_{j+1}$ is a constant independent of $\phi,\xi$,
and for every $D\in\mathcal{B}_{j+1}$, 
\[
I_{j+1}(D)=e^{-\frac{1}{4}\sigma_{j+1}\sum_{x\in D,e\in\mathcal{E}}\left(\partial_{e}P_{D^{+}}\phi(x)+\partial_{e}\xi(x)\right)^{2}}
\]
for some constant $\sigma_{j+1}$.
\end{prop}
We will prove this Lemma in Section \ref{sec:The-renormalization-group}.

\subsubsection*{2) Conditional expectation. }

This step is the main difference between this new method and \cite{brydges_lectures_2007}.
First of all,
we observe that
in \eqref{eq:next-scale1},
the sets $\Lambda\backslash\hat U$ and $U$
do not touch. In other words,
there exists a corridor $\hat{U}\backslash U$ around the set $U$
where $K_j^\sharp$ evaluates on,
and this corridor has width
$L^{j+1}$. We then take conditional expectation
and thus re-write the expectation in \eqref{eq:next-scale1} as follows:
\begin{equation} \label{eq:outlinecond}
\mathbb{E}\bigg[
	\sum_{U\in\mathcal{P}_{j+1}}
	I_{j+1}^{\Lambda\backslash\hat{U}}(\phi,\xi) \,\mathbb{E}\Big[K_{j}^{\sharp}(U,\phi,\xi)\big|(U^{+})^{c}\Big]
	\bigg]
\end{equation}
where $U\subset U^{+}\subset\hat{U}$. For notation conventions, see
subsection \ref{sub:Conventions-about-notations}. 
In order to obtain \eqref{eq:outlinecond},
one switches the expectation and the sum in \eqref{eq:next-scale1},
then take the conditional expectation right inside 
the expectation. Since 
$I_{j+1}^{\Lambda\backslash\hat{U}}$ only
depends on the values of $\phi$ being fixed,
the conditional expectation can be taken only on the $K_j^\sharp$ factor.
One then switches back the expectation and the sum.

This followed by factoring out $\phi,\xi$ independent constant
gives $K_{j+1}$ and we are back to the form (\ref{eq:basic_structure})
with all $j$ replaced by $j+1$. In case $U=\Lambda$, we just integrate
(unconditionally): $\mathbb{E}\big[K_{j}^{\sharp}(\Lambda,\phi)\big]$,
but to streamline expressions we still write (\ref{eq:outlinecond})
keeping in mind the special treatment for the $U=\Lambda$ term.

\begin{rem}
Our discussion below will frequently involve 
Laplacian operators acting on functions on a set $U$
with zero Dirichlet boundary condition on $\partial U$,
so we simply refer to them as {\it Dirichlet Laplacian} for $U$.
Similarly, for the Green's function of the Laplacian on $\partial U$
with zero Dirichlet boundary condition on $\partial U$,
we simply call it {\it Dirichlet Green's function} for $U$.
Finally, if $\zeta$ is a Gaussian field on $U$ with Dirichlet Green's function for $U$
as its covariance, then we simply say that $\zeta$ is 
the {\it Dirichlet Gaussian field} on $U$.
\end{rem}

We point out two important facts about the conditional expectation
step. The first one is that we can write the Gaussian field $\phi$
into a sum of two decoupled parts. Let $P_{U}$ be the Poisson kernel
for $U$ and recall our convention that $P_{U}\phi(x)=\phi(x)$ for
$x\notin U$ as in subsection \ref{sub:Conventions-about-notations}. 
\begin{prop}
Let $U\subset V\subset\Lambda$. 
Define $\zeta$ via $\phi(x)=P_{U}\phi(x)+\zeta(x)$.
Then the quadratic form 
\begin{equation}
-\sum_{x\in V}\phi(x)\Delta\phi(x)
= -\sum_{x\in U}\zeta(x)\Delta_{U,m}^{D}\zeta(x)
  -  \sum_{x\in V}P_{U}\phi(x)\Delta_{m}P_{U}\phi(x)
\end{equation}
where $-\Delta_{U,m}^{D}=-\Delta_{U}^{D}+m^{2}$ and $\Delta_{U}^{D}$
is the Dirichlet Laplacian for $U$, $m\geq0$.
\end{prop}
Notice that $x\in U$ does not contribute to the last summation since
$\Delta_{m}P_{U}\phi(x)=0$ in $U$. By this proposition, taking expectation
of a function $K(\phi)$ conditioned on $\{\phi(x)\big|x\in U^{c}\}$
is simply integrating out a Gaussian field $\zeta$:
\begin{equation}
\mathbb{E}\left[K(\phi,\xi)\big|U^{c}\right]=\mathbb{E}_{\zeta}\left[K(P_{U}\phi+\zeta,\xi)\right]\label{eq:simply_int_zeta}
\end{equation}
where the covariance of $\zeta$ is the $C_{U}^{D}$ - the Dirichlet
Green's function for $U$. In particular, we observe that $I_{j}$
defined in (\ref{eq:defI}) has an alternative representation
\begin{equation}
I_{j}(B,\phi,\xi)=e^{-\frac{1}{4}\sigma_{j}\sum_{x\in B,e\in\mathcal{E}}\mathbb{E}\left[\partial_{e}\phi(x)+\partial_{e}\xi(x)\big|(B^{+})^{c}\right]^{2}} \;.
\end{equation}
It is conceptually helpful to keep in mind that we are just re-arranging
the following structure (comparing with \eqref{eq:Mayer})
\begin{equation}
\mathbb{E}\Big[
	\sum_{X\in\mathcal{P}_{j}}
	e^{-\frac{1}{4}\sigma_{j}\sum_{x\notin\hat{X},e\in\mathcal{E}}\mathbb{E}\left[\partial_{e}\phi(x)+\partial_{e}\xi(x)\big|(B^{+})^{c}\right]^{2}}
	\mathbb{E}\big[\cdots\big|(X^{+})^{c}\big]\Big]
\end{equation}
namely an outmost (unconditional) expectation of a simple combination
of many conditional expectations.
\begin{rem}
\label{rem:wellposedness}
In the paper, $P_{U}\phi$ will always be
well-defined: by Prop 1.11 of \cite{kumagai_random_2010}, if the
probability that the random walk starting from any point in $U$ exits
$U$ in finite time is $1$, then the harmonic extension exists and
is unique. Domains $U\subsetneq\Lambda$ will always satisfy this
condition because the random walk hits any point in $\Lambda$ in
finite time with probability one.
\end{rem}
The next fact is as follows:
\begin{prop} \label{prop:scaling}
Let $d\geq2$, $x\in X\subset U\subset\Lambda$. If $d(x,\partial X)\geq cL^{j}$,
then 
\begin{equation}
|(\partial_{x}P_{X})C_{U}^{D}(\partial_{x}P_{X})^{\star}(x,x)|\leq O(1)L^{-dj}
\end{equation}
where $O(1)$ depends on $c$, and $C_{U}^{D}$ is the Dirichlet Green's
function for $U$. 
\end{prop}
For the proof, see Proposition~\ref{prop:covest}. This result gives the scaling
for the covariance of $\partial P_{X}\zeta$ where $P_{X}$ is a Poisson
kernel obtained from the previous RG step. We take a heuristic test
to see the necessity of this proposition: setting $\xi=0$, for $X\subset U$,
if we perform an expectation conditioned on $\{\phi(x)\big|x\in X^{c}\}$,
followed by another expectation conditioned on $\{\phi(x)\big|x\in U^{c}\}$,
by (\ref{eq:simply_int_zeta}),
\begin{equation}
\mathbb{E}_{\zeta_{U}}\mathbb{E}_{\zeta_{X}}
\Big[K\Big(P_{X}(P_{U}\phi+\zeta_{U})+\zeta_{X}\Big)\Big]
=\mathbb{E}_{\zeta_{U}}\mathbb{E}_{\zeta_{X}}
\Big[K\big(P_{U}\phi+P_{X}\zeta_{U}+\zeta_{X}\big)\Big] \;,
\end{equation}
then we need this proposition to deal with $P_{X}\zeta_{U}$ when
integrating over $\zeta_{U}$.

Proofs of the above two results are in the following sections.

\subsubsection*{Linearization and stable manifold theorem}

We have just outlined a single RG map
\[
(\sigma_{j},\sigma_{j+1},E_{j+1},K_{j})\rightarrow K_{j+1} \;.
\]
We will show smoothness of this map in Section \ref{sec:Smoothness-of-RG}.
Note that two issues have not been discussed: 1) choice of $\sigma_{j+1},E_{j+1}$,
which should be a function of $(\sigma_{j},K_{j})$, so that the RG
map becomes $(\sigma_{j},K_{j})\rightarrow(\sigma_{j+1},K_{j+1})$
(notice that we will not regard $E_{j+1}$ as dynamical parameter and
we will factorize it out); 2) choice of $\sigma$ in the a priori tuning
step. We will outline how to treat these two issues now.

Clearly $(\sigma,K)=(0,0)$ is a fixed point of the RG map. In Section
\ref{sec:Linearized-RG} we show that the linearization of the map
$(\sigma_{j},\sigma_{j+1},E_{j+1},K_{j})\rightarrow K_{j+1}$ around
$(0,0,0,0)$ has a form $\mathcal{L}=\mathcal{L}_{1}+\mathcal{L}_{2}+\mathcal{L}_{3}$
where $\mathcal{L}_{1}$ captures the ``large polymers'' contributions
to $K_{j+1}$, and $\mathcal{L}_{2}$ involves the remainder of second
order Taylor expansion of conditionally expected $K_{j}$ on {}``small
polymers'', both of which will be shown contractive with arbitrarily
small norm by suitable choices of constants $L$ and $A$ introduced
above. Furthurmore, $\mathcal{L}_{3}(D)$ will roughly have a form 
\begin{equation}
 L^{d}E_{j+1}+\sigma_{j+1}\sum_{x\in D}(\partial P_{D^{+}}\phi(x))^{2}-\sigma_{j}\big(\sum_{x\in D}(\partial P_{D^{+}}\phi(x))^{2}+\delta E_{j}\big)+Tay
\end{equation}
where $Tay$ is the second order Taylor expansion of conditionally
expected $K_{j}$ on small polymers, which consists of constant and
quadratic terms, and $D$ is a $j+1$ block. Now it is easy to see
that there is a way to choose $E_{j+1}$ and $\sigma_{j+1}$ so that
$\mathcal{L}_{3}$ is almost $0$, up to a localization procedure
for {}``$Tay$''. For proofs see Section \ref{sec:Linearized-RG}.

Once we have shown a way to choose the constants $\sigma_{j+1},E_{j+1}$
to ensure contractivity of the above linear map, a stable manifold
theorem can be applied to prove that there exists a suitable tuning
of $\sigma$ so that
\begin{equation} \label{eq:sigmaKgoto0}
\left|\sigma_{j}\right|\lesssim2^{-j}\qquad\left\Vert K_{j}\right\Vert _{j}\lesssim2^{-j} \;.
\end{equation}

\subsubsection*{Main result: the scaling limit}
\begin{thm}
For any $p>d$ there exists constants $M>0$ and $z_{0}>0$ so that:
for all $\Vert\tilde{f}\Vert_{L^{p}}\leq M$ and all $\left|z\right|\leq z_{0}$
there exists a constant $\epsilon$ depending on $z$ and
\begin{equation}
\lim_{N\rightarrow\infty} Z_{N}(f)
= \exp\left(\frac{1}{2}\int_{\tilde{\Lambda}}\tilde{f}(x)(-\epsilon\tilde{\Delta})^{-1}\tilde{f}(x)d^{d}x\right)\label{eq:main_thm}
\end{equation}
where $\tilde{\Delta}$ is the Laplacian in continuum,
and $Z_{N}(f)$ is defined in (\ref{eq:generating_good}).
\end{thm}
The main ingredient of the proof is that 
with $j=N-1$,
by eq. (\ref{eq:basic_structure}) and (\ref{eq:sigmaKgoto0}),
one can bound $Z_{N}^{\prime}(\xi)$ essentially by
\begin{equation}
e^{\mathcal{E}_{N-1}}\sum_{X\in\mathcal{P}_{N-1}}(1+2^{-N})^{\Lambda\backslash\hat{X}}2^{-N} 
\end{equation}
Bounding the number of terms by $ $$2^{L^{d}}$ we see that it is
almost $e^{\mathcal{E}_{N-1}}$ as $N$ becomes large. The constant
$e^{\mathcal{E}_{N-1}}$ will be the same for $Z_{N}^{\prime}(\xi)$
and $Z_{N}^{\prime}(0)$. So only the exponential factor in equation
(\ref{eq:generating_good}) survives in the $N\rightarrow\infty$
limit and it goes to the right hand side of (\ref{eq:main_thm}).
The details are given in Section \ref{sec:Proof-of-scaling}. We remark that
the assumption on $\tilde{f}$, which makes $f$ smooth at the scale
$N$ is for simplicity of the demonstration of the method.

\section{The renormalization group steps\label{sec:The-renormalization-group}}

\subsection{Some additional definitions} \label{sub:Definitions}
\begin{enumerate}
\item A j-polymer $X$ is called a small set or small polymer if it is connected
and $\left|X\right|_{j}\leq2^{d}$. Otherwise it is called large. We
write by $\mathcal{S}_{j}(X)$ the set of all small j-polymers in
$X$. 
\item Define $\hat{\mathcal{S}_{j}}$ to be the set of pairs $(B,X)$ so
that $X\in\mathcal{S}_{j}$ and $B\in\mathcal{B}_{j}(X)$.
\item We also introduce a notation $Y\in_{X}\mathcal{P}_{j}$ which means
$Y\in\mathcal{P}_{j}$ and that if $X=\emptyset$ then $Y=\emptyset$.
\item Let $X\in\mathcal{P}_{j}$. Define its closure $\bar{X}\in\mathcal{P}_{j+1}$
to be the smallest (j+1)-polymer that contains $X$. 
\item 
We define a notation $\chi_{\mathcal{A}}^{j}$ where $\mathcal{A}$
is a set of polymers: $\chi_{\mathcal{A}}^{j}=1$ if any two polymers
in $\mathcal{A}$ are strictly disjoint as $j$-polymers and $\chi_{\mathcal{A}}^{j}=0$
otherwise. Also, if $\mathcal{A}$ is a set of polymers, we write
$X_{\mathcal{A}}$ to be the union of all elements of $\mathcal{A}$.
\end{enumerate}

\subsection{Renormalization group steps} \label{subsec:RG-steps}

Now we focus on a single RG map from scale $j$ to $j+1$. For simpler
notations we omit the subscript $j$ and objects at scale $j+1$ will
be labelled by a prime, e.g. $K^{\prime}$, $\mathcal{P}^{\prime}$.
The guidance principle will be that for all kinds of $I$'s below,
$I-1$ and their difference $\delta I$ and $K$ will be small, so
their products will be higher order small quantities. These remarks
will make more sense after we discuss the linearization of the smooth
RG map in Section~\ref{sec:Linearized-RG}.

\subsubsection*{Extraction and Reblocking}

We start to prove Proposition~\ref{prop:extra-reblo}. 
Before the proof,
we describe here the main ideas in the strategies we use below,
and a reader may find helpful to read the proof along with these descriptions.
The way to construct $I^{\prime},K^{\sharp}$ is certainly
not unique. However, our construction (see \eqref{eq:Ksharp} below)
must have the foresight that 
$K^{\sharp}$ will be smooth in its arguments (which will be shown
in Section~\ref{sec:Smoothness-of-RG}), with respect to 
certain norms defined in Section~\ref{sec:Norms}.
Due to the nature of these norms, the proof of this smoothness in
Section~\ref{sec:Smoothness-of-RG}
will reply on some separation properties of different factors 
appeared in the $K^\sharp$ finally constructed in \eqref{eq:Ksharp}.
Ensuring these separation properties complicates the proof.

The proof of Proposition~\ref{prop:extra-reblo} then consists of two steps.
The first step is called an extraction step,
in which we extract $\delta I(B)$ from $I(B)$,
see the third line of \eqref{eq:extraction},
 resulting in 
a new quantity $\tilde I(B)$ defined as \eqref{eq:Itilde}.
The extracted quantities $\delta I(B)$
will show 
up as factors multiplying with $K$ in \eqref{eq:after-expand}. 

The second step is called a reblocking step. 
In this step, summations over various sets in \eqref{eq:after-expand}
will eventually become one single sum over 
next scale polymers
$U\in\mathcal P'$
as in \eqref{eq:combine-everything}.
During this reblocking, some $\tilde I$ factors will also become
factors multiplying with $K$ (see \eqref{eq:Ksharp}).

There are two subtleties which one has to take care and thus complicates the proof. 
The first subtlety is that $\delta I(B)$ and $\tilde I(B)$
involve a Poisson kernel for $(\bar B)^+$ which is a set of length size $O(L^{j+1})$.
When these factors  $\delta I$ and $\tilde I$ 
show up as factors multiplying with $K(X)$ as discussed above, the factor $K(X)$ actually only has a corridor $\hat X\backslash X$ of width $L^j$
(formed from the previous RG step),
so the sets $(\bar B)^+$ may intersect with $X$. 
This intersection would be disastrous when we estimate the norm
of the product of these $\delta I$, $\tilde I$ and $K$ factors.
Therefore, in the proof we actually only extract $\delta I(B)$
for those $B$ far enough from $X$, that is, outside the set $\left\langle X\right\rangle$ defined below. 
Inside $\left\langle X\right\rangle$, we do different extractions
as in the second line of \eqref{eq:extraction}, so that the $L^j$ width corridor of $K(X)$
is sufficient to ensure separation.

The other subtlety is 
 that according
to the conclusion of Proposition~\ref{prop:extra-reblo},
 one has to ensure existence of a corridor around $U$ in \eqref{eq:combine-everything}. 
This is not ensured in the ``naive reblocking"  \eqref{eq:reblocking} below (though
in  \eqref{eq:reblocking} we do obtain one single sum over next scale polymers $V\in\mathcal P'$).
Therefore, as an intermediate 
step between extraction and reblocking,
we will perform another expansion by
$\tilde{I}=(\tilde{I}-e^{E^{\prime}})+e^{E^{\prime}}$
right after \eqref{eq:Itilde-and-eE},
and arrange such that
some of the $\tilde{I}-e^{E^{\prime}}$ also show up as multiplying factors in \eqref{eq:Ksharp},
and the other $\tilde{I}$ will be separated away by a corridor (between $\Lambda\backslash \hat U$ and $U$ in the last line of the proof).

\begin{rem}
We also have a remark on notations. 
The hats in the notation for a set of pairs such as the $\hat{\mathcal{S}_j}$
defined above in Subsection \ref{sub:Definitions} 
and the $\hat{\mathcal Y}$ in the following proof
are simply symbols, which have nothing to do 
with the hat  operation $\hat{\,}$  on a single polymer defined in Subsection \ref{sec:Basics-Polymers}.
\end{rem}

\begin{proof} [Proof of Proposition~\ref{prop:extra-reblo}] 
Define $\tilde{I}\in\mathcal{N}^{\mathcal{B}_{j}}$
as
\begin{equation} \label{eq:Itilde}
\tilde{I}(B)=e^{E^{\prime}-\frac{1}{4}\sigma^{\prime}\sum_{x\in B,e\in\mathcal{E}}\left(\partial_{e}P_{(\bar{B})^{+}}\phi(x)+\partial_{e}\xi(x)\right)^{2}}
\end{equation}
where $E^{\prime}$ and $\sigma^{\prime}$ will be chosen later. 
Note that the above quantity $\tilde I(B)$ differs from 
the quantity $I(B)$ defined in (\ref{eq:defI}) by the new constants $E',\sigma'$ and 
the Poisson kernel $P_{B^+}$ is replaced by the 
 Poisson kernel $P_{(\bar B)^+}$.
For a j-polymer $X$, denote
\[
\left\langle X\right\rangle :=\cup\{B\in\mathcal{B}_{j}:(\bar{B})^{+}\cap\hat{X}\neq\emptyset\}
\]
where the $+$ operation is on the scale $j+1$ and the hat is on
the scale $j$. Then we let
\begin{equation}\label{eq:extraction}
\begin{cases}
1(B)=(1-e^{E^{\prime}})+e^{E^{\prime}} & \mbox{if}\; B\subseteq\hat{X}\backslash X\\
I(B)=(I(B)-e^{E^{\prime}})+e^{E^{\prime}} & \mbox{if}\; B\subseteq\left\langle X\right\rangle \backslash\hat{X}\\
I(B)=\delta I(B)+\tilde{I}(B) & \mbox{if}\; B\subseteq\left\langle X\right\rangle ^{c}\\
K(X)=\sum_{B\in\mathcal{B}(X)}\frac{1}{|X|_{j}}K(B,X) & \mbox{if}\; X\in\mathcal{S}
\end{cases}
\end{equation}
where $\delta I$ is defined implicitly, and $K(B,X):=K(X)$. Insert
these summations into the product factors in (\ref{eq:form_j}), and
expand. We obtain 
\begin{equation} \label{eq:after-expand}
\begin{aligned}
Z_{N}^{\prime}(\xi) 
& = e^{\mathcal{E}}\,\mathbb{E}\bigg[\sum_{X}I^{\Lambda\backslash\hat{X}}1^{\hat{X}\backslash X}\prod_{Y\in\mathcal{C}(X)\backslash\mathcal{S}}K(Y)\prod_{Y\in\mathcal{C}(X)\cap\mathcal{S}}K(Y)\bigg]\\
 & =e^{\mathcal{E}}\,\mathbb{E}\bigg[\sum_{\mathcal{X},\hat{\mathcal{Y}}}\chi_{\mathcal{X}\cup\mathcal{Y}}\sum_{P,Q,Z}(1-e^{E^{\prime}})^{P}(I-e^{E^{\prime}})^{Q}(e^{E^{\prime}})^{(\left\langle X\right\rangle \backslash X)\backslash(P\cup Q)}\\
 & \qquad\cdot
 \delta I^{Z}\tilde{I}^{\left\langle X\right\rangle ^{c}\backslash Z}
 \prod_{Y\in\mathcal{X}}K(Y)\prod_{(B,Y)\in\hat{\mathcal{Y}}}\frac{1}{|Y|_{j}}K(B,Y)\bigg]
\end{aligned}
\end{equation}
where the first summation is over $\mathcal{X}$ which is a family
of connected large polymers, and $\hat{\mathcal{Y}}$ which is a family
of elements in $\hat{\mathcal{S}}$ i.e. $\hat{\mathcal{Y}}=\{(B_{i},Y_{i})\in\hat{\mathcal{S}_{j}}\}_{1\leq i\leq n}$
for some $n\geq0$, and we have defined $\mathcal{Y}:=\{Y_{i}\}_{1\leq i\leq n}$.
In the above equation and in the sequel of this proof,
\[
X:=X_{\mathcal{X}\cup\mathcal{Y}}
\]
and the second summation above is over $P\in\mathcal{P}(\hat{X}\backslash X)$,
$Q\in\mathcal{P}(\left\langle X\right\rangle \backslash\hat{X})$,
and $Z\in\mathcal{P}(\left\langle X\right\rangle ^{c})$. 

Now observe that one can re-arrange the above summations in the following
way:
\begin{equation} \label{eq:reblocking}
\sum_{\mathcal{X},\hat{\mathcal{Y}}}\chi_{\mathcal{X}\cup\mathcal{Y}}\sum_{P,Q,Z}=\sum_{V\in\mathcal{P}^{\prime}}\sum_{(P,Q,Z,\mathcal{X},\hat{\mathcal{Y}})\rightarrow V}
\end{equation}
where the second summation on the right hand side means
\[
\sum_{(P,Q,Z,\mathcal{X},\hat{\mathcal{Y}})\rightarrow V}
:=\sum_{\mathcal{X},\hat{\mathcal{Y}}}\chi_{\mathcal{X}\cup\mathcal{Y}}
	\sum_{\substack{P\in\mathcal{P}(\hat{X}\backslash X)\\Q\in\mathcal{P}(\left\langle X\right\rangle \backslash\hat{X})}}
	\sum_{Z\in\mathcal{P}(\left\langle X\right\rangle ^{c})}1_{\overline{P\cup Q\cup Z\cup\left(\cup_{i=1}^{n}B_{i}\right)\cup X_{\mathcal{X}}}=V} \;.
\]
We would like to write the factors $\tilde{I}$ and $e^{E^{\prime}}$
into parts in $V$ and outside $V$:
\begin{equation} \label{eq:Itilde-and-eE}
\begin{aligned}
\tilde{I}^{\,\left\langle X\right\rangle ^{c}\backslash Z} 
& =\tilde{I}^{\,V^{c}\,\cap\,\left\langle X\right\rangle ^{c}}
\, \tilde{I}^{\,V\,\cap\,(\left\langle X\right\rangle ^{c}\backslash Z)} \;,\\
(e^{E^{\prime}})^{(\left\langle X\right\rangle \backslash X)\backslash(P\cup Q)} & =(e^{E^{\prime}})^{V^{c}\,\cap\,(\left\langle X\right\rangle \backslash X)}\,(e^{E^{\prime}})^{V\,\cap\,(\left\langle X\right\rangle \backslash X)\backslash(P\cup Q)} \;.
\end{aligned}
\end{equation}

Note that $V^{c}\cap\left\langle X\right\rangle ^{c}$ (where some
$\tilde{I}$ live on) could possibly touch $V$, so our next step
is to make a corridor so that such touchings will be avoided. Write
$\tilde{I}=(\tilde{I}-e^{E^{\prime}})+e^{E^{\prime}}$ and expand,
\[
\tilde{I}^{V^{c}\cap\left\langle X\right\rangle ^{c}}=\sum_{W\in\mathcal{P}^{\prime}(V^{c})}(\tilde{I}-e^{E^{\prime}})^{W\cap\left\langle X\right\rangle ^{c}}(e^{E^{\prime}})^{(V^{c}\backslash W)\cap\left\langle X\right\rangle ^{c}} \;.
\]
For each $V$ and $W$, define $U$ to be the smallest union of connected
components of $V\cup W$ that contains $V$:
\[
U=U_{W,V}:=\cap\{T\big|T\in\mathcal{UC}(V\cup W),T\supseteq V\}\in\mathcal{P}^{\prime}
\]
where $\mathcal{UC}(V\cup W)$ is the set of unions of ($j+1$ scale)
connected components of $V\cup W$. Observe that if $L$ is sufficiently
large, one has $\left\langle X\right\rangle \subseteq\hat{V}\subseteq\hat{U}$.
So
\[
\begin{aligned}
\tilde{I}^{\,V^{c}\,\cap\,\left\langle X\right\rangle ^{c}}
=\sum_{W\in\mathcal{P}^{\prime}(V^{c})}
& (\tilde{I}-e^{E^{\prime}})^{W\backslash\hat{U}}
(\tilde{I}-e^{E^{\prime}})^{W\,\cap \,U\,\cap\,\left\langle X\right\rangle ^{c}} \\
& \times
(e^{E^{\prime}})^{(V^{c}\backslash W)\backslash\hat{U}}
(e^{E^{\prime}})^{(V^{c}\backslash W)\,\cap\,\hat{U}\,\cap\,\left\langle X\right\rangle ^{c}} \;.
\end{aligned}
\]
Let $R:=W\backslash U=W\backslash\hat{U}$. Note that one has the
following identities for the sets appearing in the above equation:
$W\cap U=U\backslash V$ and 
\[
(V^{c}\backslash W)\backslash\hat{U}=(\hat{U})^{c}\backslash R \;,
\]
\[
(V^{c}\backslash W)\cap\hat{U}=\hat{U}\backslash U  \;.
\]
The summation over $W$ amounts to a summation over $U$ and $R$:
\begin{equation}
\begin{aligned}
\tilde{I}^{\,V^{c}\,\cap\,\left\langle X\right\rangle ^{c}} 
& =\sum_{U\in_{V}\mathcal{P}^{\prime},U\supseteq V}
\sum_{R\in\mathcal{P}^{\prime}(\Lambda\backslash\hat{U})}
(\tilde{I}-e^{E^{\prime}})^{R}
(\tilde{I}-e^{E^{\prime}})^{(U\backslash V)\,\cap\,\left\langle X\right\rangle ^{c}}\\
&\qquad\qquad\qquad\times
(e^{E^{\prime}})^{(\hat{U})^{c}\backslash R}
(e^{E^{\prime}})^{(\hat{U}\backslash U)\,\cap\,\left\langle X\right\rangle ^{c}}\\
 & =\sum_{U\in_{V}\mathcal{P}^{\prime},U\supseteq V}\tilde{I}^{\Lambda\backslash\hat{U}}(\tilde{I}-e^{E^{\prime}})^{(U\backslash V)\,\cap\,\left\langle X\right\rangle ^{c}}(e^{E^{\prime}})^{(\hat{U}\backslash U)\,\cap\,\left\langle X\right\rangle ^{c}} \;.
\end{aligned}
\label{eq:ItildeVcXc}
\end{equation}

The factor $(e^{E^{\prime}})^{V^{c}\,\cap\,(\left\langle X\right\rangle \backslash X)}$
appearing in (\ref{eq:Itilde-and-eE}) is treated as follows. Since
$\left\langle X\right\rangle \subseteq\hat{U}$ 
\begin{equation}
\begin{aligned}
(e^{E^{\prime}})^{V^{c}\,\cap\,(\left\langle X\right\rangle \backslash X)} & =(e^{E^{\prime}})^{V^{c}\,\cap\,\left\langle X\right\rangle }(e^{-E^{\prime}})^{V^{c}\,\cap\, X}\\
 & =(e^{E^{\prime}})^{(\hat{U}\backslash U)\,\cap\,\left\langle X\right\rangle }(e^{E^{\prime}})^{V^{c}\,\cap\,\left\langle X\right\rangle \,\cap\, U}(e^{-E^{\prime}})^{V^{c}\,\cap\, X} \;.
\end{aligned}
\label{eq:eEtildeVcXX}
\end{equation}
Combine (\ref{eq:after-expand}) - (\ref{eq:eEtildeVcXX}),
\begin{equation}\label{eq:combine-everything}
Z_{N}^{\prime}(\xi)=e^{\mathcal{E}} \,\mathbb{E}\bigg[\sum_{U\in\mathcal{P}^{\prime}}\tilde{I}^{\Lambda\backslash\hat{U}}(e^{E^{\prime}})^{\hat{U}}K^{\sharp}(U)\bigg]
\end{equation}
where for $U\neq\emptyset$
\begin{equation} \label{eq:Ksharp}
\begin{aligned}
K^{\sharp}(U):= & 
\sum_{V\subseteq U,V\neq\emptyset}
\sum_{(P,Q,Z,\mathcal{X},\hat{\mathcal{Y}})\rightarrow V}
\!\!\!\!\!\!\!\!\!\!
(1-e^{E^{\prime}})^{P}
(I-e^{E^{\prime}})^{Q}
\delta I^{Z}
(\tilde{I}-e^{E^{\prime}})^{(U\backslash V)\,\cap\,\left\langle X\right\rangle ^{c}}\\
 & \qquad\times
 \,(e^{E^{\prime}})^{(\left\langle X\right\rangle \backslash X)\,\cap\, U\backslash(P\cup Q)}
 \,(e^{-E^{\prime}})^{U\,\cup\, X}
 \,\tilde{I}^{\,V\,\cap\,(\left\langle X\right\rangle ^{c}\backslash Z)} \\
 & \qquad\times
 \prod_{Y\in\mathcal{X}}K(Y)\prod_{(B,Y)\in\hat{\mathcal{Y}}}\frac{1}{|Y|_{j}}K(B,Y)\;.
\end{aligned}
\end{equation}
Factorizing the constant $e^{E^{\prime}}$ by letting
\[
\mathcal{E}^{\prime}=\mathcal{E}+E^{\prime}|\Lambda|_{j}
\]
\[
I^{\prime}(D)=e^{-L^{d}E^{\prime}}\prod_{B\in\mathcal{B}(D)}\tilde{I}(B)=e^{-\frac{1}{4}\sigma_{j+1}\sum_{x\in D,e\in\mathcal{E}}\left(\partial_{e}P_{D^{+}}\phi(x)+\partial_{e}\xi(x)\right)^{2}}
\]
for $D\in\mathcal{B}^{\prime}$, we obtain
\[
Z_{N}^{\prime}(\xi)=  e^{\mathcal{E}^{\prime}}\mathbb{E}\bigg[\sum_{U\in\mathcal{P}^{\prime}}(I^{\prime})^{\Lambda\backslash\hat{U}}K^{\sharp}(U)\bigg] \;.
\]
This is precisely the statement \eqref{eq:next-scale1}.
\end{proof}

\subsubsection*{Conditional expectation}
\begin{lem}
$K^{\sharp}$ factorizes over $j+1$ scale connected components, namely
\begin{equation}
K^{\sharp}(U)=\prod_{V\in\mathcal{C}_{j+1}(U)}K^{\sharp}(V)
\end{equation}
where $\mathcal{C}_{j+1}(U)$ is the set of connected components of
$U$ as a $j+1$ polymer.\end{lem}
\begin{proof}
Let $V_{1},\dots,V_{|\mathcal{C}(U)|}$ be all the connected components
of $U$. For any $E$ which may stand for $U,Z,P,Q$, elements of
$\mathcal{X}\cup\mathcal{Y}$, one of the $B_{i}$, or $X=X_{\mathcal{X}\cup\mathcal{Y}}$,
let $E^{(p)}=E\backslash\cup_{q\neq p}V_{q}$. It is easy to check
that for $i\neq j$, $E^{(i)}$ and $E^{(j)}$ are strictly disjoint
on scale $j$. Then the lemma is proved by the factorization property
of $I,K$ on scale $j$.
\end{proof}

We are now ready to take the expectation of $K^{\sharp}(V)$ conditioned
on $\phi$ outside $V^{+}$ for each $V\in\mathcal{C}(U)\backslash\{\Lambda\}$,
because $\Lambda\backslash\hat{V}$ and $V^{+}$ do not touch. In the
case $V=\Lambda$, we just take expectation of $K^{\sharp}(V)$ without
conditioning, but write 
$\mathbb{E}\big[K^{\sharp}(\Lambda)\big|(\Lambda^{+})^{c}\big]:=\mathbb{E}\big[K^{\sharp}(\Lambda)\big]$
to shorten the notations. So we obtain the following structure as announced in \eqref{eq:outlinecond}:

\begin{equation} \label{eq:def-KnextU}
\begin{aligned}
Z_{N}^{\prime}(\xi)=e^{\mathcal{E}_{j+1}}  \mathbb{E}\bigg[
	\sum_{\substack{U\in\mathcal{P}_{j+1}}} 
	I_{j+1}^{\Lambda\backslash\hat{U}} 
	K_{j+1}(U)
\bigg] \;, \\
	K_{j+1}(U) := \prod_{V\in\mathcal{C}(U)}\mathbb{E}\left[K_{j}^{\sharp}(V)\big|(V^{+})^{c}\right] \;.
\end{aligned}
\end{equation}

Now we have come back to the basic structure (\ref{eq:basic_structure})
with $j$ replaced by $j+1$. Obviously, $K_{j+1}(U)\in\mathcal{P}_{j+1,c}$.
In Section \ref{sec:Norms} we give precise definitions for norms
and spaces of the $K_{j}$ above, and in section \ref{sec:Smoothness-of-RG}
we prove smoothness of the above map $(\sigma_{j},E_{j+1},\sigma_{j+1},K_{j})\mapsto K_{j+1}$.

\subsection{Properties about conditional expectation}

\subsubsection*{The variation principle}

One of our main ideas is to write the Gaussian field $\phi$ into
a sum of two decoupled parts. This is important for the conditional
expectation.
\begin{fact*}
Given any positive definite quadratic form $Q(v)$ for vector $v$,
if $v=(x,y)$, one can write $Q(v)=Q_{1}(x)+L(x,y)+Q_{2}(y)$ where
$Q_{1,2}$ are positive definite quadratic forms and $L(x,y)$ is
the crossing term. Let $\tilde{x}(y)$ be the minimizer of $Q(v)=Q(x,y)$
with $y$ fixed. Then, one can cancel $L(x,y)$ by shifting $x$ by
$\tilde{x}$:
\begin{equation} \label{eq:quadratic_fact}
Q(v)=Q_{1}(x-\tilde{x})+Q\left((\tilde{x},y)\right) \;.
\end{equation}
\end{fact*}

Before introducing the next proposition, let us recall our convention
that $P_{U}\phi(x)=\phi(x)$ for $x\notin U$ as in subsection \ref{sub:Conventions-about-notations}. 
\begin{prop}
Let $U\subset V\subset\mathbb{Z}^{d}$ be two finite sets. Let $\phi_{U}$
and $\phi_{U^{c}}$ be the restriction of $\phi$ to $U$ and $U^{c}$.
Let $P_{U}$ be the Poisson kernel for $U$ and write $\phi(x)=P_{U}\phi(x)+\zeta(x)$.
Then, 
\begin{equation}
-\sum_{x\in V}\phi(x)\Delta_{m}\phi(x)=-\sum_{x\in U}\zeta(x)\Delta_{U}^{D}\zeta(x)-\sum_{x\in V}P_{U}\phi(x)\Delta_{m}P_{U}\phi(x)
\end{equation}
where $\Delta_{U}^{D}$ is the Dirichlet Laplacian for $U$.
\end{prop}
\begin{proof}
We can apply the Fact (\ref{eq:quadratic_fact}) for $\phi=(\phi_{U},\phi_{U^{c}})$,
and
\[
\begin{aligned}Q(\phi) & =-\sum_{x\in V}\phi(x)\Delta\phi(x)\\
 & =-\sum_{x\in U}\phi_{U}(x)\Delta_{U}^{D}\phi_{U}(x)+L(\phi_{U},\phi_{U^{c}})-\sum_{x\in U^{c}}\phi_{U^{c}}(x)\Delta_{U^{c}}^{D}\phi_{U^{c}}(x)
\end{aligned}
\]
where $L$ is the crossing term, and $\Delta_{U^{c}}^{D}$ is the
Dirichlet Laplacian for $U^{c}$. 
Since the minimizer of $Q(\phi)$ with $\phi_{U^{c}}$ fixed is
the harmonic extension of $\phi$ from $U^c$ into $U$,
and the harmonic field is equal to $P_{U}\phi$, one has
\[
\begin{aligned}
Q(\phi)
& =-\sum_{x\in U}\left(\phi_{U}-P_{U}\phi\right)(x)\Delta_{U}^{D}\left(\phi_{U}-P_{U}\phi\right)(x)-Q\left((P_{U}\phi,\phi_{U^{c}})\right)\\
 & =-\sum_{x\in U}\zeta(x)\Delta_{U}^{D}\zeta(x)-\sum_{x\in V}P_{U}\phi(x)\Delta P_{U}\phi(x) \;.
\end{aligned}
\]
This completes the proof. We remark that 
in the last term, the points $x\in U$ do not actually contribute 
to the sum since $\Delta P_U \phi=0$ in $U$. 
\end{proof}

By this proposition, taking expectation of a function $K(\phi)$ conditioned
on $\{\phi(x)\big|x\in U^{c}\}$ is equivalent to simply integrating
out $\zeta$:
\begin{equation}
\mathbb{E}\left[K(\phi)\big|U^{c}\right]=\mathbb{E}_{\zeta}\left[K(P_{U}\phi+\zeta)\right]
\end{equation}
where the covariance of $\zeta$ is the Dirichlet Green's function
for $U$.

As another important fact, we note that $K(X,\phi,\xi)$
constructed above (see \eqref{eq:def-KnextU}) has
a ``special structure'': it only depends on $\phi,\xi$ via $P_{X^{+}}\phi+\xi$;
in other words there exists a function $\widetilde{K}(X,\psi)$ so that
\begin{equation} \label{eq:depend-via}
K(X,\phi,\xi)=\widetilde{K}(X,P_{X^{+}}\phi+\xi) \;.
\end{equation}
In fact, we have the following lemma.
\begin{lem}
\label{lem:only-dep-harmo}
Let $U\subset \Lambda$ be a given set.
For every $k=1,\dots,m$, let $Y_{k}\subset U$, and $H_{k}(\phi,\xi)$ be a given function
of $\phi$ and $\xi$.
Suppose that there exist functions $\widetilde H_k$ such that
\[
H_{k}(\phi,\xi)=\widetilde{H}_{k}(P_{Y_{k}}\phi+\xi) \;,
\]
namely $H_k$ only depends on $\phi$, $\xi$ via $P_{Y_{k}}\phi+\xi$.
Then, the function $\mathbb{E}\left[\prod_{k}H_{k}(\phi,\xi)\big|U^{c}\right]$
only depends on $\phi,\xi$ via $P_{U}\phi+\xi$. 
\end{lem}
\begin{proof}
We write the expectation conditioned on $\phi\big|_{U^c}$ as expectation over the Dirichlet  Gaussian field  $\zeta$ on $U$, and then exploit the assumption on $H_k$:
\begin{equation} \label{eq:proof-only-dep}
\mathbb{E}\Big[\prod_{k}H_{k}(\phi,\xi)\big|U^{c}\Big]
=\mathbb{E}_{\zeta}\Big[\prod_{k}H_{k}(P_U \phi+\zeta,\xi)\Big]
=\mathbb{E}_{\zeta}\Big[\prod_{k}\widetilde{H}_{k} \Big(P_{Y_{k}}(P_{U}\phi+\zeta)+\xi\Big)\Big] .
\end{equation}
The last quantity 
depends on $\phi,\xi$ via $P_U \phi+\xi$ by 
noting that $P_{Y_{k}}P_{U}=P_{U}$. 
\end{proof}

Note that $K_0(X,\phi,\xi)$ is actually a function
of $\phi+\xi$. By our convention, when $j=0$,
$P_{X^+}$ is understood as the identity operator,
so we do start from functions with this special
structure (\eqref{eq:depend-via}). 
Together with the above lemma and \eqref{eq:Ksharp}, \eqref{eq:def-KnextU},
we see that for every $j\ge 0$, the fact
\eqref{eq:depend-via} holds:

\begin{cor}
Let $K_j(X,\phi,\xi)$ be the functions constructed in  \eqref{eq:def-KnextU}.  Then for every $j\ge 0$, 
there exists a function $\widetilde{K}_j(X)$ such that
$K_j(X,\phi,\xi)=\widetilde{K}_j(X,P_{X^{+}}\phi+\xi) $.
\end{cor}

In the following, it will be helpful to have this 
point of view in mind.

\subsubsection*{The important scaling}

Our main result in this subsection is Proposition~\ref{prop:covest}.
We first collect some general results about harmonic functions on the
lattice. These will include derivative estimates and ``mean value" type bounds.

\begin{lem}
Let $\mathcal B_R$ be 
the discrete ball of radius $R$ centered on the origin, namely
$\mathcal B_R=\{x\in\mathbb Z^d: |x|<R\}$.
There exists a constant $c>0$ such that
the following holds for every $R$ sufficiently large.
\begin{itemize}
\item If $g$ is harmonic in $\mathcal B_R$, then
for every $e\in\mathcal S$,
\begin{equation} \label{eq:diff-est-any}
|\partial_e g(0)| \le c R^{-1} \sup_{x\in\mathcal B}
 	|g(x)| \;.
\end{equation}
\item If $f$ is harmonic and non-negative in $\mathcal B_R$,
then
for every $e\in\mathcal S$,
\begin{equation} \label{eq:diff-est-positive}
|\partial_e f(0)| \le c R^{-1} f(0) \;.
\end{equation}
\end{itemize}
\end{lem}

\begin{proof}
This is \cite[Theorem 6.3.8 of Section 6.3]{lawler_random_2010}.
The statement of that theorem is about harmonic functions
related with general ``$\mathcal P_d$ class" (i.e. symmetric, finite range) random walks.
In particular it is true for harmonic functions associated with standard Laplacian
related with simple random walks. The large $R$ requirement was used to deal with the lattice effect on the boundary of the ball in their proof.
\end{proof}

Note that the constant $c$ in the above lemma does not 
depend on the function $g$ or $f$. In the second statement, non-negativity condition is necessary: the linear function $f(x)=x$ on $[-1,1]$ would violate the bound \eqref{eq:diff-est-positive}.

The next result is a mean value type bound.
For $R>0$ and $a\in\mathbb{Z}^{d}$,
we define a cube of size $R$ centered at $a$ by
\begin{equation}
\mathcal{K}_{R}
:=\left\{ y\in\mathbb{Z}^{d}\big|\left|y-a\right|_{\infty}\leq R\right\}  \;.
\end{equation}

\begin{lem}
\label{lem:derbnd}
Given real numbers $s,t$ such that $0<3s<r<1$.
Let $\mathcal{K}_{R}$ and $\mathcal{K}_{rR}$
be cubes of sizes $R$ and $rR$ respectively centered at the same
point. Assume that $u$ is harmonic in the cube $\mathcal{K}_{R}$.
Let $X=\mathcal{K}_{R}\backslash\mathcal{K}_{rR}$, 
$x\in\mathcal{K}_{rR}$
and $d(x,\partial\mathcal{K}_{rR})>sR$. Then
\begin{equation} \label{eq:average1} 
|u(x)|\leq O(R^{-d})\sum_{y\in X} |u(y)| \;,
\end{equation}
\begin{equation} \label{eq:average2}
u(x)^{2}\leq O(R^{-d})\sum_{y\in X}u(y)^{2} \;.
\end{equation}
Here, the constants in the big-O notation depend on $s,t$.
\end{lem}
\begin{proof}
For any integer $rR \leq b<R$, let $\mathcal{K}_{b}$ be
cubes of sizes $b$ co-centered with $\mathcal{K}_{R}$. 
Then since $u$ is harmonic, and the Poisson kernel 
$0\le P_{\mathcal{K}_{b}}(x,y) \le c\, b^{-(d-1)}$ for some constant $c>0$
by the assumption on $x$,
one has
\[
|u(x)|
= \Big|
\sum_{y\in\partial\mathcal{K}_{b}} P_{\mathcal{K}_{b}} (x,y) u(y)
\Big|
\leq c\,b^{-(d-1)} \sum_{y\in\partial\mathcal{K}_{b}} |u(y)| \;.
\]
Multiply both sides by $b^{d-1}$ and sum over $rR \leq b\leq R$,
we have
\begin{equation}  \label{eq:average11}
R^{d}|u(x)|\leq c^{\prime}\sum_{y\in X} |u(y)|
\end{equation}
for some constant $c'>0$ which proves (\ref{eq:average1}). By Cauchy-Schwartz inequality,
\[
|u(x)|
\leq O(R^{-d})\big(\sum_{y\in X}u(y)^{2}\big)^{1/2}|X|^{1/2} \;.
\]
This together with $|X|=O(R^d)$ proves (\ref{eq:average2}). 
\end{proof}
The next Proposition plays an important role in controlling the fundamental
scaling. See the paragraph below Proposition~\ref{prop:scaling} for a motivation.
\begin{prop}
\label{prop:covest}
Let $x\in X\subset U\subset\Lambda$. If $d(x,\partial X)\geq cL^{j}$,
then 
\begin{equation}  \label{eq:decay}
\sum_{y_{1},y_{2}\in\partial X}(\partial_{x,e}P_{X})(x,y_{1})\,
C_{U}(y_{1},y_{2})\,
(\partial_{x,e}P_{X})(x,y_{2})\leq O(1)L^{-dj}
\end{equation}
for all $e\in\mathcal{E}$ where the constant $O(1)$ only depends
on the constant $c$. Here $\partial_{x,e}$ is the discrete derivative w.r.t.
the argument $x$ to the direction $e$.
\end{prop}

\begin{proof}
Notice that $C_{U}\leq C_{\Lambda}$ as quadratic forms, so it is enough
to prove the statement with $C_{U}$ replaced by $C_{\Lambda}$. Since
$y_{2}\in\partial X$ and $C_{\Lambda}(x-y_{2})$ is $-\Delta_{m}$-harmonic
in $x\in X$, one has
\[
\sum_{y_{1}\in\partial X}
	P_{X}(x,y_{1}) C_{\Lambda}(y_{1},y_{2})
=C_{\Lambda}(x,y_{2}) \;.
\]
Taking derivative w.r.t. $x$ on the above equation, we obtain that
the left hand side of eq. (\ref{eq:decay}) is equal to
\begin{equation} \label{eq:alpha12-1}
\sum_{y_{2}\in\partial X}
	\partial_{x,e}C_{\Lambda}(x,y_{2}) \, \partial_{x,e}P_{X}(x,y_{2}) \;.
\end{equation}

By Corollary \ref{cor:Cdecay} (for decay rate of $\nabla C_{\Lambda}$)
and the assumption $d(x,\partial X)\geq cL^{j}$,  one has
\[
|\partial_{x,e}C_{\Lambda}(x,y_{2})|\leq O(L^{-(d-1)j}) \;.
\]
Using again the same assumption, there exists a discrete ball 
$\mathcal B_R(x)\subset X$ centered on $x$ with radius $R=\frac{c}{2}L^j$ (and $R$ is independent of $x$).
For every $y_2\in\partial X$,
$P_X(x,y_2)$ is harmonic and non-negative in $\mathcal B_R(x)$.
Applying \eqref{eq:diff-est-positive},
\[
\left|\partial_{x,e}P_{X}(x,y_{2})\right|
\le c_1 R^{-1} P_{X}(x,y_{2})
\]
with $c_1$ depending on $c$ but independent of $x$ and $y_2$ (since it is independent of the harmonic function).
%
So (\ref{eq:alpha12-1}) is bounded by
\[
O(L^{-(d-1)j}) O(L^{-j}) 
\sum_{y_{2}\in\partial X}
	P_{X}(x,y_{2}) \;.
\]
Since $\sum_{y_{2}\in\partial X} P_{X}(x,y_{2})\leq1$
for all $m>0$ (where $m$ is the mass in $\Delta_m$
and $P_X$ depends on $m$),
the above quantity is bounded by $O(L^{-dj})$.
\end{proof}

\begin{rem}
One may find that our method also resembles Gawedzki and Kupiainen's
approach \cite{gawedzki_rigorous_1980,gawedzki_block_1983} because
the Poisson kernel here plays a similar role as their spin blocking
operator. However, there are many differences. For example, our fluctuation
fields $\zeta$ have finite range covariances; the integrands at different
scales do not have to be in Gibbsian forms; and our polymer arrangements
are closer to Brydges \cite{brydges_lectures_2007}.
\end{rem}

\section{Norms} \label{sec:Norms}

Before we define the norms, we have a remark 
about the choices of four important constants: $L$, $A$, $\kappa$ and $h$
where $L$ has already appeared above and $A$, $\kappa$ and $h$ will appear in the definitions of norms below.

We will first fix $L>L_0(d)$ large enough which satisfies
all the largeness requirements in
Lemma~\ref{lem:geometric} (a geometric result), Lemma~\ref{lem:L2} and Proposition~\ref{prop:L3}.
These results establish contractivity of the three linear maps defined in Proposition~\ref{prop:The-linearization}, and $L$ has to be large to overwhelm some $O(1)$ constants appearing in the estimates of the norms of these linear maps.

We then choose $A>A_0(d,L)$ large enough which satisfies
all the largeness requirements in 
Proposition~\ref{prop:smoothness} (smoothness of RG) and
Proposition~\ref{prop:large-sets} (contractivity of the linear map $\mathcal L_1$ defined in Proposition~\ref{prop:The-linearization}).

After this, we choose $0<\kappa <\kappa_0(d,L,A)$ small enough which satisfies
all the smallness requirements in Lemma~\ref{lem:intproperties} (integrating ``regulators" defined in \eqref{eq:def-reg}) and Lemma~\ref{lem:integrab_poly}.
Finally, we choose $h>h_0(d,L,A,\kappa)$ large enough for the arguments in the proof
of Lemma~\ref{lem:integrab_poly}.

\subsection{Definitions of norms}

We now define the norm of the fields,
the norm of a function of the fields (i.e. elements in $\mathcal N$) at a fixed field,
and the norm of a function in $\mathcal N^{\mathcal P_j}$.
For $j>0$, the definitions are as follows.
\begin{enumerate}
\item
Define $h_{j}=hL^{-(d-2)j/2}$ for constant $h>0$. We first
define the norm for the fields. Let us recall that $\xi$ is the field
introduced in Section \ref{sec:Outline}.  
$X\subset Y$ and $\lambda\in\mathbb{R}$,
we define
\begin{equation} \label{eq:def_Phi}
\Vert(f,\lambda\xi)\Vert_{\Phi_{j}(X,Y)}:=h_{j}^{-1}\sup_{x\in X,e}\left|L^{j}\partial_{e}(P_{Y}f(x)+\lambda\xi(x))\right| \;.
\end{equation}
The notation $\Vert f\Vert_{\Phi_{j}(X,Y)}$ where $\xi$ part is dropped will be understood as $\Vert(f,0)\Vert_{\Phi_{j}(X,Y)}$.
As a special case, if $X\in\mathcal{P}_{j}$ then we simply write
\begin{equation} \label{eq:def-Phi-j-X}
\Vert(f,\lambda\xi)\Vert_{\Phi_{j}(X)}:=\Vert(f,\lambda\xi)\Vert_{\Phi_{j}(\dot{X},X^{+})} \;.
\end{equation}

\item
We then define differentials for functions of the fields, and their
norm. 
Let $K(\phi,\xi)$ be a function of $\phi,\xi$.
For test functions 
\[
\left(f,\lambda\right)^{\times n}:=(f_{1},\lambda_{1}\xi,\cdots,f_{n},\lambda_{n}\xi) \;,
\]
the n-th differential of $K(\phi,\xi)$ is 
\[
K^{(n)}(\phi,\xi;\left(f,\lambda\right)^{\times n})
:=\frac{\partial^{n}}{\partial t_{1}\dots\partial t_{n}}
	K(\phi+\sum_{i=1}^{n}t_{i}f_{i},\xi+\sum_{i=1}^{n}t_{i}\lambda_{i}\xi)\bigg|_{t_{i}=0}  \;.
\]
It is normed with a space of test functions $\Phi$ by
\[
\Vert K^{(n)}(\phi,\xi)\Vert_{T_{\phi}^{n}(\Phi)}
:=\sup_{\Vert(f_{i},\lambda_{i}\xi)\Vert_{\Phi}\leq1}
	\big|K^{(n)}(\phi,\xi;\left(f,\lambda\right)^{\times n})\big| \;.
\]
We then measure the amplitude of $K(\phi,\xi)$ at a fixed field
$\phi$ by incorporating all its derivatives at $\phi$ that we want
to control: 
\begin{equation} \label{eq:allders}
\Vert K(\phi,\xi)\Vert_{T_{\phi}(\Phi)}
:=\sum_{n=0}^{4}\frac{1}{n!}\Vert K^{(n)}(\phi,\xi)\Vert_{T_{\phi}^{n}(\Phi)}
\end{equation}
In most of the discussions,
we actually have a function $K(X,\phi,\xi)$ which is element in $\mathcal N^{\mathcal P_j}$. Then the above $T_\phi(\Phi)$ norm 
is taken for every $X\in\mathcal P_j$, and 
 $\Phi$ will be chosen to be $\Phi_{j}(X)$
 defined in \eqref{eq:def-Phi-j-X}.

\item
For $\kappa>0$, we define ``regulators'':
\begin{equation} \label{eq:def-reg}
G(X,Y):=\mathbb{E}\left[e^{\frac{\kappa}{2}\sum_{x\in X,e\in\mathcal{E}}(\partial_{e}\phi(x))^{2}}\big|Y^{c}\right]\big/N(X,Y)
\end{equation}
for $X\subset Y$ where the normalization factor is defined by 
\begin{equation}
N(X,Y):=\mathbb{E}\left[e^{\frac{\kappa}{2}\sum_{x\in X,e\in\mathcal{E}}(\partial_{e}\phi(x))^{2}}\big|\phi_{Y^{c}}=0\right]
\end{equation}

For $K\in \mathcal N^{\mathcal P_j}$, define
\begin{equation}
\Vert K(X)\Vert_{j}:=\sup_{\phi}\Vert K(X,\phi,\xi)\Vert_{T_{\phi}(\Phi_{j}(X))}G(\ddot{X},X^{+})^{-1}\label{eq:def_normKXj}
\end{equation}
Finally, for $A>0$,
\begin{equation}
\Vert K\Vert_{j}:=\sup_{X\in\mathcal{P}_{j}}\Vert K(X)\Vert_{j}A^{|X|_{j}}\label{eq:def_normKj}
\end{equation}

\end{enumerate}

For the case $j=0$: (\ref{eq:def_Phi})-(\ref{eq:allders}) are still
defined for $j=0$ with $P_{Y}=id$ and $\dot{X}=X$ (recall these
conventions made in Section \ref{sec:Outline}). (\ref{eq:def_normKXj})
is defined with $G$ replaced by 
\begin{equation}
G_{0}(X):=e^{\frac{\kappa}{2}\sum_{x\in X,e\in\mathcal{E}}(\partial_{e}\phi(x))^{2}}
\end{equation}

\subsection{Properties}
\begin{lem}
\label{lem:Tphiproperties}
Let $F$ be function of $\phi$, $\xi$, and $X\subset Y\subset U$.
We have the following property for the $T_{\phi}(\Phi)$ norms:
\begin{equation}
\Vert F^{(n)}(\phi,\xi)\Vert_{T_{\phi}^{n}(\Phi_{j}(Y,U))}
\leq\Vert F^{(n)}(\phi,\xi)\Vert_{T_{\phi}^{n}(\Phi_{j}(X,U))}
\end{equation}
which also holds without $n$.
\end{lem}
\begin{proof}
The proof is immediate because $\left\Vert f\right\Vert _{\Phi_{j}(Y,U)}\geq\left\Vert f\right\Vert _{\Phi_{j}(X,U)}$. 
\end{proof}
Before the discussion on further properties, 
we recall that our functions of the fields have the 
special structure  \eqref{eq:depend-via}.
It turns out that in view of this structure, 
it is sometimes more convenient to consider 
a type of function spaces $\widetilde{\Phi}_{j}(X,Y)$ for $X\subset Y$ defined as follows:
\[
\widetilde{\Phi}_{j}(X,Y):=\{g: \Delta g=0 \mbox{ on }Y, g=0 \mbox{ on }Y^c\}\oplus\mathbb{R}\xi
\]
equipped with semi-norm
\[
\Vert g\oplus\lambda\xi\Vert_{\widetilde{\Phi}_{j}(X,Y)}:=h_{j}^{-1}\sup_{x\in X,e}\left|L^{j}\partial_{e}(g(x)+\lambda\xi(x))\right| \;.
\]
Note that the above sum is really a direct sum since the test function $f$ in \eqref{eq:scalinglimit-2} is not identically zero.
Now if a function $F(\phi,\xi)=\tilde F(\psi)$ with $\psi=P_Y \phi+\xi$,
one actually has
\[
\Vert F^{(n)}(\phi,\xi) \Vert_{T^n_\phi (\Phi_{j}(X,Y))} 
= \!\!\!\!
\sup_{\left\Vert g_{i}\oplus\lambda_{i}\xi\right\Vert _{\widetilde{\Phi}_{j}(X,Y)}\leq 1}
\Big| 
	\partial_{t_{i}}^{n}\big|_{t_{i}=0}
	\tilde F(\psi+  \sum_{i=1}^{n} t_{i} (g_{i}+\lambda_{i}\xi))
\Big|
\]
for any subset $X\subset Y$ since in this situation, varying $\phi$ by 
$t_i f_i$ for generic functions $f_i$ is equivalent with
varying $P_Y \phi$ by harmonic functions on $Y$.

%


\begin{lem}
\label{lem:normEleqEnorm}
Assume the setting of Lemma~\ref{lem:only-dep-harmo}.
For every $k=1,\dots,m$, let $X_k\subset Y_{k}\subset U$.
Define $X:=\cup_{k=1}^m X_k$.
Then, one has
\begin{equation} \label{eq:norm-product}
\Big\Vert 
	\mathbb{E}\Big[\prod_{k=1}^mH_{k}(\phi,\xi)\big|U^{c}\Big]
\Big\Vert_{T_{\phi}(\Phi_{j}(X,U))}
\leq 
\mathbb{E}\Big[\prod_{k=1}^m\left\Vert H_{k}(\phi,\xi)\right\Vert _{T_{\phi}(\Phi_{j}(X_{k},Y_{k}))}\big|U^{c}\Big] \;.
\end{equation}
\end{lem}

\begin{rem}  
Lemma \ref{lem:normEleqEnorm} is stated in terms of generic functions $H_k$. The typical situation in which we apply this lemma is that
$Y_k=X_k^+$, and $H_k(\phi,\xi)=K_k(X_k,\phi,\xi)$
with each $K_k(X_k,\phi,\xi)$ satisfying \eqref{eq:depend-via}.
\end{rem}

\begin{rem}
Lemma \ref{lem:normEleqEnorm}
is analogous with \cite[Lemma~6.7]{brydges_lectures_2007}
(the norm of a product bounded by product of norms)
and \cite[Lemma~6.9]{brydges_lectures_2007} (the norm of an expectation bounded by expectation of the norm).
The difference is that in our approach we combine the two results; in fact, here
both sides of \eqref{eq:norm-product} have the conditional expectation
with the same conditioning, so that the two sides are comparable.
\end{rem}

\begin{proof}[Proof of Lemma~\ref{lem:normEleqEnorm}]
Let $\zeta=\phi-P_U \phi$ and define
\[
F(\phi,\xi):=\mathbb{E}_\zeta \Big[\prod_{k}H_{k}(P_U \phi+\zeta,\xi)\Big] \;.
\]
Lemma~\ref{lem:only-dep-harmo} states that there exists $\widetilde F$ such that $F(\phi,\xi)=\tilde F(P_{U}\phi+\xi)$.
Write $\langle t,f\rangle_n:=\sum_{i=1}^n t_i f_i$.
By the discussion before this lemma,
the $T^n_{\phi}(\Phi_{j}(X,U)) $ norm of $F^{(n)}(\phi,\xi)$
is equal to
\[
\sup_{\left\Vert g_{i}\oplus\lambda_{i}\xi\right\Vert_{\widetilde{\Phi}_{j}(X,U)}\leq 1}
\bigg| 
	\partial_{t_{i}}^{n}\big|_{t_{i}=0}
	\mathbb{E}_\zeta \Big[\prod_{k}H_{k}\Big(
		P_U \phi+\langle t,g\rangle_n+\zeta,\,
		\xi+\langle t,\lambda\xi\rangle_n
	\Big)\Big]
\bigg| \;.
\]
This is bounded by taking the $\mathbb{E}_\zeta$ outside the supremum,
and we apply the product rule of derivatives.
We then obtain
factors of the form
\[
\sup_{ g_{i}\oplus\lambda_{i}\xi}
\Big|  \partial_{t_{i}}^{r}\big|_{t_{i}=0}
H_{k}\Big(
	\phi+\langle t,g\rangle_r,\,\xi+\langle t,\lambda\xi\rangle_r
\Big)
 \Big|
\]
with the sup over the same set as above. 
Since $g_i$ are harmonic on $Y_k$ and by Lemma~\ref{lem:Tphiproperties},
the supremum can be replaced by one taken over all
$g_i\oplus\lambda_i \xi $ such that $g_i$ are harmonic on $Y_k$ and
$ \Vert g_{i}\oplus\lambda_{i}\xi\Vert_{\widetilde{\Phi}_{j}(X_k,Y_k)}\le 1$.
By the assumption on $H_k$, and $P_{Y_k} g =g$,
the above function $H_k$ is equal to
$\widetilde H_k\big(
	P_{Y_k}\phi+\langle t,g\rangle_r + \xi+\langle t,\lambda\xi\rangle_r
\big)$.
Again by the discussion before this lemma,
the above quantity is actually bounded by 
$\Vert H_k(\phi,\xi)\Vert_{T^r_{\phi}(\Phi_{j}(X_k,Y_k))} $.
Summing over multi-indices $(r_1,...,r_m)$ with $|r|=n$,
followed by summing over $n$, one obtains the desired bound.
\end{proof}
Before the next lemma we introduce a short notation
\begin{equation}
(\partial_{m}f)^{2}:=(\partial f)^{2}+m^{2}f^{2}
\end{equation}

\begin{lem}
\label{lem:simproperties}
We have the following properties for the
regulator.
\begin{enumerate}
\item \label{enu:reg1}  $G(X,Y,\phi=0)=1$\;.
\item \label{enu:reg2}If $X_{1}\subset Y_{1}$, $X_{2}\subset Y_{2}$,
and $Y_{1}\cup\partial Y_{1},Y_{2}\cup\partial Y_{2}$ are disjoint,
then
\begin{equation}
G(X_{1},Y_{1})G(X_{2},Y_{2})=G(X_{1}\cup X_{2},Y_{1}\cup Y_{2}) \;.
\end{equation}

\item \label{enu:reg3}
We have an alternative representation of $G(X,Y)$
\begin{equation} \label{eq:alt-rep-G}
G(X,Y)=\exp\left(\frac{\kappa}{2}\sum_{X}(\partial\psi_{1})^{2}-\frac{1}{2}\sum_{Y}(\partial_{m}\psi_{1})^{2}+\frac{1}{2}\sum_{Y}(\partial_{m}\psi_{2})^{2}\right)
\end{equation}
where $\psi_{1}$ is the minimizer of $\sum_{Y}(\partial_{m}\phi)^{2}-\kappa\sum_{X}(\partial\phi)^{2}$
with $\phi_{Y^{c}}$ fixed, and $\psi_{2}$ is the minimizer of $\sum_{Y}(\partial_{m}\phi)^{2}$
with $\phi_{Y^{c}}$ fixed.

\item \label{enu:reg4}
Fixing $Y$, $G(X,Y)$ is monotonically increasing
in $X$ for all $X\subset Y$.
\item \label{enu:reg2bds}With $\psi_{1,2}$ defined in (\ref{enu:reg3}),
\begin{equation}
\exp\Big(\frac{\kappa}{2}\sum_{X}(\partial\psi_{2})^{2}\Big)\leq G(X,Y)
\leq
\exp\Big(\frac{\kappa}{2}\sum_{X}(\partial\psi_{1})^{2}\Big) \;.
\end{equation}
\end{enumerate}
\end{lem}

\begin{proof}
(\ref{enu:reg1})(\ref{enu:reg2}) hold by definition and the fact
that $G(X,Y)$ is a function of $\phi$ on $\partial Y$. For (\ref{enu:reg3}),
\begin{equation}
\begin{aligned}G(X,Y) & =\frac{\int e^{\frac{\kappa}{2}\sum_{X}(\partial\phi)^{2}-\frac{1}{2}\sum_{\Lambda}(\partial_{m}\phi)^{2}}d^{Y}\phi\bigg/\int e^{-\frac{1}{2}\sum_{\Lambda}(\partial_{m}\phi)^{2}}d^{Y}\phi}{\int e^{\frac{\kappa}{2}\sum_{X}(\partial\phi)^{2}-\frac{1}{2}\sum_{Y}(\partial_{m}^{D}\phi)^{2}}d^{Y}\phi\bigg/\int e^{-\frac{1}{2}\sum_{Y}(\partial_{m}^{D}\phi)^{2}}d^{Y}\phi}\end{aligned}
\end{equation}
where $d^{Y}\phi$ is the Lebesgue measure on $\{\phi(x):x\in Y\}\cong\mathbb{R}^{Y}$,
$\partial^{D}$ takes Dirichlet boundary condition on $\partial Y$.
Using Fact (\ref{eq:quadratic_fact}) for both quadratic forms $-\frac{\kappa}{2}\sum_{X}(\partial\phi)^{2}+\frac{1}{2}\sum_{\Lambda}(\partial_{m}\phi)^{2}$
and $\frac{1}{2}\sum_{\Lambda}(\partial_{m}\phi)^{2}$, we obtain
(\ref{enu:reg3}), where the quantity 
\[
\int e^{\frac{\kappa}{2}\sum_{X}(\partial\phi)^{2}-\frac{1}{2}\sum_{Y}(\partial_{m}^{D}\phi)^{2}}d^{Y}\phi
\]
appears in both numerator and denominator and thus cancels, and so
does the quantity 
\[
\int e^{-\frac{1}{2}\sum_{Y}(\partial_{m}^{D}\phi)^{2}}d^{Y}\phi \;.
\]

(\ref{enu:reg4}) holds because of (\ref{enu:reg3}) and that
\begin{equation}
\inf_{\phi}\Big\{ \sum_{Y}(\partial_{m}\phi)^{2}-\kappa\sum_{X}(\partial\phi)^{2}\big|Y^{c}\Big\} 
\end{equation}
is monotonically decreasing in $X$. The two inequalities in (\ref{enu:reg2bds})
hold by replacing $\psi_{1}$ by $\psi_{2}$ or replacing $\psi_{2}$
by $\psi_{1}$, and using definitions of $\psi_{1},\psi_{2}$.
\end{proof}

\begin{rem}
The regulator in \cite{brydges_lectures_2007} has the form $e^{\kappa\sum(\partial\phi^{\prime})^{2}+\mbox{the other terms}}$,
since the smoothed field $\phi^{\prime}$ there is analogous to our
$\psi$, the last property above implies that our regulator has about
the same amplitude as the one in \cite{brydges_lectures_2007}, except
that we no longer need the other terms.
\end{rem}

Before proving a furthur property we recall a formula. If $U$ is
a finite set and $\psi=\{\psi(x):x\in U\}$ is a family of centered
Gaussian random variables with covariance identity, and $T:l^{2}(U)\rightarrow l^{2}(U)$
satisfies $\Vert T\Vert<1$ then
\begin{equation} \label{eq:Gaussianint}
\mathbb{E}
\Big[\exp\Big(\frac{1}{2}\left(\psi,T\psi\right)_{l^{2}(U)}\Big)\Big]
=\det\left(1-T\right)^{-1/2}
=\exp\Big(\frac{1}{2}\sum_{n=1}^{\infty}\frac{1}{n}Tr(T^{n})\Big)
\end{equation}

The next lemma shows that the conditional expectations almost automatically
do the work when one wants to see how the regulators undergo integrations,
except that we need to manually control a ratio of normalizations.
\begin{lem}
\label{lem:intproperties}
Suppose that $\kappa>0$ is sufficiently small.
For $X\subset Y\subset U$, and $d(X,Y^{c})=c_{0}L^{j}$,
one has the bound
\begin{equation}
\mathbb{E}\left[G(X,Y)\big|U^{c}\right]\leq c^{L^{-dj}|X|}G(X,U)
\end{equation}
if $U\neq\Lambda$, for some constant $c$ only depending on $c_{0}$.
One also has, as the special case, the short-hand notation and bound
\[
\mathbb{E}\left[G(X,Y)\big|(\Lambda^{+})^{c}\right]:=\mathbb{E}\left[G(X,Y)\right]\leq c^{L^{-dj}|X|} \;.
\]

In particular if $X=\ddot{X}_{0}$ for some $X_{0}\in\mathcal{P}_{j}$,
then the factor $c^{L^{-dj}|X|}$ can be written as $c^{|X_{0}|_{j}}$.
Furthurmore, $G_{0}$ also satisfies the same bound.
\end{lem}

\begin{proof}
By definition one has
\[
\mathbb{E}\left[G(X,Y)\big|U^{c}\right]=\mathbb{E}\left[e^{\frac{\kappa}{2}\sum_{x\in X,e\in\mathcal{E}}(\partial_{e}\phi(x))^{2}}\big|U^{c}\right]\big/N(X,Y)=G(X,U)\frac{N(X,U)}{N(X,Y)}  \;.
\]
Define $\phi=C_{Y}^{1/2}\psi$ so that $\psi$ has covariance identity,
where $C_{Y}$ is the Dirichlet Green's function for $Y$. Then define
$T_{Y}=\frac{1}{2}\sum_{e\in\mathcal{E}}(\partial_{e}C_{Y}^{1/2})^{\star}1_{X}(\partial_{e}C_{Y}^{1/2})$
as an operator on $l^{2}=l^{2}(\Lambda)$. We define similar operators
$C_{U},T_{U}$ in the same way for $U$. Let $\partial_{e}^{D}$,
$-\Delta_{Y}$ take Dirichlet boundary condition on $\partial Y$.
Because $C_{Y}$ is the inverse of $-\Delta_{Y}$ 
\begin{equation}
\begin{aligned}(f,T_{Y}f)_{l^{2}} & =\frac{1}{2}\sum_{x\in X,e\in\mathcal{E}}(\partial_{e}C_{Y}^{1/2}f(x))^{2}\leq\frac{1}{2}\sum_{x\in Y,e\in\mathcal{E}}(\partial_{e}^{D}C_{Y}^{1/2}f(x))^{2}\\
 & \leq\sum_{x\in Y}C_{Y}^{1/2}f(x)(-\Delta_{Y})C_{Y}^{1/2}f(x)\\
 & \leq(f,f)_{l^{2}}
\end{aligned}
\label{eq:T<1}
\end{equation}
So $\left\Vert T_{Y}\right\Vert \leq1$. Similarly $\left\Vert T_{U}\right\Vert \leq1$.
By (\ref{eq:Gaussianint})
\begin{equation}
\frac{N(X,U)}{N(X,Y)}=\frac{\mathbb{E}\left[e^{\frac{\kappa}{2}(\psi,T_{U}\psi)}\right]}{\mathbb{E}\left[e^{\frac{\kappa}{2}(\psi,T_{Y}\psi)}\right]}=\left(\frac{\det(1-\kappa T_{U})}{\det(1-\kappa T_{Y})}\right)^{-1/2}
\end{equation}
Taking logarithm, we need to compute
\[
Tr\left(\log(1-\kappa T_{U})-\log(1-\kappa T_{Y})\right)\leq O(1)Tr\left(\kappa T_{U}-\kappa T_{Y}\right)
\]
where we have used $\left\Vert T_{Y}\right\Vert \leq1$, $\left\Vert T_{U}\right\Vert \leq1$,
$\kappa$ is small, and $\log(1-x)$ is Lipschitz on $x\in[-\frac{1}{2},\frac{1}{2}]$.
Since $C_{U}-C_{Y}=P_{Y}C_{U}$, 
\begin{equation}
\begin{aligned}Tr\left(T_{U}-T_{Y}\right) & =\frac{1}{2}\sum_{e\in\mathcal{E},x\in X}\partial_{e}(C_{U}-C_{Y})\partial_{e}^{\star}(x,x)\\
 & =\frac{1}{2}\sum_{e\in\mathcal{E},x\in X}\sum_{y\in\partial Y}\partial_{x,e}P_{Y}(x,y)\partial_{x,e}C_{U}(y,x)
\end{aligned}
\label{eq:intprop_trace}
\end{equation}
By Lemma \ref{lem:derbnd} and proceed similarly as eq. (\ref{eq:alpha12-1})
in proof of Proposition \ref{prop:covest}, making use of the $O(L^{j})$
distance between $x$ and $y$, the above expression is bounded by
$O(L^{-jd})\left|X\right|$ which concludes the proof. 

The bound on $\mathbb{E}\left[G(X,Y)\big|(\Lambda^{+})^{c}\right]$
is similar. The only modification is to replace $C_{U}$ by $C_{\Lambda}$
which satisfies periodic instead of Dirichlet boundary condition.
For $G_{0}$, we can directly bound $\mathbb{E}\left[e^{\frac{\kappa}{2}\sum_{x\in X,e\in\mathcal{E}}(\partial_{e}\phi(x))^{2}}\big|U^{c}\right]$
by $c^{|X|}$.
\end{proof}

\section{Smoothness of RG} \label{sec:Smoothness-of-RG}

In this section we prove that the RG map constructed in Section \ref{sec:The-renormalization-group}
is smooth. First of
all, we need some geometric results from \cite{brydges_lectures_2007}.
\begin{lem}
\label{lem:geometric}(Brydges \cite{brydges_lectures_2007}) There
exists an $\eta=\eta(d)>1$ such that for all $L\geq2^{d}+1$ and
for all large connected sets $X\in\mathcal{P}_{j}$, $|X|_{j}\geq\eta|\bar{X}|_{j+1}$.
In addition, for all $X\in\mathcal{P}_{j}$, $|X|_{j}\geq|\bar{X}|_{j+1}$,
and 
\begin{equation}
|X|_{j}\geq\frac{1}{2}(1+\eta)|\bar{X}|_{j+1}-\frac{1}{2}(1+\eta)2^{d+1}|\mathcal{C}(X)|
\end{equation}
\end{lem}
\begin{proof}
The lemma is the same with \cite{brydges_lectures_2007} (Lemma 6.15
and 6.16), so we omit the proof.
\end{proof}

In the following result, assuming $j\ge 0$,
we omit subscript $j$ for objects at scale $j$ and write a prime
for objects at scale $j+1$, as in Section \ref{sec:The-renormalization-group}.
Recall that the spaces $\mathcal N^{\mathcal P_c}$, $\mathcal N^{\mathcal P^\prime_c}$ are defined in Subsection~\ref{sec:Basics-Polymers}, 
and they are equipped with norms defined in Subsection~\ref{eq:def_normKj}.

\begin{prop}
\label{prop:smoothness}
Let $B^{\prime}(\mathcal{N}^{\mathcal{P}_c^\prime})$
be a ball centered on the origin in $\mathcal{N}^{\mathcal{P}_c^\prime}$.
There exists $A(d,L,B^{\prime})$ and $A^{\star}(d,A)$ such that
for $A>A(d,L,B^{\prime})$ and $A^{\star}>A^{\star}(d,A)$, the map
$(\sigma_{j},E_{j+1},\sigma_{j+1},K_{j})\mapsto K_{j+1}$ defined
in Subsection~\ref{subsec:RG-steps} is smooth 
from $(-A^{\star-1},A^{\star-1})^{3}\times B_{A^{\star-1}}(\mathcal{N}^{\mathcal{P}_c})$
to $B^{\prime}(\mathcal{N}^{\mathcal{P}_c^\prime})$ where $B_{A^{\star-1}}(\mathcal{N}^{\mathcal{P}_c})$
is a ball centered on the origin in $\mathcal{N}^{\mathcal{P}_c}$
with radius $A^{\star-1}$.
\end{prop}

\begin{proof}
Let
$A^{\star-1}\ll\kappa $.
For $U\in\mathcal{P}_{c}^{\prime}$, by definition of $K^{\sharp}$,
\begin{equation} \label{eq:Kj+1_smoothness}
\begin{aligned}
K^{\prime}(U)&=  
\sum_{V\subseteq U,V\neq\emptyset}
\sum_{(P,Q,Z,\mathcal{X},\hat{\mathcal{Y}})\rightarrow V}
\mathbf{E}^{U^{+}} \\
&\qquad\times
\underbrace{
(1-e^{E^{\prime}})^{P}
\,(e^{E^{\prime}})^{(\left\langle X\right\rangle \backslash X)\,\cap\, U\backslash(P\cup Q)}
\,(e^{-E^{\prime}})^{U\,\cup\, X}}_{
	\leq(A^{\star}/2)^{-|P|_{j}}
	\,2^{|(\left\langle X\right\rangle \backslash X)\,\cap\, U\backslash(P\cup Q)|_{j}}
	\,2^{|U\cup X|_{j}}}
\end{aligned}
\end{equation}
where, with $\prod K:=\prod_{Y\in\mathcal{X}}K(Y)\prod_{(B,Y)\in\hat{\mathcal{Y}}}\frac{1}{|Y|_{j}}K(B,Y)$
as a short-hand notation, 
\begin{equation}
\begin{aligned} & \mathbf{E}^{U^{+}}:=\mathbb{E}\bigg[(\tilde{I}-e^{E^{\prime}})^{(U\backslash V)\cap\left\langle X\right\rangle ^{c}}\tilde{I}^{V\cap(\left\langle X\right\rangle ^{c}\backslash Z)}\delta I^{Z}(I-e^{E^{\prime}})^{Q}\prod K\big|(U^{+})^{c}\bigg]\\
= & \mathbb{E}\bigg[\underbrace{\mathbb{E}\bigg[(\tilde{I}-e^{E^{\prime}})^{(U\backslash V)\cap\left\langle X\right\rangle ^{c}}\tilde{I}^{V\cap(\left\langle X\right\rangle ^{c}\backslash Z)}\delta I^{Z}(I-e^{E^{\prime}})^{Q}\big|(W^{+})^{c}\bigg]}_{=:\mathbf{E}^{W^{+}}}\prod K\big|(U^{+})^{c}\bigg]
\end{aligned}
\end{equation}
where $W=U\backslash\hat{X}$ (recall that $X:=X_{\mathcal{X}\cup\mathcal{Y}}$)
and the last step used the corridors around $K(Y)$ in order to make
sense of the $(W^{+})^{c}$ conditional expectation. In the above
$W^{+}$ is a $+$ operation at scale $j$ and $U^{+}$ is a $+$
operation at scale $j+1$. 

We first control $\mathbf{E}^{W^{+}}$. With $\phi=P_{W^{+}}\phi+\zeta$
and the inequality $(a+b)^{2}\leq2a^{2}+2b^{2}$, and using assumption
$A^{\star-1}\ll\kappa $, Lemma \ref{lem:estimateV}, we list the estimates
for each factors.
\[
\Vert(I-e^{E^{\prime}})(B)\Vert_{T_{\phi}(\Phi_{j}(B))}\leq5(\kappa A^{\star})^{-1}e^{\frac{\kappa}{2}\sum_{B}(\partial P_{W^{+}}\phi)^{2}+\frac{\kappa}{2}\sum_{B}(\partial P_{B^{+}}\zeta)^{2}}
\]
for all $B\in Q$, where $B^{+}\subseteq W^{+}$ since $Q\subseteq\left\langle X\right\rangle \backslash\hat{X}$;
and, 
\[
\Vert(\tilde{I}-e^{E^{\prime}})(B)\Vert_{T_{\phi}(\Phi_{j}(B))}\leq5(\kappa A^{\star})^{-1}e^{\frac{\kappa}{2}\sum_{B}(\partial P_{W^{+}}\phi)^{2}+\frac{\kappa}{2}\sum_{B}(\partial P_{(\bar{B})^{+}}\zeta)^{2}}
\]
for all $B\in\mathcal{B}_{j}((U\backslash V)\cap\left\langle X\right\rangle ^{c})$,
where $(\bar{B})^{+}\subseteq W^{+}$ since $\left\langle X\right\rangle $
is designed to ensure that; and
\[
\Vert\tilde{I}(B)\Vert_{T_{\phi}(\Phi_{j}(B))}\leq2e^{\frac{\kappa}{2}\sum_{B}(\partial P_{W^{+}}\phi)^{2}+\frac{\kappa}{2}\sum_{B}(\partial P_{(\bar{B})^{+}}\zeta)^{2}}
\]
for all $B\in\mathcal{B}_{j}(V\cap(\left\langle X\right\rangle ^{c}\backslash Z))$,
where $(\bar{B})^{+}\subseteq W^{+}$ since $B\subseteq\left\langle X\right\rangle ^{c}$;
and
\[
\begin{aligned}\left\Vert \delta I(B)\right\Vert _{T_{\phi}(\Phi_{j}(B))} & \leq\left\Vert I(B)-1\right\Vert _{T_{\phi}(\Phi_{j}(B))}+\Vert\tilde{I}(B)-1\Vert_{T_{\phi}(\Phi_{j}(B))}\\
 & \leq8(\kappa A^{\star})^{-1}e^{\frac{\kappa}{2}\sum_{B}(\partial P_{W^{+}}\phi)^{2}}e^{\frac{\kappa}{2}\sum_{B}(\partial P_{B^{+}}\zeta)^{2}+\frac{\kappa}{2}\sum_{B}(\partial P_{(\bar{B})^{+}}\zeta)^{2}}
\end{aligned}
\]
by $e^{a}+e^{b}\leq2e^{a+b}$ ($a,b>0$) for all $B\in\mathcal{B}_{j}(Z)$,
where $(\bar{B})^{+}\subseteq W^{+}$ since $Z\subseteq\left\langle X\right\rangle ^{c}$.
Combining all above estimates, together with Lemma \ref{lem:normEleqEnorm},
we have
\begin{equation}
\Vert\mathbf{E}^{W^{+}}\Vert_{T_{\phi}(\Phi_{j}(W))}\leq(\kappa A^{\star}/8)^{-|Q\cup Z\cup((U\backslash V)\backslash\left\langle X\right\rangle )|_{j}}e^{\frac{\kappa}{2}\sum_{W}(\partial P_{W^{+}}\phi)^{2}}\mathcal{M}
\end{equation}
where
\begin{equation} \label{eq:mathcalM}
\mathcal{M}
\leq  
\mathbb{E}_{\zeta}\big[
e^{\frac{\kappa}{2}\sum_{B\in\mathcal{B}_{j}(W)}\sum_{B}(\partial P_{B^{+}}\zeta)^{2}}e^{\frac{\kappa}{2}\sum_{B\in\mathcal{B}_{j}\left(W\right)}\sum_{B}(\partial P_{(\bar{B})^{+}}\zeta)^{2}}
\big] \;.
\end{equation}
In the next Lemma we show that $\mathcal{M}\leq c^{|U|_{j}}$.

Now we proceed to control $\mathbf{E}^{U^{+}}$. Instead of $(a+b)^{2}\leq2a^{2}+2b^{2}$
we use properties of the regulator established in Section \ref{sec:Norms}.
Since for all $X\in\mathcal{P}_{j,c}$
\[
\Vert K_{j}(X)\Vert_{T_{\phi}(\Phi_{j}(X))}\leq A^{\star-1}G(\ddot{X},X^{+})A^{-|X|_{j}} \;.
\]
By Lemma \ref{lem:normEleqEnorm}, Lemma \ref{lem:simproperties} (\ref{enu:reg2})(\ref{enu:reg4})(\ref{enu:reg2bds})
and Lemma \ref{lem:intproperties}, 
\begin{equation}
\begin{aligned}
& \Vert \mathbf{E}^{U^{+}}  \Vert_{T_{\phi}(\Phi_{j}(U))}
\leq 
c^{|U|_{j}}\,
(\kappa A^{\star}/8)^{
	-|Z\,\cup\, Q\,\cup\,\big((U\backslash V)\backslash\left\langle X\right\rangle \big)|_{j}-|\mathcal{X}|-|\mathcal{Y}|}
A^{-|X_{\mathcal{X}\cup\mathcal{Y}}|_{j}}  \\
 & \qquad\qquad\qquad\times
 \mathbb{E}\Big[e^{\frac{\kappa}{2}\sum_{W}(\partial P_{W^{+}}\phi)^{2}}
 \prod_{Y\in\mathcal{X}}G(\ddot{Y}_{k},Y_{k}^{+})\prod_{Y\in\mathcal{Y}}G(\ddot{Y}_{i},Y_{i}^{+})\big|(U^{+})^{c}\Big]
 \\
 & \leq c^{|U|_{j}}\cdot(\kappa A^{\star}/8)^{-|Z\cup Q\cup((U\backslash V)\backslash\left\langle X\right\rangle )|_{j}-|\mathcal{X}|-|\mathcal{Y}|}G(\ddot{U},U^{+})c^{\prime|W|_{j}}(A/c^{\prime})^{-|X_{\mathcal{X}\cup\mathcal{Y}}|_{j}}  .
\end{aligned}
\label{eq:EUplus}
\end{equation}

We can bound the number of terms in the summation in (\ref{eq:Kj+1_smoothness})
by $k^{|U|_{j}}$ with $k=2^{7}$, because every $j$-block in $U$
either belongs to $V$ or $V^{c}$, and the same statement is true
if $V$ is replaced by $P,Q,Z,X_{\mathcal{X}},Y_{\mathcal{Y}}$, and
if it is in $Y\in\mathcal{Y}$ it is either the $B$ of $(B,Y)\in\hat{\mathcal{Y}}$
or not. By Lemma \ref{lem:geometric}, for $a=\frac{1}{2}(1+\eta)$,
with $\mathcal{X}=\{X_{k}\}$, $\hat{\mathcal{Y}}=\{(B_{i},Y_{i})\}$,
the quantity $a|U|_{j+1} $ can be bounded by
\[
\begin{aligned}
& 
a|\bar{Z}|_{j+1}+a|\cup_{i}\bar{B}_{i}|_{j+1}+a|\cup_{k}\bar{X}_{k}|_{j+1}+a|\bar{Q}|_{j+1}+a|(U\backslash V)\cap\left\langle X\right\rangle ^{c}|_{j+1}\\
 & \leq\big(|Z|_{j}+a2^{d+1}|\mathcal{C}(Z)|\big)+a|\hat{\mathcal{Y}}|+\big(\sum_{k}|X_{k}|_{j}+a2^{d+1}|\mathcal{X}|\big)\\
 & \qquad+\big(|Q|_{j}+a2^{d+1}|\mathcal{C}(Q)|\big)+aL^{d}|(U\backslash V)\cap\left\langle X\right\rangle ^{c}|_{j}\\
 & \leq(1+a2^{d+1})\,(|Z|_{j}+|Q|_{j})+a|\hat{\mathcal{Y}}|\\
 &\qquad
 +(|X_{\mathcal{X}}|_{j}+a2^{d+1}|\mathcal{X}|)+aL^{d}|(U\backslash V)\cap\left\langle X\right\rangle ^{c}|_{j} \;.
\end{aligned}
\]
Then we can easily check that with $A,A^{\star}$ sufficiently large
as assumed in the proposition 
\[
\left\Vert K^{\prime}\right\Vert _{j+1}=\sup_{U\in\mathcal{P}^{\prime}}\left\Vert K^{\prime}(U)\right\Vert _{j+1}A^{a|U|_{j+1}}A^{(1-a)|U|_{j+1}}<r
\]
where $r$ is the radius of $B^{\prime}(\mathcal{N}_{j+1}^{\mathcal{P}_{j+1}})$,
because $A^{|X_{\mathcal{X}}|_{j}}$ is cancelled by its inverse
in (\ref{eq:EUplus}), and 
\begin{equation}
\begin{aligned}
\lim_{A\rightarrow\infty}
& A^{(1-a)|U|_{j+1}}
\cdot A^{-|X_{\mathcal{Y}}|_{j}}
\,\cdot\, k^{|U|_{j}}
\,\cdot\, c^{|U|_{j}}
\,\cdot\, c^{\prime \,|W|_{j}+|X_{\mathcal{X}\cup\mathcal{Y}}|_{j}}\\
&\qquad\qquad\times
2^{|(\left\langle X\right\rangle \backslash X) \,\cap \,U\backslash(P\cup Q)|_{j}}
\,\cdot\, 2^{|U\cup X|_{j}} =0
\end{aligned}
\end{equation}
\begin{equation}
\begin{aligned}\lim_{A^{\star}\rightarrow\infty} & (\kappa A^{\star}/8)^{-|Q\cup Z\cup((U\backslash V)\backslash\left\langle X\right\rangle )|_{j}-|\mathcal{X}|-|\mathcal{Y}|}\\
 & \cdot A^{(1+a2^{d+1})|Q\cup Z|_{j}+a|\hat{\mathcal{Y}}|+a2^{d+1}|\mathcal{X}|+aL^{d}|(U\backslash V)\cap\left\langle X\right\rangle ^{c}|_{j}}=0 \;.
\end{aligned}
\end{equation}
The derivatives of the map $(\sigma_{j},E_{j+1},\sigma_{j+1},K_{j})\mapsto K_{j+1}$
are bounded similarly.
\end{proof}

\begin{lem}
Let $\mathcal{M}$ be the quantity introduced in the proof of Proposition
\ref{prop:smoothness}. There exists a constant $c$ independent of
$L,A,A^{\star}$ such that
\begin{equation}
\mathcal{M}\leq c^{|U|_{j}} \;.
\end{equation}
\end{lem}
\begin{proof}
Defining $\zeta=C_{W^{+}}^{1/2}\psi$ where $C_{W^{+}}$ is the Dirichlet
Green's function for $W^{+}$ and $\psi\in L^{2}(W^{+})$, $\mathcal{M}$
is bounded by 
\begin{equation}
\mathbb{E}_{\psi}\exp\Big\{ 4\kappa\sum_{x\in W}\psi(x)T\psi(x)\Big\} 
\end{equation}
where $\psi$ has identity covariance and 
\begin{equation}
\begin{aligned}
T & =\frac{1}{4}\sum_{B\in\mathcal{B}_{j}(W),e\in\mathcal{E}}
\!\!\!\!\!\!\!\!
\big(C_{U^{+}}^{1/2}P_{B^{+}}^{\star}\partial_{e}^{\star}1_{B}\partial_{e}P_{B^{+}}C_{U^{+}}^{1/2}+C_{U^{+}}^{1/2}P_{(\bar{B})^{+}}^{\star}\partial_{e}^{\star}1_{B}\partial_{e}P_{(\bar{B})^{+}}C_{U^{+}}^{1/2}\big)\\
 & =:T_{1}+T_{2}
\end{aligned}
\end{equation}
is a linear map from $L^{2}(W^{+})$ to itself. $T_{1},T_{2}$ are
defined to be the two terms respectively. We have by Proposition \ref{prop:covest},
\begin{equation}
\begin{aligned}Tr(T) & =\frac{1}{4}\sum_{B\in\mathcal{B}_{j}(W),e\in\mathcal{E}}\bigg(\sum_{x\in B}\partial_{e}P_{B^{+}}C_{U^{+}}\left(\partial_{e}P_{B^{+}}\right)^{\star}(x,x)\\
 & \qquad\qquad
+\sum_{x\in B}\partial_{e}P_{(\bar{B})^{+}}C_{U^{+}}
	\big(\partial_{e}P_{(\bar{B})^{+}}\big)^{\star}(x,x)\bigg)\\
 & \leq O(1)(L^{-dj}+L^{-d(j+1)})|W|\\
 & \leq O(1)|W|_{j}  \;.
\end{aligned}
\label{eq:traceT}
\end{equation}
For the next step we bound $\Vert T\Vert$. In fact, 
\begin{equation}
\begin{aligned}
(f,T_{1}f)_{l^{2}} 
& =\frac{1}{4}\sum_{B\in\mathcal{B}_{j}(W)}
\sum_{x\in B,e}
\big(\partial_{e}P_{B^{+}}C_{U^{+}}^{\frac{1}{2}}f(x)\big)^{2}\\
 & \leq
 \frac{1}{4}\sum_{B\in\mathcal{B}_{j}(W)}
 \sum_{x\in B^{+},e}
 \big(\partial_{e}C_{U^{+}}^{\frac{1}{2}}f(x)\big)^{2}\\
 & \leq 
 c_{d}\sum_{x\in W,e}
 	\big(\partial_{e}C_{U^{+}}^{\frac{1}{2}}f(x)\big)^{2}
\end{aligned}
\end{equation}
where we used the fact that the harmonic extension minimizes the Dirichlet
form to get rid of the Poisson kernels. The constant $c_{d}$ comes
from overlapping of $B^{+}$'s. Then we can proceed as (\ref{eq:T<1})
to bound the above expression by $c_{d}(f,f)_{l^{2}}$. $T_{2}$ is
bounded in the same way. Now by $\left|Tr(T^{n})\right|\leq\left|Tr(T)\right|\left\Vert T\right\Vert ^{n-1}$,
and formula (\ref{eq:Gaussianint}) the proof of the lemma is completed.
\end{proof}

\section{Linearized RG} \label{sec:Linearized-RG}

Having established smoothness, in this section we study the linearization
of the RG map in $\sigma_j,K_j,E_{j+1}$ and $\sigma_{j+1}$. 

In view of Lemma \ref{lem:normEleqEnorm}, we can show, by induction
along all the RG steps, that $K_{j}(X)$ depends on $\phi,\xi$ via
$P_{X^{+}}\phi+\xi$ (at scale $0$, $I_{0},K_{0}$ depend on $\phi,\xi$
via $\phi+\xi$). We write
\[
Tay\,\mathbb{E}\left[K_{j}(X)|(U^{+})^{c}\right]
\]
to be the second order Taylor expansion of $\mathbb{E}\left[K_{j}(X)|(U^{+})^{c}\right]$
in $P_{U^{+}}\phi+\xi$.
\begin{prop}
\label{prop:The-linearization}
The linearization of the map $(\sigma_{j},E_{j+1},\sigma_{j+1},K_{j})\rightarrow K_{j+1}$
around $(0,0,0,0)$ is $\mathcal{L}_{1}+\mathcal{L}_{2}+\mathcal{L}_{3}$
where
\begin{equation}
\mathcal{L}_{1}K_{j}(U)=\sum_{X\in\mathcal{P}_{j,c}\backslash\mathcal{S}_{j},\bar{X}=U}\mathbb{E}\left[K_{j}(X)\big|(U^{+})^{c}\right] \;,
\end{equation}
\begin{equation} \label{eq:def-CL2}
\mathcal{L}_{2}K_{j}(U)
=\sum_{B\in\mathcal{B}_{j},\bar{B}=U}\sum_{X\in\mathcal{S}_{j},X\supseteq B}
\frac{1}{|X|_{j}} (1-Tay)\mathbb{E}\left[K_{j}(X)\big|(U^{+})^{c}\right]  \;,
\end{equation}
\begin{equation}
\begin{aligned} 
\mathcal{L}_{3}\big(\sigma_{j},E_{j+1}, &\sigma_{j+1},K_{j}\big)(U) 
=\sum_{B\in\mathcal{B}_{j},\bar{B}=U}
\bigg(
\frac{\sigma_{j+1}}{4}\sum_{x\in B,e}\left(\partial_{e}P_{(\bar{B})^{+}}\phi(x)+\xi(x)\right)^{2} \\
&\qquad
+E_{j+1}(B)
-\frac{\sigma_{j}}{4}\sum_{x\in B}\mathbb{E}\left[\left(\partial P_{B^{+}}\phi(x)+\xi(x)\right)^{2}\big|(U^{+})^{c}\right]
\\
 & \qquad
+ \sum_{X\in\mathcal{S}_{j},X\supseteq B}
 \frac{1}{|X|_{j}}\,Tay\,\mathbb{E}\left[K_{j}(X)|(U^{+})^{c}\right]\bigg) \;.
\end{aligned}
\end{equation}
\end{prop}
\begin{proof}
In Proposition \ref{prop:smoothness} we proved that the map $(\sigma_{j},E_{j+1},\sigma_{j+1},K_{j})\rightarrow K_{j+1}$
is smooth around $(0,0,0,0)$ so that we can linearize the map. In
(\ref{eq:Ksharp}) since $V\neq\emptyset$, the $\tilde{I}_{j}-e^{E_{j+1}}$
factor doesn't contribute to the linear order. Also if $X=\emptyset$
then $\hat{X}=\left\langle X\right\rangle =\emptyset$, so $1-e^{E_{j+1}}$
and $I_{j}-e^{E_{j+1}}$ don't contribute to the linear order either.
The terms that contribute to the linear order correspond to $(Z,|\mathcal{X}|,|\hat{\mathcal{Y}}|)$
equal to $(\emptyset,0,1)$ or $(\emptyset,1,0)$ or $(B,0,0)$ where
$B\in\mathcal{B}_{j}$. Grouping these terms into large sets part
and small sets part with Taylor leading terms and remainder we obtain
the above linear operators.
\end{proof}

\subsection{Large sets}

\begin{prop} \label{prop:large-sets}
Let $L$ be sufficiently large. Let $A$ be sufficiently large depending
on $L$. Then $\mathcal{L}_{1}$ in Proposition \ref{prop:The-linearization}
is a contraction. Moreover, 
$
\lim_{A\rightarrow\infty}\left\Vert \mathcal{L}_{1}\right\Vert =0$.
\end{prop}

\begin{proof}
By Lemma \ref{lem:intproperties}
\begin{equation}
\Vert\mathcal{L}_{1}K_{j}(U)\Vert_{j+1}\leq\sum_{X\in\mathcal{P}_{j,c}\backslash\mathcal{S}_{j},\bar{X}=U}\Vert K_{j}\Vert_{j}c^{|X|_{j}}A^{-|X|_{j}} \;,
\end{equation}
therefore by Lemma \ref{lem:geometric},
\begin{equation}
\begin{aligned}\Vert\mathcal{L}_{1}K_{j}\Vert_{j+1} & =\sup_{U\in\mathcal{P}_{j+1}}\Vert\mathcal{L}_{1}K_{j}(U)\Vert_{j+1}A^{|U|_{j+1}}\\
 & \leq\bigg[\sup_{U\in\mathcal{P}_{j+1}}A^{|U|_{j+1}}\sum_{X\in\mathcal{P}_{j,c}\backslash\mathcal{S}_{j},\bar{X}=U}c^{|X|_{j}}A^{-|X|_{j}}\bigg]\Vert K_{j}\Vert_{j}\\
 & \leq\bigg[\sup_{U\in\mathcal{P}_{j+1}}A^{|U|_{j+1}}2^{L^{d}|U|_{j+1}}(A/c)^{-\eta|U|_{j+1}}\bigg]\Vert K_{j}\Vert_{j}
\end{aligned}
\end{equation}
where $\eta>1$ is introduced in Lemma \ref{lem:geometric}. The bracketed
expression goes to zero as $A\rightarrow\infty$.
\end{proof}

\subsection{Taylor remainder}

We prepare to show contractivity of $\mathcal{L}_{2}$. We first show
that the Taylor remainder after the second derivative is bounded by
the third derivative. It is a general result about the $T_{\phi}(\Phi)$
norm with no need to specify the test function space $\Phi$.
\begin{lem}
\label{lem:3rdder}
For $F$ a function of $\phi$ let $Tay_{n}$ be
its n-th order Taylor expansion about $\phi=0$, and $\Phi$ be a
space of test functions, then
\begin{equation}
\begin{aligned}
\left\Vert (1-Tay_2)F(\phi)\right\Vert _{T_{\phi}(\Phi)}\leq\left(1+\Vert\phi\Vert_{\Phi}\right)^{3}\sup_{\substack{t\in[0,1]\\
k=3,4
}
}\left\Vert F^{(k)}(t\phi)\right\Vert _{T_{t\phi}^{k}(\Phi)} \;.
\end{aligned}
\end{equation}
\end{lem}

\begin{proof}
By Taylor remainder theorem,  with $f^{\times n}:= (f_1,...,f_n)$,
\begin{equation}
\begin{aligned} 
& \left\Vert (1-Tay_{2})F(\phi)\right\Vert _{T_{\phi}(\Phi)}
=\sum_{n=0}^{4}\frac{1}{n!}
	\sup_{\left\Vert f_{i}\right\Vert _{\Phi}\leq1}
	\left|\left(F-Tay_{2}F\right)^{(n)}(\phi;f^{\times n})\right|\\
= & \sum_{n=0}^{4}\frac{1}{n!}
	\sup_{\left\Vert f_{i}\right\Vert _{\Phi}\leq1}
	\left|\left(F^{(n)}-Tay_{2-n}(F^{(n)})\right)(\phi;f^{\times n})\right|\\
\end{aligned}
\end{equation}
and the absolute valued quantity is equal to
\begin{equation}
\begin{aligned} 
&\bigg|1_{\{n<3\}}\int_{0}^{1}\frac{(1-t)^{2-n}}{(2-n)!}\partial_{t}^{3-n}F^{(n)}(t\phi;f^{\times n})+1_{\{n\geq3\}}F^{(n)}(\phi;f^{\times n})\bigg|\\
&= 
\bigg|1_{\{n<3\}}\int_{0}^{1}\frac{(1-t)^{2-n}}{(2-n)!}F^{(3)}(t\phi;\phi^{\times(3-n)},f^{\times n})+1_{\{n\geq3\}}F^{(n)}(\phi;f^{\times n})\bigg|
\end{aligned}
\end{equation}
where $\phi^{\times(3-n)}$ means $3-n$ test functions $\phi$. Calculating
the t integrals,
\begin{equation}
\begin{aligned}
\Vert(1 & -Tay_{2})F(\phi)\Vert_{T_{\phi}(\Phi)}\\
 & \leq\sum_{n=0}^{3}\frac{1}{n!}
 	\sup_{
		\left\Vert f_{i}\right\Vert _{\Phi}\leq1
		}
	\bigg|\frac{1}{(3-n)!}\sup_{t\in[0,1]}F^{(3)}(t\phi;\phi^{\times(3-n)},f^{\times n})\bigg|+\Vert F^{(4)}(\phi)\Vert_{T_{\phi}^{4}(\Phi)}\\
 & \leq\left(1+\Vert\phi\Vert_{\Phi}\right)^{3}\sup_{t\in(0,1),k=3,4}\Vert F^{(k)}(t\phi)\Vert_{T_{t\phi}^{k}(\Phi)}
\end{aligned}
\end{equation}
where in the last step binomial theorem is applied.
\end{proof}

\begin{lem}
\label{lem:integrab_poly}
Let $(B,X)\in\hat{\mathcal{S}}_{j}$, $\bar{B}=U$,
if $\kappa$ is small enough depending on $L$, and $h$ is large
enough depending on $\kappa$ and $L$, then
\begin{equation}
\left(2+\Vert\phi\Vert_{\Phi_{j+1}(\dot{X},U^{+})}\right)^{3}G(\ddot{X},U^{+})\leq qG(\ddot{U},U^{+})\label{eq:integrab_poly}
\end{equation}
for a constant $q$, where the dot(s) operations on $X$ are at scale
$j$, and the dots and $+$ operations on $U$ are at scale $j+1$.
\end{lem}

\begin{proof}
Let $\psi=P_{U^{+}}\phi$.
For each $e\in\mathcal{E}$, $\partial_{e}\psi$ is harmonic in
$U^{+}\cap(U^{+}-e)$. 
Since $X,U$ are $j$ and $j+1$ scale small sets respectively
and $d(X,\partial \ddot{U})=O(L^{j+1})$,
 we can find a set $Y\subset\ddot{U}$,
such that: 1) $Y$ is of the form $\mathcal{K}_{R}\backslash\mathcal{K}_{rR}$
for some $r\in(0,\frac{1}{2})$
as in Lemma~\ref{lem:derbnd};
2) $Y\cap \ddot{X} = \emptyset$ and $d(\ddot{X},Y)= O(L^{j})$; 
3) $d(Y,\partial\ddot{U})=O(L^{j+1})$;
4) $R=\mbox{diam}(Y)=O(L^j)$.

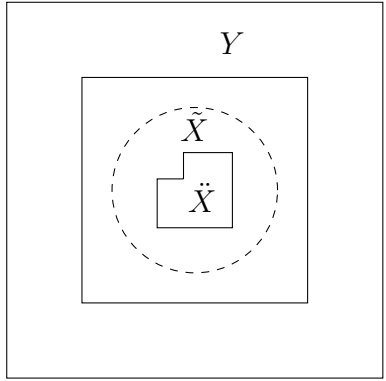
\begin{center}
\begin{tikzpicture} [scale=0.5]
\draw (-5,-5) rectangle (5,5);
\draw (-3,-3) rectangle (3,3);
\draw (-0.3,0.3) -- (-0.3,1) -- (1,1) --(1,-1)--(-1,-1)--(-1,0.3)--(-0.3,0.3);
\node at (0.2,-0.2) {$\ddot{X}$};
\node at (1,3.9) {$Y$};
\draw [dashed] (0,0) circle (2.2);
\node at (0,1.7) {$\tilde X$};
\end{tikzpicture}
\end{center}

Then by \eqref{eq:average2} of Lemma~\ref{lem:derbnd},
\begin{equation}
 \sup_{e\in\mathcal{E},x\in\dot{X}}\left|\partial_{e}\psi(x)\right|^{2}\leq O(L^{-dj})\sum_{e\in E(Y)}\left(\partial_{e}\psi \right)^{2} \;.
\end{equation}
By definition of the norm 
$\Vert\phi\Vert_{\Phi_{j+1}(\dot{X},U^{+})}^2 =h^{-2}L^{d(j+1)}
\sup_{e,x\in\dot{X}} \left|\partial_{e}\psi(x)\right|^{2}$,
if we choose $h$ large enough such that $h^{-1}O(L^d)\le 1$, then
\begin{equation}
\Vert\phi\Vert_{\Phi_{j+1}(\dot{X},U^{+})}^{2}
\leq h^{-1}\sum_{e\in E(Y)}\left(\partial_{e}\psi \right)^{2} \;.
\end{equation}
Since there exists a $q\ge 1$
such that for all $s\ge 0$,
$(2+s)^3\le q e^{s^2/2}$, one has
\begin{equation} \label{eq:existsq-1}
\left(2+\Vert\phi\Vert_{\Phi_{j+1}(\dot{X},U^{+})}\right)^{3}
\leq q
\exp\Big(\frac{h^{-1}}{2}\sum_{e\in E(Y)}\left(\partial_{e}\psi \right)^{2}\Big) \;.
\end{equation}
Apply \eqref{eq:alt-rep-G} of Lemma \ref{lem:simproperties} 
to $G(\ddot{X},U^{+})$,
and use the fact that $\psi=P_{U^+}\phi$
together with \eqref{eq:existsq-1}, then the left hand side of (\ref{eq:integrab_poly}) is bounded by
\[
q\exp\bigg\{ \frac{\kappa}{2}\sum_{e\in E(U^{+})}(a_{e}\partial_{e}\psi_{1})^{2}+\frac{h^{-1}}{2}\sum_{e\in E(Y)}\left(\partial_{e}\psi \right)^{2}-\frac{1}{2}\sum_{U^{+}}(\partial\psi_{1})^{2}+\frac{1}{2}\sum_{U^{+}}(\partial\psi)^{2}\bigg\} 
\]
where the function $a_{e}=1$ if $e\in E(\ddot{X})$ and decays to
zero in a neighborhood of $\ddot{X}$, and the support of $a_{e}$,
that is, $\tilde{X}:=\mbox{supp}(a)=\{x:\exists\bar{e}\in\mathcal{E}\mbox{ s.t. }a_{x,x+\bar{e}}\neq0\}$,
still satisfies $d(\tilde X,Y)=O(L^{j})$,
and $|\nabla^{k}a_{e}|\leq O(L^{-kj})$ for $k=0,...,3$,
and finally
\begin{equation} \label{eq:non-const-prob} 
\psi_{1}\mbox{ maximizes }\quad
\kappa \!\!\!\!\! \sum_{e\in E(U^{+})}(a_{e}\partial_{e}\phi)^{2}-\sum_{U^{+}}(\partial\phi)^{2}\mbox{ fixing }\phi\big|_{(U^{+})^{c}} \;.
\end{equation}
Notice that  applying \eqref{eq:alt-rep-G} of Lemma \ref{lem:simproperties} to $G(\ddot{X},U^{+})$
results in a term $\frac{\kappa}{2}$ times a Dirichlet form over $\ddot{X}$, and
we ``enlarged'' the set $\ddot{X}$ to $\tilde X$ 
by smoothing out the coefficient $a_{e}$, followed by a replacement of that Dirichlet form
with that of the maximizer $\psi_{1}$ solving the new elliptic problem \eqref{eq:non-const-prob} - this only makes the above exponential larger. 
In
the following we show that by choosing $h$ large enough one has
\begin{equation} \label{eq:polyReg_crucial}
\frac{h^{-1}}{2}\sum_{e\in E(Y)}\left(\partial_{e}\psi\right)^{2}
\leq
\frac{\kappa}{2}\sum_{e\in E(Y)}\left(\partial_{e}\psi_1\right)^{2} \;.
\end{equation}
Then the left hand side of (\ref{eq:integrab_poly}) is bounded
by
\[
q\exp\bigg\{ 
	\frac{\kappa}{2}\sum_{e\in E(\ddot{U})}(\partial_{e}\bar\psi)^{2}
	-\frac{1}{2}\sum_{U^{+}}(\partial\bar \psi)^{2}
	+\frac{1}{2}\sum_{U^{+}}(\partial\psi)^{2}
\bigg\} = qG(\ddot{U},U^{+})
\]
which holds by replacing $\psi_{1}$ by $\bar\psi$ which is the maximizer of
$\frac{\kappa}{2}\sum_{e\in E(\ddot{U})}(\partial_{e}\phi)^{2}-\frac{1}{2}\sum_{U^{+}}(\partial\phi)^{2}$ with $\phi\big|_{(U^{+})^{c}}$ fixed.

To show (\ref{eq:polyReg_crucial}), 
let $\bar{a}=1-\kappa  a$.
We have
\[
\begin{cases} \Delta\psi =0 \quad &\mbox{in }U^+  \\ 
	\psi =\phi \quad &\mbox{on }\partial U^+ \end{cases}
\qquad\qquad\qquad
\begin{cases} 
	\Delta_{\bar a}\psi_1 =0 \quad &\mbox{in }U^+  \\ 
	\psi_1 =\phi \quad &\mbox{on }\partial U^+ \end{cases}
\]
where $\Delta_{\bar a} f(x) = \sum_e \bar a_e (f(x+e)-f(x))$.
Subtract them and we obtain a non-constant coefficient elliptic problem for $\psi_1-\psi$
\[
\begin{cases}
\Delta_{\bar{a}} (\psi_1-\psi)  = \kappa \Delta_{a} \psi & \qquad\mbox{in }U^{+}
\\
\psi_1-\psi_0 =0 & \qquad \mbox{on }\partial U^{+}
\end{cases}
\]
One has the following representation of derivative of the solution to the above equation
(note that the support of $a$ is $\tilde X$ so $\Delta_a \psi=0$ outside $\tilde X$)
\begin{equation}\label{eq:elli-prob-diff}
\partial_e (\psi_1-  \psi)(y)
= \kappa\sum_{x\in \tilde{X}}\partial_{y,e} G_{\bar a}(y,x) \Delta_a \psi(x)
\end{equation}
for $y\in Y,e\in \mathcal E$, where  $G_{\bar a}$ is the Dirichlet Green's function associated with $\Delta_{\bar a}$.

Our situation is that for a Laplacian with non-constant coefficient $\Delta_{\bar a}$, although 
one has desired bound for 
the Green's function $G_{\bar a}$ (i.e. bound with the decay rate as if the Laplacian
was a constant coefficient one), the desired bound for $\partial_y G_{\bar a}(y,x)$
does not hold in general. However, we do have 
 bound with desired scaling in an averaging sense, i.e. 
after a summation over $y$ - the variable w.r.t. which $G_{\bar a}$
is differentiated. Consider
\[
\begin{aligned}
& \sum_{e\in E(Y)}  \Big(\partial_e (\psi_1-\psi)\Big)^2 
= \kappa^2  \sum_{e\in E(Y)} 
\Big( \sum_{x\in \tilde{X}} \partial_{y,e} G_{\bar a}(y,x) \Delta_a \psi(x) \Big)^2 \\
& \qquad =  \kappa^2  \!\!\! \sum_{x_1,x_2\in \tilde{X}} \Delta_a \psi(x_1) \Delta_a \psi(x_2)
	\sum_{e\in E(Y)} \partial_{y,e} G_{\bar a}(y,x_1) \partial_{y,e} G_{\bar a}(y,x_2)\\
&\leq \frac{\kappa^2}{2}  \!\!\!
 \sum_{x_1,x_2\in \tilde{X}} \Big|  \Delta_a \psi(x_1) \Delta_a \psi(x_2) \Big|
	\sum_{e\in E(Y)} \!\!\! \Big(\big(\partial_{y,e} G_{\bar a}(y,x_1)\big)^2+\big( \partial_{y,e} G_{\bar a}(y,x_2)\big)^2\Big) .
\end{aligned}
\]
With this bound at hand, our proof of \eqref{eq:polyReg_crucial} now follows from two claims. The first claim is that for every $x\in \tilde X$,
\begin{equation} \label{eq:first-claim}
\sum_{e\in E(Y)} \big(\partial_{y,e} G_{\bar a}(y,x) \big)^2
\le O(L^{-2j})  \sum_{y \in \tilde Y}G_{\bar a}(y,x)^2 
\le O(L^{-(d-2)j}) 
\end{equation}
where $\tilde Y$ is such that $Y\subset \tilde Y$, $d(Y,\tilde Y^c)=O(L^j)$ and $d(\tilde Y,\tilde X)=O(L^j)$.
Note that the last inequality follows from $G_{\bar a}(y,x) \le O(L^{-(d-2)j})$ 
(this is a standard bound for Green's function of non-constant coefficient Laplacian, see for instance \cite{delmotte_parabolic_1999}) and $|Y|=O(L^{dj})$.
Note that the right side of \eqref{eq:first-claim} does not depend on $x_1,x_2$,
so it remains to bound 
$\Big(\sum_{x\in \tilde{X}} \left|\Delta_a \psi(x)\right|\Big)^2$.

The second claim is  that for every $x\in\tilde X$, 
\begin{equation} \label{eq:second-claim}
|\Delta_a \psi(x) |
\leq O(L^{-\frac{d+2}{2}j}) \Big(\sum_Y |\nabla\psi|^2\Big)^{1/2}
\end{equation}
so that one has
\begin{equation}\label{eq:second-claim-1}
\Big(\sum_{x\in \tilde{X}} \left|\Delta_a \psi(x)\right|\Big)^2
\le O\Big( (L^{dj}\cdot L^{-\frac{d+2}{2}j})^2 \Big)  \sum_Y (\nabla\psi)^2
= O(L^{(d-2)j}) \sum_Y (\nabla\psi)^2
\end{equation}
As a consequence of \eqref{eq:first-claim} and \eqref{eq:second-claim-1},
one has
\[
\frac{1}{2} \sum_Y (\nabla \psi)^2 
\le \sum_Y (\nabla \psi-\nabla\psi_1)^2+  \sum_Y(\nabla\psi_1)^2
\le O(1)\kappa^2 \sum_Y (\nabla \psi)^2+ \sum_Y(\nabla\psi_1)^2
\]
Choosing $h$ large enough such that $h^{-1}\le \kappa(1/2-O(1)\kappa^2)$
we obtain \eqref{eq:polyReg_crucial}.

The proof to the first inequality of \eqref{eq:first-claim}
is motivated by Cacciopoli's inequality in the continuum setting, 
which roughly states that for a solution $u$ to an elliptic problem 
one can bound the $L^2$ norm of $u$ by the $L^2$ norm (over a larger domain) of $\nabla u$
(as a reverse of Poincare inequality),
under certain conditions (see for instance \cite[Chapter 3]{giaquinta_multiple_1983}). 
We don't provide the proof of its discrete counterpart in full generality,
but only prove a weak version that is sufficient for our purpose.

Fixing $x\in\tilde X$, let $u(y)=G_{\bar a}(y,x)$, which is $\Delta_{\bar a}$-harmonic
in $U^+$ away from the singular point $y=x$: namely 
$\sum_{e\in\mathcal E} \bar a_e (u(y+e)-u(y)) =0$ for $y\in U^+ \backslash \{x\}$. 
Since $\kappa$ is small, the function $ \bar a_e$ is such that
there exist $0<\lambda<\Lambda$ and $\lambda<\bar a_e<\Lambda $.
Then, for every function $v$ on $\tilde Y$, one has 
\[
\sum_{e\in E(\tilde Y)} \bar a_e \, \partial_e u \,\partial_e v =0\;.
\]
Let $v=u\varphi^2$ for some non-negative function $\varphi$ supported on $\tilde Y$,
then $\partial_e v =\partial_e u\cdot \varphi^2 + 2\varphi \partial_e \varphi\cdot u$.
Substituting this into the identity above, one has
\[
\begin{aligned}
\lambda \sum_{y,y+e\in \tilde Y} & \varphi(y)^2  (\partial_e u(y))^2
\le -\!\!\!\!\!\! \sum_{y,y+e\in \tilde Y} 2\varphi(y)u(y)\bar a_{(y,y+e)} \partial_e u(x)\partial_e\varphi(y) \\
& \le \frac{\lambda}{2} \sum_{y,y+e\in \tilde Y} \varphi(y)^2 (\partial_e u(y))^2
+ \frac{2}{\lambda}  \sum_{y,y+e\in \tilde Y}\bar a_{(y,y+e)}^2 (\partial_e \varphi(y))^2 u(y)^2
\end{aligned}
\]
where the first inequality used $\bar a>\lambda$ and the second inequality is
by Cauchy-Schwartz. 
Therefore,
\[
\frac{\lambda}{2} \sum_{y,y+e\in \tilde Y} \varphi(y)^2 (\partial_e u(y))^2
\le 
\frac{2\Lambda^2}{\lambda}  \sum_{y,y+e\in \tilde Y} (\partial_e \varphi(y))^2 u(y)^2
\]
Choosing $\varphi=1$ on $Y$ and $|\nabla \varphi|\le O(L^{-j})$ we obtain
the first inequality of \eqref{eq:first-claim}.

The proof of \eqref{eq:second-claim} is 
based on the idea of writing $\Delta_a \psi$ 
in terms of (derivatives of) $a$ 
and constant coefficient derivatives of $\psi$,
in a way analogous to
the relation $\nabla\cdot(a\nabla f)=\nabla a\cdot\nabla f+a\Delta f$
in continuum.
Note that $a_e$ above is defined on edges $e$.
For a lattice site $x$, define $a(x)=(2d)^{-1} \sum_e a_{(x,x+e)}$.
Then 
\[
\begin{aligned}
|\Delta_a\psi(x)| &=|\sum_{e\in\mathcal E} a_e (\psi(x+e)-\psi(x)) | \\
&\le
|\sum_{e\in\mathcal E} (a_{(x,x+e)}-a(x)+a(x)) (\psi(x+e)-\psi(x)) |\\
&\leq 
\sup_{e\in\mathcal E} |a_{(x,x+e)} -a(x)| |\nabla\psi(x)|+ |a(x)| |\Delta\psi(x) | 
\end{aligned}
\]
Note that the last term is zero since $\Delta\psi=0$.
The term $|a_{(x,x+e)} -a(x)|$ is 
bounded by $(2d)^{-1} \sum_{e'\in\mathcal E} | a_{(x,x+e)} -a_{(x,x+e')}|$
which by the choice of $a$ is bounded by $O(L^{-j})$. 
Lemma~\ref{lem:derbnd} allows us to bound 
\[
|\nabla\psi(x)| \le O(L^{-dj/2}) (\sum_Y (\nabla\psi)^2)^{1/2} \;.
\]
Therefore we obtain \eqref{eq:second-claim}.
So \eqref{eq:polyReg_crucial} is shown and the proof of the lemma is completed.
\end{proof}
Before the next Lemma we define 
\begin{equation} \label{eq:def-F_X}
F_{X}(U,\phi,\xi):=\mathbb{E}\left[K_{j}(X,\phi,\xi)\big|\left(U^{+}\right)^{c}\right] \;.
\end{equation}
It depends on $\phi,\xi$ via $\psi:=P_{U^{+}}\phi+\xi$,
i.e. there exists a function $\widetilde{F}_{X}$ such that $F_{X}(U,\phi,\xi)=\widetilde{F}_{X}(U,\psi)$. 

\begin{lem}
\label{lem:L2}
Let $L$ be sufficiently large. Then $\mathcal{L}_{2}$
in Proposition \ref{prop:The-linearization} is a contraction
with the norm going to zero as $L\to\infty$.
\end{lem}
\begin{proof}
Let $Tay$ be the second order Taylor expansion in $\psi=P_{U^{+}}\phi+\xi$.
With the the $F_X$ defined in \eqref{eq:def-F_X}, we aim to bound
\begin{equation} \label{eq:lem31-aim}
\big\Vert(1-Tay)  F_{X}(U,\phi,\xi) \big\Vert_{T_{\phi}(\Phi_{j+1}(U))} \;.
\end{equation}
Recall that $\Phi_{j+1}(U)$ is short for $\Phi_{j+1}(\dot U,U^+)$
and by Lemma \ref{lem:Tphiproperties} this can be replaced by
$\Phi_{j+1}(\dot X,U^+)$.
Applying Lemma \ref{lem:3rdder} with
the test function space $\Phi := \widetilde{\Phi}_{j+1}(\dot{X},U^{+})$, we can
bound \eqref{eq:lem31-aim} by
\begin{equation} \label{eq:-3d/2}
  \big\Vert(1-Tay)\widetilde{F}_{X}(U,\psi)\big\Vert_{T_{\psi}(\Phi)}
  \leq
 \Big(1+\Vert\psi\Vert_{\Phi}\Big)^{3}
 \sup_{k=3,4}
 \big\Vert\widetilde{F}_{X}^{(k)}(U,\psi)\big\Vert_{T_{\psi}^{k}(\Phi)}
\end{equation}
Now by linearity of $\widetilde{F}_{X}^{(k)}$
in test functions, 
\begin{equation}
\begin{aligned}
\big\Vert  & \widetilde{F}_{X}^{(3)}  (U,\psi)\big\Vert_{T_{\psi}^{3}(\widetilde{\Phi}_{j+1}(\dot{X},U^{+}))}
\leq L^{-\frac{3}{2}d}
\big\Vert \widetilde{F}_{X}^{(3)}(U,\psi)\big\Vert _{T_{\psi}^{3}(\widetilde{\Phi}_{j}(\dot{X},U^{+}))}\\
 & \leq L^{-\frac{3}{2}d}\cdot 3!\cdot \mathbb{E}\left[\Vert K_{j}(X,\phi,\xi)\Vert_{T_{\phi,\xi}(\Phi_{j}(X))}\big|(U^{+})^{c}\right]\\
 & \leq O(L^{-\frac{3}{2}d})\, \Vert K_{j}(X)\Vert_{j} \, c^{|X|_{j}} \,G(\ddot{X},U^{+})
\end{aligned}
\label{eq:Ftilde3-1}
\end{equation}
where in the last step Lemma \ref{lem:intproperties} is applied.
We can prove analogously that $\tilde{F}_{X}^{(4)}(U,\psi)$ satisfies a similar bound with a factor $O(L^{-2d})$.
Next we estimate
\begin{equation} \label{eq:estimate-psiPhi}
\begin{aligned}
\Vert\psi\Vert_{\Phi} 
& \leq h_{j}^{-1}\sup_{x\in\dot{X},e}
	\big|L^{j}\partial_{e}P_{U^{+}}\phi(x)\big|
+h_{j}^{-1}\sup_{x\in\dot{X},e}
	\big|L^{j}\partial_{e}\xi(x)\big|\\
 & \leq\Vert\phi\Vert_{\Phi_{j+1}(\dot{X},U^{+})}+1
\end{aligned}
\end{equation}
by (\ref{eq:smallness_xi}). 
Combining \eqref{eq:-3d/2}--
\eqref{eq:estimate-psiPhi},
and applying Lemma \ref{lem:integrab_poly}, followed by (\ref{enu:reg4}) of Lemma
\ref{lem:simproperties}, one obtains
\begin{equation}
\left\Vert (1-Tay)F_{X}(U)\right\Vert _{j+1}
\leq O(L^{-\frac{3d}{2}})\,c^{|X|_{j}}\Vert K_{j}(X)\Vert_{j}\leq O(L^{-\frac{3d}{2}})(\frac{A}{c})^{-|X|_{j}}\Vert K_j\Vert_{j} .
\end{equation}
Note that the sum over $B$ and $X$ in the definition \eqref{eq:def-CL2} of $\mathcal L_2$ gives a factor $O(L^d)$.
Apply the geometric Lemma \ref{lem:geometric} to $|X|_j$, one then has
\[
\begin{aligned}
\left\Vert \mathcal{L}_{2}K_{j}\right\Vert _{j+1} 
 & \leq O(L^{-3d/2})\bigg[\sup_{U\in\mathcal{P}_{j+1}}A^{|U|_{j+1}}O(L^{d})A^{-|U|_{j+1}}c^{2^{d}}\bigg]\Vert K\Vert_{j}\\
 & =O(L^{-d/2})\Vert K_j\Vert_{j}  \;.
\end{aligned}
\]
As $L\to \infty$, the factor $L^{-d/2}$ overwhelms the constants 
hidden in the big-O notation, and therefore $\mathcal L_2$
has arbitrarily small norm.
\end{proof}

\subsection{$\mathcal{L}_{3}$ and determination of coupling constants}

We now localize the last term in $\mathcal{L}_{3}$, which is the
second order Taylor expansion of $\widetilde{F}_{X}(U,\psi)$ in $\psi$
(which are introduced in \eqref{eq:def-F_X}). To do this we fix a point $z\in B$, and replace
$\psi(x)$ by $x\cdot\partial\psi(z)$ (which according to our convention
means $\frac{1}{2}\sum_{e\in\mathcal{E}}x_{e}\partial_{e}\psi(z)$),
and then average over $z\in B$. We will show that the error of this
replacement is irrelevant. Then 
\[
\frac{1}{2}\tilde{F}_{X}^{(2)}(U,0;\psi,\psi)=\text{Loc}K_{j}(B,X,U)+(1-\text{Loc})K_{j}(B,X,U)
\]
where we have defined
\[
\begin{aligned}
\text{Loc}\, &K_{j}(B,X,U)\\
&:=  \frac{1}{8|B|}\sum_{\substack{z\in B,\mu,\nu\in\mathcal{E}}
}\partial_{t_{1}t_{2}}^{2}\bigg|_{t_{i}=0}\mathbb{E}_{\zeta}\left[K_{j}(X,t_{1}x_{\mu}+t_{2}x_{\nu}+\zeta)\right]\partial_{\mu}\psi(z)\partial_{\nu}\psi(z)
\end{aligned}
\]
and
\begin{equation}
\begin{aligned}
&(1-\text{Loc})  K_{j}(B,X,U):=\frac{1}{2|B|}\sum_{\substack{z\in B}}\Big(\partial_{t_{1}t_{2}}^{2}\bigg|_{t_{i}=0}\mathbb{E}_{\zeta}\left[K_{j}(X,t_{1}\psi+t_{2}\psi+\zeta)\right]\\
 &\qquad
  -\partial_{t_{1}t_{2}}^{2}\bigg|_{t_{i}=0}\mathbb{E}_{\zeta}\left[K(X,t_{1}x\cdot\partial\psi(z)+t_{2}x\cdot\partial\psi(z)+\zeta)\right]\Big)\\
 & =
 \frac{1}{2|B|}\sum_{\substack{z\in B}
}\Big(\widetilde{F}_{X}^{(2)}(U,0;\psi-x\cdot\partial\psi(z),\psi)+\widetilde{F}_{X}^{(2)}(U,0;\psi-x\cdot\partial\psi(z),x\cdot\partial\psi(z))\Big) .
\end{aligned}
\label{eq:1-loc}
\end{equation}
We show that $\psi-x\cdot\partial\psi(z)$ gives additional contractive
factors as going to the next scale:

\begin{lem} \label{lem:psi-xdpsi}
If $\psi=P_{U^{+}}\phi+\xi\in\widetilde{\Phi}_{j}(\dot{X},U^{+})$, then
\begin{equation}\label{eq:lemma28}
\left\Vert \psi-x\cdot\partial\psi(z)\right\Vert _{\widetilde{\Phi}_{j}(\dot{X},U^{+})}\leq O(L^{-\frac{d}{2}-1})\left(\left\Vert \phi\right\Vert _{\Phi_{j+1}(U)}+1\right) \;.
\end{equation}
\end{lem}

\begin{proof}
Since $P_{U^{+}}x=x$, the left side of \eqref{eq:lemma28} is equal to
\begin{equation}
h_{j}^{-1}\sup_{x\in\dot{X},e}L^{j}\bigg|\partial_{e}P_{U^{+}}\phi(x)+\partial_{e}\xi(x)-\partial_{e}P_{U^{+}}\phi(z)-\partial_{e}\xi(z)\bigg| \;.
\label{eq:ArgOfx-1}
\end{equation}
We apply Newton-Leibniz formula along a curve connecting
$x,z$, and then apply \eqref{eq:diff-est-any} with $R=O(L^{j+1})$ using
the distance $O(L^{j+1})$ between $\dot{X}$ and $\partial\dot{U}$,
\[
\begin{aligned}
h_{j}^{-1} &\sup{}_{x\in\dot{X},e} L^{j}  \left|\partial_{e}P_{U^{+}}\phi(x)-\partial_{e}P_{U^{+}}\phi(z)\right|\\
 & \leq h_{j}^{-1}\sup{}_{x\in\dot{U}}L^{j}\,\mbox{diam}\,(\dot{X})\,O(L^{-j-1})\left|\partial P_{U^{+}}\phi(x)\right|\\
 & \leq O(L^{-\frac{d+2}{2}})\left\Vert \phi\right\Vert _{\Phi_{j+1}(U)}
\end{aligned}
\]
where $\mbox{diam}(\dot{X})=O(L^{j})$ since $X$ is small. The second
term in (\ref{eq:ArgOfx-1}) can be bounded by
\[
h_{j}^{-1}\sup_{x\in\dot{X},e}L^{j}\left|\partial_{e}\xi(x)-\partial_{e}\xi(z)\right|\leq O(L^{-\frac{d}{2}(N-j)})\leq O(L^{-\frac{d+2}{2}})
\]
as long as $j+1<N$, and by $d\geq2$ and (\ref{eq:smallness_xi}).
Combining the above bounds completes the proof.
\end{proof}

\begin{lem}
\label{lem:nonlocal}
If $L$ be sufficiently large and define
\begin{equation} \label{eq:def-CL3prime}
\mathcal{L}_{3}^{\prime}K_{j}(U)
=\sum_{\bar{B}=U}\sum_{X\in\mathcal{S}_{j},X\supseteq B}
	(1-\text{Loc})K_{j}(B,X,U)
\end{equation}
then $\mathcal{L}_{3}^{\prime}$ is contractive with arbitrarily small
norm; namely, $\left\Vert \mathcal{L}_{3}^{\prime}\right\Vert \rightarrow0$
as $L\rightarrow\infty$.\end{lem}
\begin{proof}
In view of the definition \eqref{eq:1-loc} of $(1-\text{Loc})K_{j}$, we let
\begin{equation}
H_{z,X}(U,\phi,\xi)=\tilde{F}_{X}^{(2)}(U,0;\psi-x\cdot\partial\psi(z),\psi)
\end{equation}
then with $\tilde{f}:=P_{U^{+}}f+\lambda\xi$,
\begin{equation}
\begin{aligned}
H_{z,X}^{(1)} & (U,\phi,\xi;  (f,\lambda\xi))\\
 & =\tilde{F}_{X}^{(2)}(U,0;\psi-x\cdot\partial\psi(z),\tilde{f})
 + \tilde{F}_{X}^{(2)}(U,0;\tilde{f}-x\cdot\partial\tilde{f}(z),\psi)\;;\\
H_{z,X}^{(2)} & (U,\phi,\xi;  (f_{1},\lambda_{1}\xi),(f_{2},\lambda_{2}\xi))\\
 & =\tilde{F}_{X}^{(2)}(U,0;\tilde{f}_{1}-x\cdot\partial\tilde{f}_{1}(z),\tilde{f}_{2})+\tilde{F}_{X}^{(2)}(U,0;\tilde{f}_{2}-x\cdot\partial\tilde{f}_{2}(z),\tilde{f}_{1})
\end{aligned}
\label{eq:devsofH}
\end{equation}
and $H_{z,X}^{(3)}=0$. In the calculations here, though $z$ is fixed,
$P_{U^{+}}\phi(z)$ should also participate in the differentiations:
$P_{U^{+}}\phi(z)\rightarrow P_{U^{+}}(\phi+tf)(z)$. 

We now bound the all the test functions appeared in \eqref{eq:devsofH}. 
The bound for $\psi-x\cdot\partial\psi(z)$ is given in Lemma~\ref{lem:psi-xdpsi}.
Similarly one can bound $\tilde f-x\cdot\partial\tilde f(z)$
by $O(L^{-\frac{d}{2}-1})\big\Vert(f,\lambda\xi)\big\Vert_{\Phi_{j+1}(U)}$.
Since $\left\Vert -\right\Vert _{\Phi_{j}(U)}
\leq L^{-d/2}\left\Vert -\right\Vert _{\Phi_{j+1}(U)}$
we also have estimates 
\begin{equation}
\begin{aligned}
\| \psi \| _{\Phi_{j}(\dot{X},U^{+})} 
& \leq O(L^{-d/2})\big(\| \phi \|_{\Phi_{j+1}(U)}+1\big)\;;\\
\| \tilde f \|_{\Phi_{j}(\dot{X},U^{+})} 
& \leq O(L^{-d/2}) \| (f,\lambda \xi) \|_{\Phi_{j+1}(U)} \;.
\end{aligned}
\label{eq:psiandPUf}
\end{equation}
Therefore we obtain the bound
\begin{equation}
\begin{aligned}\bigg|H_{z,X}^{(n)}(U,\phi,\xi; & (f,\lambda\xi)^{\times n})\bigg|\leq O(L^{-d-1})\big\Vert\tilde{F}_{X}^{(2)}(U,0)\big\Vert_{T_{0}^{2}(\tilde{\Phi}_{j}(\dot{X},U^{+}))}\\
 & \cdot\left(\left\Vert \phi\right\Vert _{\Phi_{j+1}(U)}+1\right)^{2-n}\prod_{i=1}^{n}\big\Vert(f_{i},\lambda_{i}\xi)\big\Vert_{\Phi_{j+1}(U)} \;.
\end{aligned}
\end{equation}
By the same arguments as (\ref{eq:existsq-1}) and Lemma \ref{lem:simproperties}(5), one can bound $(1+\left\Vert \phi\right\Vert _{\Phi_{j+1}(U)})^{2}$
by $G(\ddot{U},U^{+})$. Therefore,
\begin{equation}
\big\Vert  H_{z,X}(U,\phi,\xi)  \big\Vert_{T_{\phi}(\Phi_{j+1}(U))}
  \leq O(L^{-d-1})\big\Vert\widetilde{F}_{X}^{(2)}(U,0)
  	\big\Vert_{T_{\phi}^{2}(\widetilde{\Phi}_{j}(\dot{X},U^{+}))}G(\ddot{U},U^{+}) \;.
\end{equation}
By 
Lemma \ref{lem:intproperties} followed by
Lemma \ref{lem:simproperties}(\ref{enu:reg1}), together with $X\in\mathcal{S}_{j}$
\begin{equation}
\begin{aligned}\big\Vert\tilde{F}_{X}^{(2)}(U,0) & \big\Vert_{T_{\phi}^{2}(\tilde{\Phi}_{j}(\dot{X},U^{+}))}\leq\mathbb{E}\left[\left\Vert K_{j}(X,\phi,\xi=0)\right\Vert _{T_{\phi}(\Phi_{j}(\dot{X},U^{+}))}\big|\phi_{(U^{+})^{c}}=0\right]\\
 & \leq
 \mathbb{E}\left[\left\Vert K_{j}(X)\right\Vert _{j}G(\ddot{X},X^{+})\big|\phi_{(U^{+})^{c}}=0\right]\leq\left\Vert K_{j}(X)\right\Vert _{j}c^{|X|_{j}}\\
 & \leq O(1)A^{-1}\left\Vert K_{j}\right\Vert _{j} \;.
\end{aligned}
\end{equation}
Combining the above inequalities, we obtain
\[
\left\Vert H_{z,X}(U)\right\Vert _{j+1}\leq O(L^{-d-1})A^{-1}\left\Vert K\right\Vert _{j} \;.
\]
It can be shown analogously that the other term on the right side of (\ref{eq:1-loc}) satisfies the same bound.
Finally, the sum over $B$ and $X$ in \eqref{eq:def-CL3prime} gives a factor $O(L^d)$, so one has
\begin{equation}
\left\Vert \mathcal{L}_{3}^{\prime}K(U)\right\Vert _{j+1}
\leq 
O(L^{-1})A^{-1}\left\Vert K\right\Vert _{j} \;.
\end{equation}
Since $\mathcal{L}_{3}^{\prime}K_{j}(U)=0$ unless $U$ is a block,
$\left\Vert \mathcal{L}_{3}^{\prime}K_{j}\right\Vert _{j+1}\leq O(L^{-1})\left\Vert K\right\Vert _{j}$.
\end{proof}
Now we turn to $\text{Loc}\,K_{j}$. We observe that the coefficient of $\partial_{\mu}\psi(z)\partial_{\nu}\psi(z)$
is 
\begin{equation} \label{eq:Bdependent_alpha}
\alpha_{\mu\nu}(B):=\frac{1}{8|B|}\sum_{X\in\mathcal{S}_{j},X\supseteq B}\partial_{t_{1}t_{2}}^{2}\bigg|_{t_{i}=0}\mathbb{E}_{\zeta}\left[K_{j}(X,t_{1}x_{\mu}+t_{2}x_{\nu}+\zeta)\right] \;.
\end{equation}
Note that each summand
above is just derivative of $\mathbb{E}_{\zeta} K_j (X)$ at zero field with test functions $x_\mu$ and $x_\nu$. Since
 $\| x_{\mu}\|_{\Phi_{j}(X)}\leq h^{-1}L^{dj/2}$ (for this one needs the fact that the Poisson kernel in the definition of $\Phi_{j}$ norm acting on $x_\mu$ still gives $x_\mu$),
we have 
\begin{equation}\label{eq:bound-alphamn}
|\alpha_{\mu\nu}(B)|\leq O(1)h^{-2}\| K_{j}\|_{j}A^{-1} \;.
\end{equation}

Note that for a fixed $D\in\mathcal{B}_{j+1}$, and for all $\bar{B}=D$,
$\alpha_{\mu\nu}(B)$ depends on the position of $B$ in $D$ because
$\zeta$ is not translation invariant. This problem was not present
in the method \cite{brydges_lectures_2007}. We cure this problem
by the following lemma.
\begin{lem}
\label{lem:alpha_nontransinv}
Let $D\in\mathcal{B}_{j+1}$, and let
$B_{ct}\in\mathcal{B}_{j}$ be the $j$-block at the center of $D$.
Then with definition (\ref{eq:Bdependent_alpha}), 
\begin{equation}
\left|\alpha_{\mu\nu}(B)-\alpha_{\mu\nu}(B_{ct})\right|\leq O(L^{-d})h^{-4}\| K_{j}\|_{j}A^{-1}
\end{equation}
for all $B\in\mathcal{B}_{j}$ such that $\bar{B}=D$.\end{lem}
\begin{proof}
Let $T$ be a translation so that $TB=B_{ct}$, and $\zeta_{D^{+}},\zeta_{TD^{+}}$
be Gaussian fields on $D^{+},TD^{+}$ with Dirichlet Green's functions
$C_{D^{+}},C_{TD^{+}}$ as covariances respectively. 
Then $\alpha_{\mu\nu}(B)$ 
can be rewritten as the right side of \eqref{eq:Bdependent_alpha}
with $B$ replaced by $B_{ct}$ and $\zeta=\zeta_{D^+}$ replaced by $\zeta=\zeta_{TD^+}$, so that
\begin{equation}
\begin{aligned}
\big| &\alpha_{\mu\nu} (B)  -\alpha_{\mu\nu}(B_{ct})\big|\\
 & \leq
 \frac{1}{8|B_{ct}|}\sum_{X\in\mathcal{S}_{j},X\supseteq B_{ct}}
 \bigg|\partial_{t_{1}t_{2}}^{2}\big|_{t_{i}=0}
 \Big(\mathbb{E}_{\zeta_{TD^{+}}}\left[K_{j}(X,t_{1}x_{\mu}+t_{2}x_{\nu}+\zeta_{TD^{+}})\right]\\
 & \qquad\qquad
  \:-\mathbb{E}_{\zeta_{D^{+}}}\left[K_{j}(X,t_{1}x_{\mu}+t_{2}x_{\nu}+\zeta_{D^{+}})\right]\Big)
  \bigg| \;.
\end{aligned}
\end{equation}
To estimate the difference of the two expectations, define
\[
C(t):=tC_{D^{+}}+(1-t)C_{TD^{+}}
\]
and recall that $K_{j}$ depends on $\zeta$ via $\nabla\zeta$, let
\[
\mathcal{K}(\nabla\zeta):=K_{j}(X,t_{1}x_{\mu}+t_{2}x_{\nu}+\zeta)  \;.
\]
Then, one has the formula
\[
\mathbb{E}_{\nabla^{2}C(1)}\mathcal{K}-\mathbb{E}_{\nabla^{2}C(0)}\mathcal{K}
 =\int_{0}^{1}\frac{d}{dt}\mathbb{E}_{\nabla^{2}C(t)}\mathcal{K}dt
 =\frac{1}{2}\int_{0}^{1}\mathbb{E}_{\nabla^{2}C(t)}\left[\Delta_{\nabla^{2}\dot{C}(t)}\mathcal{K}\right]dt
\]
where for any covariance $C$ (in our case $C=\nabla^{2}\dot C(t)$) the Laplacian is defined as 
\[
\Delta_{C}:=\sum_{x,y}C(x,y)\frac{\delta}{\delta\phi(x)}\frac{\delta}{\delta\phi(y)} \;.
\]

Now we aim to show a pointwise bound for $\nabla^{2}\dot{C}(t)=\nabla^{2}C_{D^{+}}-\nabla^{2}C_{TD^{+}}$.
One has
\[
 \nabla^{2}C_{\mathbb Z^d}  (x,y)  -\nabla^{2}C_{D^{+}}(x,y)
 = \nabla^{2}P_{D^{+}}C_{\mathbb Z^d}(x,y)
\]
Observe that $x,y$ have distance of $O(L^{j+1})$ from $\partial D^{+}$,
because $\mathcal K$ only depends on the field on $\partial X^+$.
We can proceed as the arguments following (\ref{eq:alpha12-1}) in
proof of Lemma \ref{lem:derbnd}, or the arguments following (\ref{eq:intprop_trace})
in proof of Lemma \ref{lem:intproperties}, to show that $\nabla^{2}P_{D^{+}}C_{\mathbb Z^d}(x,y)$ is
bounded by $O(L^{-d(j+1)})$. 
Analogously, $\nabla^{2}C_{\mathbb Z^d}   -\nabla^{2}C_{TD^{+}}$
satisfies the same bound. Therefore 
\begin{equation}\label{eq:nabla2Cdot}
|\nabla^{2}\dot{C}(t)|\leq O(L^{-d(j+1)}).
\end{equation}

Our situation is that we would like to bound the fourth derivative of $K_j$ by
$\|K_j\|_j$. This is the reason we incorporated the fourth derivative in the definition
of $\|K_j\|_j$, see \eqref{eq:allders}.
Note that $\partial/\partial\phi(x_0)$ acting on 
$\mathcal K$ is equivalent  with
\[
\partial_s |_{s=0 } K_j (X,t_1 x_\mu+t_2 x_\nu+\zeta+s\delta_{x_0})
\]
where $\delta_{x_0}$ is the Kronecker function at $x_0$.
In fact, we have $\| \delta_{x_0}\|_{\Phi_{j}(X)}\leq h^{-1}L^{-dj/2}$
because the $\partial_e P_{X^+}$ in the definition of $\Phi_{j}(X)$ norm 
acting on $\delta_{x_0}$ gives a factor $O(L^{-dj})$.
Proceeding as in \eqref{eq:bound-alphamn}, we have  $\| x_{\mu}\|_{\Phi_{j}(X)}\leq h^{-1}L^{dj/2}$, and $|B_{ct}|^{-1}=O(L^{-dj})$, and the sum $\sum_{x,y}$ gives a factor $O(L^{2dj})$. Combining these with \eqref{eq:nabla2Cdot}, we then
obtain the desired bound.
\end{proof}

Let $D\in\mathcal{B}_{j+1}$. Define $\alpha_{\mu\nu}:=\alpha_{\mu\nu}(B_{ct})$
where $B_{ct}\in\mathcal{B}_{j}$ is at the center of $D$. Clearly
it's well defined (independent of $D$). By reflection and rotation
symmetries, there exists an $\alpha$ so that $\alpha_{\mu\nu}=\frac{1}{2}\alpha(\delta_{\mu\nu}+\delta_{\mu,-\nu})$.

\begin{lem}
\label{lem:alpha_munu}
Let $\psi=P_{U^{+}}\phi+\xi$ and $L$ be sufficiently large. Then,
\begin{equation}
\mathcal{L}_{3}^{\prime\prime}:=\frac{1}{4}\sum_{\bar{B}=D}\bigg(\sum_{\substack{x\in B,e\in\mathcal{E}}
}\alpha\left(\partial_{e}\psi(x)\right)^{2}-\sum_{\substack{x\in B,e\in\mathcal{E}}
}\alpha_{\mu\nu}\left(\partial_{e}\psi(x)\right)^{2}\bigg)
\end{equation}
is contractive with norm going to zero as $L\to\infty$.
\end{lem}

\begin{proof}
This is essentially Lemma 10 of \cite{dimock_infinite_2009}, so the
proof is omitted.
\end{proof}
\begin{prop}\label{prop:L3}
We can choose $E_{j+1}$ and $\sigma_{j+1}$ so that if $L$ be sufficiently
large then $\mathcal{L}_{3}$ in Proposition \ref{prop:The-linearization}
is contractive, with arbitrarily small norm as $L\to \infty$.
\end{prop}
\begin{proof}
As the first step with $D=\bar{B}\in\mathcal{P}_{j+1}(\Lambda)$,
$\phi=P_{D^{+}}\phi+\zeta$ we compute 
\begin{equation}
\mathbb{E}\bigg[\sum_{x\in B,e\in\mathcal{E}}(\partial_{e}P_{B^{+}}\phi+\partial_{e}\xi(x))^{2}\big|(D^{+})^{c}\bigg]=\sum_{\substack{x\in B}
,e\in\mathcal{E}}(\partial_{e}P_{D^{+}}\phi(x)+\partial_{e}\xi(x))^{2}+\delta E_{j}
\end{equation}
where $\delta E_{j}=\sum_{x\in B,e\in\mathcal{E}}\mathbb{E}_{\zeta}\left[(\partial_{e}P_{B^{+}}\zeta)^{2}\right]=O(1)$
by Proposition~\ref{prop:covest}. 

Let $\psi=P_{D^{+}}\phi+\xi$. By Lemma \ref{lem:nonlocal}, Lemma
\ref{lem:alpha_nontransinv} and Lemma \ref{lem:alpha_munu}, it remains
to show the contractivity of
\begin{equation}
\begin{aligned}
\tilde{\mathcal{L}}_{3}
	= & \sum_{\bar{B}=U}
	\Big[ E_{j+1}(B)
	+\frac{\sigma_{j+1}}{4} \!\!\!\! \sum_{\substack{x\in B},e\in\mathcal{E}} \!\!\!\!(\partial_{e}\psi(x))^{2}
	-\frac{\sigma_{j}}{4}  \Big(\!\!\!\! \sum_{\substack{x\in B,e\in\mathcal{E}} \!\!\!\!
}(\partial_{e}\psi(x))^{2}+\delta E_{j}\Big)\\
 & \qquad\qquad
 +\mathbb{E}_{\zeta}\left[K_{j}(X,\zeta)\right]+\frac{\alpha}{4}\sum_{\substack{x\in B}
,e\in\mathcal{E}}(\partial_{e}\psi(x))^{2}\Big] \;.
\end{aligned}
\end{equation}
where $\alpha$ is given before Lemma~\ref{lem:alpha_munu}.
Choose 
\begin{equation}
\begin{aligned}
\sigma_{j+1} & =\sigma_{j}-\alpha\\
E_{j+1} & =\sigma_{j}\delta E_{j}-\mathbb{E}_{\zeta}\left[K_{j}(X,\zeta)\right]
\end{aligned}
\end{equation}
then we actually have $\tilde{\mathcal{L}}_{3}=0$.
\end{proof}
By the above choice of $E_{j+1}$ we can easily see that it's the
same number for $Z_{N}^{\prime}(\xi)$ and $Z_{N}^{\prime}(0)$. Therefore
$e^{\mathcal{E}_{j}}$ is the same for $Z_{N}^{\prime}(\xi)$ and
$Z_{N}^{\prime}(0)$, for all $j$.

\section{Proof of scaling limit of the generating function\label{sec:Proof-of-scaling}}
\begin{prop}
\label{prop:mainest}
Let $L$ be sufficiently large; $A$ sufficiently
large depending on $L$; $\kappa$ sufficiently small depending on
$L,A$; $h$ sufficiently large depending on $L,A,\kappa$; and $r$
sufficiently small depending on $L,A,\kappa,h$. Then for $|z|<r$
there exists a constant $\sigma$ depending on $z$ so that the dynamic
system 
\begin{equation}
\begin{aligned}\sigma_{j+1} & =\sigma_{j}+\alpha(K_{j})\\
K_{j+1} & =\mathcal{L}K_{j}+f(\sigma_{j},K_{j})
\end{aligned}
\end{equation}
satisfies
\begin{equation}
\left|\sigma_{j}\right|\leq r2^{-j}\qquad\left\Vert K_{j}\right\Vert _{j}\leq r2^{-j}\label{eq:mainest}
\end{equation}
\end{prop}
\begin{proof}
By contractivity of $\mathcal{L}$ we apply Theorem 2.16 in \cite{brydges_lectures_2007}
(i.e. the stable manifold theorem) to obtain a smooth function $\sigma=h(K_{0})$
so that (\ref{eq:mainest}) hold. Since $K_{0}$ depends on $z$ and
$\sigma$, we solve $\sigma$ from equation $\sigma-h(K_{0}(z,\sigma))=0$,
using Lemma \ref{lem:to_startRG}. Noting that this equation holds
with $(\sigma,z)=0$, and that $K_{0}(z=0,\sigma)=0$, the derivative
of left hand side w.r.t. $\sigma$ is $1$. So by implicit function
theorem there exists a $\sigma$ depending on $z$ so that $\sigma=h(K_{0}(z,\sigma))$.
Therefore the proposition is proved.
\end{proof}
With the generating function $Z_{N}(f)$ defined in (\ref{eq:scalinglimit-2}),
we have
\begin{thm}
For any $p>d$ there exists constants $M>0$ and $z_{0}>0$ so that
for all $\Vert\tilde{f}\Vert_{L^{p}}\leq M$, and all $\left|z\right|\leq z_{0}$
there exists a constant $\epsilon$ depending on $z$ so that
\[
\lim_{N\rightarrow\infty}Z_{N}(f)=\exp\left(-\frac{1}{2}\int_{\tilde{\Lambda}}\tilde{f}(x)(-\epsilon\bar{\Delta})^{-1}\tilde{f}(x)d^{d}x\right)
\]
where $\bar{\Delta}$ is the Laplacian in continuum.\end{thm}
\begin{proof}
By (\ref{eq:generating_good}),
\begin{equation} \label{eq:generating_good-1}
Z_{N}(f)=\lim_{m\rightarrow0}e^{\frac{1}{2}\sum_{x\in\Lambda}f(x)(-\epsilon\Delta_{m})^{-1}f(x)}Z_{N}^{\prime}(\xi)\big/Z_{N}^{\prime}(0) \;.
\end{equation} 
In fact, since $\int_{\tilde{\Lambda}}\tilde{f}=0$
\begin{equation}
e^{\frac{1}{2}\sum_{x\in\Lambda}f(x)(\epsilon\Delta_{m})^{-1}f(x)}\rightarrow e^{-\frac{1}{2}\int_{\tilde{\Lambda}}\tilde{f}(x)(-\epsilon\bar{\Delta})^{-1}\tilde{f}(x)d^{d}x}
\end{equation}
as $m\rightarrow0$ followed by $N\rightarrow\infty$. 

At scale $N-1$ (we do not want to continue all the way to the last
step since it would be not clear how to define $\tilde{I}_{N-1}$
and $I_{N}$), by Proposition \ref{prop:mainest} and Lemma \ref{lem:intproperties}
\begin{equation}
\begin{aligned}
& \big| Z_{N}^{\prime}(\xi)  -e^{\mathcal{E}_{N-1}}\big|
= e^{\mathcal{E}_{N-1}}
	\big|\mathbb{E} \big[ I_{N-1}(\Lambda\backslash \hat X) K_{N-1}(X) \big]-1\big|\\
 & \leq 
 e^{\mathcal{E}_{N-1}}\Big[
 	 \sum{}_{X\neq \emptyset}
	(1+2^{-N+1})^{|\Lambda\backslash\hat{X}|_{N-1}} \cdot 2^{-N+1}\mathbb{E}G(\ddot{X},X^{+})
\big|
+\big| \mathbb E I_{N-1}^{\Lambda}-1\big|\Big]\\
 & \leq e^{\mathcal{E}_{N-1}}
 	\Big[ 2^{L^{d}}(1+2^{-N+1})^{L^{d}} \cdot 2^{-N+1}c^{L^{d}}+2^{-N+1}\Big] \;.
\end{aligned}
\end{equation}
where $X\in\mathcal P_{N-1}$.
Since the constant $e^{\mathcal{E}_{N-1}}$ is identical for $Z_{N}^{\prime}(\xi)$
and $Z_{N}^{\prime}(0)$, and $Z_{N}^{\prime}(0)$ satisfies the same
bound above, one has $Z_{N}^{\prime}(\xi)\big/Z_{N}^{\prime}(0)\rightarrow1$.
Therefore the theorem is proved.
\end{proof}

\appendix

\section{Decay of Green's functions and Poisson kernels}

The decay rates of (derivatives of) Green's functions and their derivatives are essential
in our method. 
In this section we aim to show Corollary~\ref{cor:Cdecay} on the torus.

First of all consider the Green's function of $-\Delta_{m}=-\Delta+m^{2}$
on $\mathbb{Z}^{d}$. If $d\geq3$, let $G_{m}=(-\Delta_{m})^{-1}$.
If $d=2$ let $G_{m}(x)=(-\Delta_{m})^{-1}(x)-(-\Delta_{m})^{-1}(0)$
for $m>0$ and from \cite{lawler_intersections_1991} we know that
$\lim_{m\rightarrow0}G_{m}(x)$ exists. Write $G=G_{m=0}$.

\begin{lem}
\label{lem:GreenZd_decay}
Let $\bar{G}(x)=a_{d}|x|^{2-d}$ if $d\geq3$
and $\bar{G}(x)=a_{d}\log|x|$ if $d=2$.
 Let $k=\frac{2\gamma+\log8}{\pi}$ if $d=2$ where $\gamma$
is Euler's constant and $k=0$ if $d\geq3$. 
Then there exists a constant  $a_{d}$ which only depends
on $d$, such that
as $|x|\rightarrow\infty$,
\begin{equation}
G(x)=\bar{G}(x)+k+O(|x|^{-d}) \;.
\end{equation}
Furthermore, for all $e\in\mathcal{E}$
\begin{equation}
\partial_{e}G(x)=\partial_{e}\bar{G}(x)+O(|x|^{-(d+1)})\label{eq:sharpererror}
\end{equation}
where \textup{$\partial_{e}\bar{G}(x)$ is also discrete derivative.}\end{lem}
\begin{proof}
See \cite{lawler_random_2010} Theorem 4.3.1, 4.4.4, Corollary 4.3.3,
4.4.5. The only difference here is a sharper estimate of the error
term for $\nabla G$, which is remarked after those corollaries and
thus the proof is omitted.\end{proof}
\begin{lem}
\label{lem:sum_periods}
Let $d\geq2$. For all $e\in\mathcal{E}$,
$x\in\Lambda$ where $\Lambda$ is the torus defined in subsection
\ref{sub:Definition-of-model} and $m\geq0$,
\begin{equation}
\bigg|\sum_{y\in\mathbb{Z}^{d}\backslash\{0\}}\partial_{e}G_{m}(x+L^{N}y)\bigg|\leq c_{d}L^{-(d-1)N}
\end{equation}
where $c_{d}$ only depends on $d$.\end{lem}
\begin{rem}
Note that the left hand side is not absolutely summable uniformly
in $m>0$.
\end{rem}

\begin{proof}
It's enough to show the proof for $m=0$. Denote $D_{\mu}$ to be
the smooth derivative. Without loss of generality assume $e=e_{1}$.
The term $O(|x|^{-(d+1)})$ in (\ref{eq:sharpererror}) is summable:
\begin{equation}
\bigg|\sum_{y\in\mathbb{Z}^{d}\backslash\{0\}}O(|x+L^{N}y|^{-(d+1)})\bigg|=O(L^{-(d+1)N}) \;.
\end{equation}
Up to this term, $\partial_{e_{1}}G(x+L^{N}y)$ is equal to 
\begin{equation}
\begin{aligned} & \bar{G}(x+e_{1}+yL^{N})-\bar{G}(x+yL^{N})\\
= & \Big(\bar{G}(yL^{N})+(x+e_{1})\cdot D\bar{G}(yL^{N})+\frac{1}{2}(x+e_{1})^{2}\cdot D^{2}\bar{G}(yL^{N})+Err\Big)\\
 & -\Big(\bar{G}(yL^{N})+x\cdot D\bar{G}(yL^{N})+\frac{1}{2}x^{2}\cdot D^{2}\bar{G}(yL^{N})+Err\Big)\\
= & D_{1}\bar{G}(yL^{N})+(x\cdot DD_{1}\bar{G}(yL^{N})+\frac{1}{2}D_{1}^{2}\bar{G}(yL^{N}))+Err
\end{aligned}
\end{equation}
where 
\[
Err = O(L^{2N}\sup\left|D^{3}\bar{G}\right|)
\]
which comes from Taylor remainder theorem. It's a straightforward
calculation to see that the summation over $y\neq0$ of the first
three terms is zero due to cancellations. The summation over $y\neq0$
of the  term $Err$ gives $O(L^{-(d-1)N})$.
\end{proof}

\begin{cor}
\label{cor:Cdecay}
Let $d\geq2$ and $C_{m}$ be the Green's function
of $-\Delta+m^{2}$ on the torus $\Lambda$. For all $e\in\mathcal{E}$,
$x\in\Lambda$ and $m\geq0$,
\begin{equation}
\left|\partial_{e}C_{m}(x)\right|\leq c_{d} \Vert x\Vert^{-(d-1)}_\Lambda
\end{equation}
where the constant $c_{d}$ only depends on $d$,
and $\Vert x\Vert_\Lambda = d(0,x)$ with $d(-,-)$ being the 
distance function on $\Lambda$ defined in Subsection
\ref{sub:Conventions-about-notations}.
\end{cor}

Note that the distant function $d(-,-)$ is not to be confused with the dimension $d$.

\begin{proof}
The statement is immediately shown by 
\begin{equation}
\partial_{e}C_{m}(x)=\sum_{y\in\mathbb{Z}^{d}}\partial_{e}G_{m}(x+L^{N}y)
\end{equation}
and Lemma \ref{lem:GreenZd_decay}, \ref{lem:sum_periods}.
\end{proof}

%
%

\section{Estimates\label{sec:Estimates}}

In Section \ref{sec:Norms} we defined norms for functions of the
fields. In the Appendix we give estimates in terms of these norms
of some functions of interest. 
\begin{lem}
\label{lem:estimateV} 
There exists a constant $c>0$ so that if $\sigma/\kappa<c$
and $h^{2}\sigma<c$, then for every $B\in\mathcal{B}_{j}$, $j<N-1$, one has
\begin{equation} \label{eq:Iest}
\Vert e^{-\frac{\sigma}{2}\sum_{x\in B,e}(\partial_{e}P_{B^{+}}\phi(x)+\partial_{e}\xi(x))^{2}}\Vert_{T_{\phi}(\Phi_{j}(B))}\leq2e^{\frac{\kappa}{4}\sum_{B}(\partial P_{B^{+}}\phi)^{2}} \;,
\end{equation}
\begin{equation} \label{eq:tildeIest}
\Vert e^{-\frac{\sigma}{2}\sum_{x\in B,e}(\partial_{e}P_{(\bar{B})^{+}}\phi(x)+\partial_{e}\xi(x))^{2}}\Vert_{T_{\phi}(\Phi_{j}(\dot{\bar{B}},\bar{B}^{\,+}))}\leq2e^{\frac{\kappa}{4}\sum_{B}(\partial P_{(\bar{B})^{+}}\phi)^{2}} \;,
\end{equation}
\begin{equation} \label{eq:I-1est}
\Vert e^{-\frac{\sigma}{2}\sum_{x\in B,e}(\partial_{e}P_{B^{+}}\phi(x)+\partial_{e}\xi(x))^{2}}-1\Vert_{T_{\phi}(\Phi_{j}(B))}\leq4c^{-1}h^{2}|\sigma|e^{\frac{\kappa}{4}\sum_{B}(\partial P_{B^{+}}\phi)^{2}} \;,
\end{equation}
\begin{equation} \label{eq:tildeI-1est}
\Vert e^{-\frac{\sigma}{2}\sum_{x\in B,e}(\partial_{e}P_{(\bar{B})^{+}}\phi(x)+\partial_{e}\xi(x))^{2}}-1\Vert_{T_{\phi}(\Phi_{j}(\dot{\bar{B}},\bar{B}^{\,+}))}\leq4c^{-1}h^{2}e^{\frac{\kappa}{4}\sum_{B}(\partial P_{(\bar{B})^{+}}\phi)^{2}} \;.
\end{equation}
\end{lem}
\begin{rem}
Note that the prefactors on the exponentials on the right hand sides
are always $\frac{\kappa}{4}$, whereas the prefactor in our regulator
defined in Section \ref{sec:Norms} is $\frac{\kappa}{2}$.\end{rem}
\begin{proof}
To prove (\ref{eq:Iest}), let
\[
V=-\frac{1}{2}\sum_{x\in B,e}(\partial_{e}P_{B^{+}}\phi(x)+\partial_{e}\xi(x))^{2}
\]
and let $\Vert(f,\lambda\xi)^{\times n}\Vert_{\Phi_{j}(B)}\leq1$.
By $|\partial\xi|^{2}\leq h^{2}L^{-dN}$ it is straightforward to check
that if $\sigma/\kappa$ is sufficiently small, for $n=0,1,2$, 
\begin{equation}
\left|(\sigma V)^{(n)}(\phi,\xi;(f,\lambda\xi)^{\times n})\right|\leq\frac{\kappa}{2^{n+4}}\sum_{x\in B,e}(\partial_{e}P_{B^{+}}\phi(x))^{2}+2\sigma h^{2}
\end{equation}
and for $n\geq3$, $V^{(n)}=0$. Therefore for $n=0,\dots,4$,
\[
\begin{aligned}
\frac{1}{n!}\big|\big(e^{\sigma V}\big)^{(n)} & (\phi,\xi;(f,\lambda\xi)^{\times n})\big|
\leq e^{|\sigma V|} e^{|\sigma V^{(1)}|+|\sigma V^{(2)}|}\\
 & \leq e^{\frac{\kappa}{4}\sum_{x\in B,e}(\partial_{e}P_{B^{+}}\phi(x))^{2}+8\sigma h^{2}}\\
 & \leq2\, e^{\frac{\kappa}{4}\sum_{x\in B,e}(\partial_{e}P_{B^{+}}\phi(x))^{2}}
\end{aligned}
\]
if $h^{2}\sigma$ is sufficiently small, where we bounded the polynomials
in $(\sigma V)^{(n)}$ by $e^{|\sigma V^{(1)}|+|\sigma V^{(2)}|}$.
So (\ref{eq:Iest}) is proved. (\ref{eq:tildeIest}) is proved in
the same way.

To prove (\ref{eq:I-1est}), note that similarly as above one can
show that $e^{\sigma V}$ is analytic in $\sigma$, so 
\[
\begin{aligned}
\Vert e^{\sigma V}-1\Vert_{T_{\phi}(\Phi_{j}(B))}
&=\Big\Vert
	\frac{1}{2\pi i}\int_{|z|=ch^{-2}}\frac{\sigma e^{zV}}{z(z-\sigma)}dz
\Big\Vert_{T_{\phi}(\Phi_{j}(B))} \\
&\leq
4c^{-1}h^{2}|\sigma|e^{\frac{\kappa}{4}\sum_{B}(\partial P_{B^{+}}\phi)^{2}}
\end{aligned}
\]
and (\ref{eq:tildeI-1est}) is proved in the same way.
\end{proof}
Another example is the estimate of the initial interaction. At step
$j=0$ a block $B$ is a single lattice point $x$. Define 
\[
\tilde{W}(\{x\},\phi,u)
	=\frac{1}{2}\sum_{e\in\mathcal{E}}
	\cos(u\,\partial_{e}\phi(x))  \;.
\]
We also write $W(\{x\},\phi)=\tilde{W}(\{x\},\phi,\sqrt{\beta(1+\sigma)})$.
Recall that $\Vert-\Vert_{0}$ is the $\Vert-\Vert_{j}$ norm defined in \ref{eq:def_normKXj}  with $j=0$.

\begin{lem}
\label{lem:estimateK0}
If $\kappa\geq h^{-1}$, then
1) $\tilde{W}(\{x\},\phi,u)$ satisfies
\begin{equation}\label{eq:Westi}
\sum_{n=0}^{3}\frac{1}{n!}\sup_{|\partial f(x)|\leq h}
\partial_{t_{1}\dots t_{n}}\big|_{t_{i}=0}
\Big|
\partial_{u}^{m}W(\{x\},\phi+\sum_{i=1}^{n}t_{i}f_{i})
	\Big|
\leq C_{h,u} e^{\frac{\kappa}{2}\sum_{e\in\mathcal{E}}(\partial_{e}\phi(x))^{2}}
\end{equation}
for $m=0,1,2,\dots$, where $C_{h,u}=d(2h)^{m}e^{hu}$.

2) Let $\Vert-\Vert_{00}$ be the $\Vert-\Vert_{0}$ norm with $G=1$.
For $|z|$ sufficiently small, 
\begin{equation} \label{eq:eWesti} 
\Vert e^{zW(\{x\})}\Vert_{00}\leq 2 \;.
\end{equation}
\end{lem}
\begin{proof}
1) The case $m=0$ holds even without $e^{\frac{\kappa}{2}\sum_{e\in\mathcal{E}}(\partial_{e}\phi(x))^{2}}$
by straightforward computations and thus is omitted. For $m>0$, 
\[
\partial_{u}^{m}W
	=\pm\frac{1}{2}\sum_{e\in\mathcal{E}}\substack{\sin\\
\cos}
(u\partial_{e}\phi(x))\left(\partial_{e}\phi(x)\right)^{m} \;.
\]
We then have the bound
\[
\sum_{n=0}^{4}\frac{1}{n!}\sup_{\left|\partial f(x)\right|\leq h}\partial_{t_{1}\dots t_{n}}\big|_{t_{i}=0}\Big(\partial_{e}(\phi(x)+\sum_{i=1}^{n}t_{i}f_{i}(x))\Big)^{m}
\leq
(2h)^{m}e^{\frac{\kappa}{2}\sum_{e\in\mathcal{E}}(\partial_{e}\phi(x))^{2}} .
\]
The bound for $\partial_{u}^{m}W$ follows by product rule of differentiations
and the case $m=0$.

2) For $|z|$ sufficiently small, 
\[
\Vert e^{zW(B)}\Vert_{00}\leq\sum_{n=0}^{\infty}\frac{|z|^{n}}{n!}\left\Vert W(B)\right\Vert _{00}^{n}\leq\exp\left(4d|z|e^{h}\right)\leq 2 \;.
\]
This is precisely the claimed bound.
\end{proof}

\begin{lem}
\label{lem:to_startRG}
Let $K_0$ be the function defined in Proposition~\ref{prop:firstprop}.
Given $r>0$, if $|z|$ and $|\sigma|$ are
sufficiently small, then $\Vert K_{0}\Vert_0 <r$. Furthuremore, $K_{0}$
is smooth in $z$ and $\sigma$.\end{lem}
\begin{proof}
As in the proof of (\ref{eq:eWesti}), one has
\[
\Vert e^{zW(\{x\})}-1\Vert_{00}\leq\exp\left(4d|z|e^{h}\right)-1\leq c|z|
\]
for some constant $c$. 
Write $V_{0}(\{x\})=-\frac{1}{2}\sum_{e}(\partial_{e}\phi(x)+\partial_{e}\xi(x))^{2}$.
By Lemma \ref{lem:estimateV},
\[
\Vert(e^{zW(\{x\})}-1)e^{\sigma V_{0}(\{x\})}\Vert_{0}\leq2c|z|  \;,
\]
therefore
\[
\left\Vert K_{0}\right\Vert _{0}=\sup_{X\in\mathcal{P}_{0,c}}\left\Vert K_{0}(X)\right\Vert _{0}A^{|X|_{0}}\leq\sup_{X\in\mathcal{P}_{0,c}}(2c|z|A)^{|X|_{0}}<r \;.
\]
The derivative of $\prod_{x\in X}(e^{zW(\{x\})}-1)$ w.r.t.
$\sigma$ is equal to
\[
\sum_{x\in X}zW^{\prime}(\{x\})\frac{1}{2\sqrt{1+\sigma}}
\prod_{y\in X\backslash \{x\}}(e^{zW(\{y\})}-1) \;,
\]
therefore its $\Vert-\Vert_{0}$ norm is bounded by $c^{\prime}A|z|$
for some constant $c^{\prime}$. The derivative of $e^{\sigma V_{0}(\{x\})}$
and higher derivatives can be bounded similarly. The derivative of
$\prod_{x\in X}(e^{zW(\{x\})}-1)$ w.r.t. $z$ is equal
to
\[
\sum_{x\in X}W(\{x\})\prod_{y \in X\backslash \{x\}}(e^{zW(\{y\})}-1)
\]
which can be bounded in the same way.
\end{proof}
\bibliographystyle{alpha}
\bibliography{condexp}

\def\polhk#1{\setbox0=\hbox{#1}{\ooalign{\hidewidth
  \lower1.5ex\hbox{`}\hidewidth\crcr\unhbox0}}}
\begin{thebibliography}{BBS14b}

\bibitem[ACG13]{abdesselam}
A.~Abdesselam, A.~Chandra, and G.~Guadagni.
\newblock Rigorous quantum field theory functional integrals over the p-adics
  i: Anomalous dimensions.
\newblock {\em {arXiv:1302.5971}}, 2013.

\bibitem[AKM13]{adams_2013}
Stefan Adams, Roman Koteck{\'y}, and Stefan M{\"u}ller.
\newblock Finite range decomposition for families of gradient {G}aussian
  measures.
\newblock {\em J. Funct. Anal.}, 264(1):169--206, 2013.

\bibitem[Ba{\l}83]{balaban_1983}
Tadeusz Ba{\l}aban.
\newblock Ultraviolet stability in field theory. {T}he {$\varphi _{3}^{4}$}
  model.
\newblock In {\em Scaling and self-similarity in physics ({B}ures-sur-{Y}vette,
  1981/1982)}, volume~7 of {\em Progr. Phys.}, pages 297--319. Birkh\"auser
  Boston, Boston, MA, 1983.

\bibitem[Bau13]{bauerschmidt_simple_2012}
Roland Bauerschmidt.
\newblock A simple method for finite range decomposition of quadratic forms and
  {G}aussian fields.
\newblock {\em Probab. Theory Related Fields}, 157(3-4):817--845, 2013.

\bibitem[BBS12]{bauerschmidt_structural}
R.~Bauerschmidt, D.C. Brydges, and G.~Slade.
\newblock Structural stability of a dynamical system near a non-hyperbolic
  fixed point.
\newblock {\em {arXiv:1211.2477}}, 2012.

\bibitem[BBS14a]{Brydges_2014-6}
R.~Bauerschmidt, D.C. Brydges, and G.~Slade.
\newblock Critical two-point function of the 4-dimensional weakly self-avoiding
  walk.
\newblock {\em {arXiv:1403.7268}}, 2014.

\bibitem[BBS14b]{Brydges_2014-3}
R.~Bauerschmidt, D.C. Brydges, and G.~Slade.
\newblock A renormalisation group method. {III}. {P}erturbative analysis.
\newblock {\em {arXiv:1403.7252}}, 2014.

\bibitem[BBS14c]{Brydges_2014-8}
R.~Bauerschmidt, D.C. Brydges, and G.~Slade.
\newblock Scaling limits and critical behaviour of the 4-dimensional
  n-component $|\varphi|^4$ spin model.
\newblock {\em {arXiv:1403.7424}}, 2014.

\bibitem[BDH95]{brydges_short_1995}
D.~Brydges, J.~Dimock, and T.~R. Hurd.
\newblock The short distance behavior of {$(\phi^4)_3$}.
\newblock {\em Comm. Math. Phys.}, 172(1):143--186, 1995.

\bibitem[BDH98]{brydges_nonGaussian_1998}
D.~Brydges, J.~Dimock, and T.~R. Hurd.
\newblock A non-{G}aussian fixed point for {$\phi^4$} in {$4-\epsilon$}
  dimensions.
\newblock {\em Comm. Math. Phys.}, 198(1):111--156, 1998.

\bibitem[BGM04]{brydges_finite_2004}
David~C. Brydges, G.~Guadagni, and P.~K. Mitter.
\newblock Finite range decomposition of {G}aussian processes.
\newblock {\em J. Statist. Phys.}, 115(1-2):415--449, 2004.

\bibitem[BIS09]{Brydges_functional_2009}
David~C. Brydges, John~Z. Imbrie, and Gordon Slade.
\newblock Functional integral representations for self-avoiding walk.
\newblock {\em Probab. Surv.}, 6:34--61, 2009.

\bibitem[BMS03]{brydges_critical_2003}
D.~C. Brydges, P.~K. Mitter, and B.~Scoppola.
\newblock Critical {$(\Phi^4)_{3,\epsilon}$}.
\newblock {\em Comm. Math. Phys.}, 240(1-2):281--327, 2003.

\bibitem[Bry09]{brydges_lectures_2007}
David~C. Brydges.
\newblock Lectures on the renormalisation group.
\newblock In {\em Statistical mechanics}, volume~16 of {\em IAS/Park City Math.
  Ser.}, pages 7--93. Amer. Math. Soc., Providence, RI, 2009.

\bibitem[BS10]{brydges_renormalization_2010}
David Brydges and Gordon Slade.
\newblock Renormalisation group analysis of weakly self-avoiding walk in
  dimensions four and higher.
\newblock In {\em Proceedings of the {I}nternational {C}ongress of
  {M}athematicians. {V}olume {IV}}, pages 2232--2257, New Delhi, 2010.
  Hindustan Book Agency.

\bibitem[BS14a]{Brydges_2014-1}
D.C. Brydges and G.~Slade.
\newblock A renormalisation group method. {I}. {G}aussian integration and
  normed algebras.
\newblock {\em {arXiv:1403.7244}}, 2014.

\bibitem[BS14b]{Brydges_2014-2}
D.C. Brydges and G.~Slade.
\newblock A renormalisation group method. {II}. {A}pproximation by local
  polynomials.
\newblock {\em {arXiv:1403.7253}}, 2014.

\bibitem[BS14c]{Brydges_2014-4}
D.C. Brydges and G.~Slade.
\newblock A renormalisation group method. {IV}. {S}tability analysis.
\newblock {\em {arXiv:1403.7255}}, 2014.

\bibitem[BS14d]{Brydges_2014-5}
D.C. Brydges and G.~Slade.
\newblock A renormalisation group method. {V}. {A} single renormalisation group
  step.
\newblock {\em {arXiv:1403.7256}}, 2014.

\bibitem[BT06]{brydges_finite_2006}
David Brydges and Anna Talarczyk.
\newblock Finite range decompositions of positive-definite functions.
\newblock {\em J. Funct. Anal.}, 236(2):682--711, 2006.

\bibitem[BY90]{brydges_grad_1990}
David Brydges and Horng-Tzer Yau.
\newblock Grad {$\phi$} perturbations of massless {G}aussian fields.
\newblock {\em Comm. Math. Phys.}, 129(2):351--392, 1990.

\bibitem[Del99]{delmotte_parabolic_1999}
Thierry Delmotte.
\newblock Parabolic {H}arnack inequality and estimates of {M}arkov chains on
  graphs.
\newblock {\em Rev. Mat. Iberoamericana}, 15(1):181--232, 1999.

\bibitem[DH00]{Dimock_SineGordon}
J.~Dimock and T.~R. Hurd.
\newblock Sine-{G}ordon revisited.
\newblock {\em Ann. Henri Poincar\'e}, 1(3):499--541, 2000.

\bibitem[Dim09]{dimock_infinite_2009}
J.~Dimock.
\newblock Infinite volume limit for the dipole gas.
\newblock {\em J. Stat. Phys.}, 135(3):393--427, 2009.

\bibitem[Dim13]{dimock_renormalization-1}
J.~Dimock.
\newblock The renormalization group according to {B}alaban, {I}. {S}mall
  fields.
\newblock {\em Rev. Math. Phys.}, 25(7):1330010, 64, 2013.

\bibitem[Fal12]{falco_kosterlitz_2012}
Pierluigi Falco.
\newblock Kosterlitz-{T}houless transition line for the two dimensional
  {C}oulomb gas.
\newblock {\em Comm. Math. Phys.}, 312(2):559--609, 2012.

\bibitem[Fal13]{falco_Critical_2013}
Pierluigi Falco.
\newblock Critical exponents of the two dimensional coulomb gas at the
  berezinskii-kosterlitz-thouless transition.
\newblock {\em {arXiv} preprint {arXiv:1311.2237}}, 2013.

\bibitem[FP78]{frohlich_correlation_1978}
J{\"u}rg Fr{\"o}hlich and Yong~Moon Park.
\newblock Correlation inequalities and the thermodynamic limit for classical
  and quantum continuous systems.
\newblock {\em Comm. Math. Phys.}, 59(3):235--266, 1978.

\bibitem[FS81a]{frohlich_kosterlitzPRL_1981}
J{\"u}rg Fr{\"o}hlich and Thomas Spencer.
\newblock Kosterlitz-{T}houless transition in the two-dimensional plane rotator
  and {C}oulomb gas.
\newblock {\em Phys. Rev. Lett.}, 46(15):1006--1009, 1981.

\bibitem[FS81b]{frohlich_kosterlitzCMP_1981}
J{\"u}rg Fr{\"o}hlich and Thomas Spencer.
\newblock The {K}osterlitz-{T}houless transition in two-dimensional abelian
  spin systems and the {C}oulomb gas.
\newblock {\em Comm. Math. Phys.}, 81(4):527--602, 1981.

\bibitem[FS81c]{frohlich_statistical_1981}
J{\"u}rg Fr{\"o}hlich and Thomas Spencer.
\newblock On the statistical mechanics of classical {C}oulomb and dipole gases.
\newblock {\em J. Statist. Phys.}, 24(4):617--701, 1981.

\bibitem[Gal85]{gallavotti_1985}
Giovanni Gallavotti.
\newblock Renormalization theory and ultraviolet stability for scalar fields
  via renormalization group methods.
\newblock {\em Rev. Modern Phys.}, 57(2):471--562, 1985.

\bibitem[Gia83]{giaquinta_multiple_1983}
M.~Giaquinta.
\newblock {\em Multiple Integrals in the Calculus of Variations and Nonlinear
  Elliptic Systems}.
\newblock Princeton University Press, 1983.

\bibitem[GK80]{gawedzki_rigorous_1980}
K.~Gaw{\polhk{e}}dzki and A.~Kupiainen.
\newblock A rigorous block spin approach to massless lattice theories.
\newblock {\em Comm. Math. Phys.}, 77(1):31--64, 1980.

\bibitem[GK83]{gawedzki_block_1983}
K.~Gaw{\polhk{e}}dzki and A.~Kupiainen.
\newblock Block spin renormalization group for dipole gas and {$(\nabla \varphi
  )^{4}$}.
\newblock {\em Ann. Physics}, 147(1):198--243, 1983.

\bibitem[Kum10]{kumagai_random_2010}
T.~Kumagai.
\newblock Random walks on disordered media and their scaling limits.
\newblock {\em Notes of St. Flour lectures, available at
  http://www.kurims.kyoto-u.ac.jp/~kumagai/}, 2010.

\bibitem[Law91]{lawler_intersections_1991}
Gregory~F. Lawler.
\newblock {\em Intersections of random walks}.
\newblock Probability and its Applications. Birkh\"auser Boston Inc., Boston,
  MA, 1991.

\bibitem[LL10]{lawler_random_2010}
Gregory~F. Lawler and Vlada Limic.
\newblock {\em Random walk: a modern introduction}, volume 123 of {\em
  Cambridge Studies in Advanced Mathematics}.
\newblock Cambridge University Press, Cambridge, 2010.

\end{thebibliography}

\end{document}